\documentclass[11pt,twoside]{article}
\usepackage{RR}
\usepackage{geometry}
\usepackage[english]{babel}
\usepackage[utf8]{inputenc}
\usepackage{amssymb}
\usepackage{mathtools}
\usepackage{graphicx}
\usepackage[colorinlistoftodos]{todonotes}
\usepackage{hyperref}
\usepackage{url}
\usepackage{prettyref}
\usepackage{comment}
\usepackage{cancel}
\usepackage{pifont}
\usepackage{algorithm}
\usepackage[noend]{algpseudocode}
\usepackage{manfnt}
\usepackage{microtype}
\hypersetup{
    colorlinks = true,
    urlbordercolor = {magenta},
    citecolor = {blue},
}

\RRdate{January 2026}

\RRetitle{
\bf Structural Methods for handling mode changes in multimode DAE systems
} 

\RRtitle{
\bf Analyse Structurelle des changements de mode \\ dans les DAE multimodes 
} 

\titlehead{Structural Methods for handling  mode changes in multimode DAE systems}

\RRauthor{\rm Albert Benveniste\thanks{All authors are with Inria center at Rennes University, Campus de Beaulieu, 35042 Rennes cedex, France; email: surname.name@inria.fr}, Benoit Caillaud, Yahao Chen,
Khalil Ghorbal, 
Mathias Malandain
}
\authorhead{A. Benveniste, B. Caillaud, Y. Chen, K. Ghorbal, M. Malandain}

\vfill
\RRnote{}

\RRabstract{
Hybrid systems are an important concept in Cyber-Physical Systems modeling, for which multiphysics modeling from first principles and the reuse of models from libraries are key. To achieve this, DAEs must be used to specify the dynamics in each discrete state (or \emph{mode} in our context). This led to the development of DAE-based equational languages supporting multiple modes, of which Modelica is a popular standard. Mode switching can be time- or state-based. Impulsive behaviors can occur at mode changes. While mode changes are well understood in particular physics (e.g., contact mechanics), this is not the case in physics-agnostic paradigms such as Modelica. This situation causes difficulties for the compilation of programs, often requiring users to manually ``smooth out'' mode changes. In this paper, we propose a novel approach for the hot restart at mode changes in such paradigms. We propose a mathematical meaning for hot restarts (such a mathematical meaning does not exist in general), as well as a combined \emph{structural-and-impulse analysis} for mode changes, generating the hot restart even in the presence of impulses. Our algorithm detects at compile time if the mode change is insufficiently specified, in which case it returns diagnostics information to the user.

\textbf{About this revision:}
In a first version of this report, we restricted ourselves to the handling of mode changes separating \emph{long} modes (lasting for a positive duration). In this revised version, \textbf{we extend our approach to handling cascades of \emph{transient} modes} (of zero duration) such as, e.g., occurring in complex impacts in multibody mechanics.
}

\RRresume{
La mod\'elisation des syst\`emes cyber-physiques repose sur une mod\'elisation \`a partir des principes de la physique, et en r\'eutilisant au maximum des mod\`eles pr\'ed\'efinis issus d'une biblioth\`eque. Cela exige le recours aux Equations Diff\'erentielles Alg\'ebriques (DAE) admettant plusieurs modes (une DAE commut\'ee, ou une DAE hybride). Le standard de mod\'elisation est le langage Modelica. Les changements de mode peuvent \^etre d\'eclench\'es de mani\`ere externe, ou par des conditions portant sur les \'etats. Ces changements de mode sont connus et trait\'es \`a l'int\'erieur de physiques particuli\`eres (m\'ecanique avec contacts). Il en va autrement dans un cadre multi-physique g\'en\'eral, qui est, pourtant, celui de Modelica et d'autres langages de modélisation multi-physique. Dans ce papier, nous proposons une approche nouvelle pour le red\'emarrage \`a chaud suite \`a un changement de mode. Noter qu'il n'existe pas de d\'efinition math\'ematique de ce qu'est une solution dans notre cadre g\'en\'eral. Notre m\'ethode utilise une analyse structurelle doublée d'un calcul symbolique des comportements impulsifs. Notre m\'ethode s'applique lors de la phase de compilation et permet de détecter, avant toute simulation, si le mod\`ele soumis est \'eventuellement insuffisamment sp\'ecifi\'e.

\textbf{Concernant cette version:}
Dans une premi\`ere  version de ce rapport, nous traitions seulement de changements de mode s\'eparant deux modes \emph{longs}  (de dur\'ee strictement positive). Dane cette version-ci, \textbf{nous \'etendons notre approche aux cascades of modes \emph{transitoires}} (de dur\'ee nulle) telle qu'on en rencontre en m\'ecanique multi-corps.
}

\RRmotcle{analyse structurelle, \'equations algébro-diff\'erentielles
  (DAE), syst\`emes multi-mode, changement de mode \`a chaud}
\RRkeyword{structural analysis, differential-algebraic equations
  (DAE), multi-mode systems, switched system}
	\RRprojet{Hycomes}  
\RCRennes
\begin{document}
\RRNo{9603}
\hypersetup{pageanchor=false}


\newcommand{\todottj}[1]{\todo{TTJ: #1}}
\newcommand{\boldit}[1]{{\bf \texttt{#1}}}
\newcommand{\St}[1]{{\bf St}$_{#1}$}
\newcommand{\shadow}[1]{\mathsf{Sh\!}\left(#1\right)}
\newcommand{\Imporders}{\mathfrak{I}}
\newcommand{\imporder}{\mathfrak{i}}
\newcommand{\dummy}[1]{{\sf dum}(#1)}

\newcommand{\smagorder}[1]{\|#1\|}
\newcommand{\magorder}[1]{\left\|\right.\!#1\!\left.\right\|}
\newcommand{\ttail}[1]{#1^{\sf tail}}
\newcommand{\hhead}[1]{#1^{\sf head}}
\newcommand{\disable}[1]{#1^{\sf dis}}
\newcommand{\enable}[1]{#1^{\sf ena}}
\newcommand{\polyn}{f}
\newcommand{\qpolyn}{\mathbf{q}}
\newcommand{\monom}{\mathbf{m}}
\newcommand{\Nonlin}[1]{\mathsf{Nonlin}\left(#1\right)}
\newcommand{\Polyns}[1]{\mathsf{Pol}\left(#1\right)}
\newcommand{\Monoms}[1]{\mathsf{Mon}\left(#1\right)}
\newcommand{\equivsim}[1]{\widehat{#1}}
\newcommand{\simquotient}[1]{\equivsim{#1}}
\newcommand{\sizesim}[1]{{\sem{#1}}}
\newcommand{\simclosure}[1]{\widetilde{#1}}
\newcommand{\simextension}[1]{\widehat{#1}}





\newcommand{\Cons}[1]{\mathbf{Cons}\ssem{#1}}
\newcommand{\Argmax}[1]{A_{12}^{#1}}
\newcommand{\conflict}[1]{#1^{\sf c}}
\newcommand{\lbound}[1]{\delta(#1)}
\newcommand{\restart}[1]{#1^{\sf init}}
\newcommand{\Restart}[1]{\Lambda(#1)}
\newcommand{\RRestart}[1]{\Lambda(#1)}
\newcommand{\restartsys}[2]{\mathbf{R}_{#2,#1}}
\newcommand{\preorders}[1]{{\sf PO}^{#1}}
\newcommand{\bexp}{b}
\newcommand{\muLabels}{\Lambda}
\newcommand{\mulabel}{\lambda}
\newcommand{\uun}{\mathbf{1}}
\newcommand{\matching}{\mathcal{M}}
\newcommand{\starf}{f_\star}
\newcommand{\starx}{x_\star}
\newcommand{\starr}{r_\star}
\newcommand{\muvar}{\sigma}
\newcommand{\muval}{\nu}
\newcommand{\ssequiv}{\sequiv}
\newcommand{\wwequiv}{\wequiv}
\newcommand{\flow}{\varphi}
\newcommand{\monomial}{\mathbf{m}}
\newcommand{\Monomials}{\mathcal{M}}
\newcommand{\False}{\operatorname{False}}
\newcommand{\True}{\operatorname{True}}
\newcommand{\tick}{\operatorname{Tick}}
\newcommand{\soleqs}{\operatorname{Solve}}
\newcommand{\sigmamethod}{$\Sigma$-method}
\newcommand{\structureimpulse}{structure{\&}impulse}

\newcommand{\flip}[1]{\scalebox{-1}[1]{\rotatebox[origin = c]{180}{#1}}}
\newcommand{\tnode}{t}
\newcommand{\birth}{\beta}
\newcommand{\bbI}{\mathbb{I}}
\newcommand{\coPi}{\flip{$\Pi$}}
\newcommand{\sspan}{{\sf span}}
\newcommand{\ttuple}[1]{\Pi\left(#1\right)}
\newcommand{\stuple}[1]{\Pi(#1)}
\newcommand{\supp}[1]{{\sf supp}_#1}
\newcommand{\RSA}[1]{\Phi_{\sf #1}}
\newcommand{\RSAMP}[1]{\Psi_{\sf #1}}
\newcommand{\store}[1]{\mathsf{store}_{\sf #1}}
\newcommand{\ddepends}{\hookrightarrow}

\newcommand{\Lin}[1]{\cL_{\!#1\!}}
\newcommand{\LinM}[2]{\cL_{\!#1{\downarrow}{#2}}}
\newcommand{\NonLin}[1]{{\sf NonLin}_{\!#1\!}}
\newcommand{\rescaling}[1]{\mu_{#1}}
\newcommand{\rescaled}[1]{#1^{\downarrow}}
\newcommand{\Parts}[1]{\mathcal{P}(#1)}
\newcommand{\sInterfaces}{\mathbb{I}}
\newcommand{\sCointerfaces}{\mathbb{K}}
\newcommand{\sJinterfaces}[1]{\mathbb{J}[#1]}
\newcommand{\baseFunctions}{\mathbb{F}}
\newcommand{\allFunctionsECC}{\mathbb{F}}
\newcommand{\allFunctions}{\pprimeppostset{}{}{\baseFunctions}}
\newcommand{\allSystems}{\mathbb{S}}
\newcommand{\basearchis}{\mathbb{S}}
\newcommand{\sSystems}[1]{\mathbb{S}[#1]}
\newcommand{\mutex}[1]{#1_{\sf mu}}
\newcommand{\mirror}[1]{\neg{#1}}
\newcommand{\binomial}[2]{\left(\!\!\!\bea{c}#1 \\ #2\eea\!\!\!\right)}
\newcommand{\proj}[2]{\Pi_{#1}(#2)}
\newcommand{\bigproj}[2]{\Pi_{#1}\left(#2\right)}
\newcommand{\invproj}[2]{\Pi^{-1}_{#1}(#2)}
\newcommand{\amax}[2]{\alpha^{#1}_{#2}}
\newcommand{\vinterface}[1]{\prog{interface\;of}\;{#1}}
\newcommand{\vcointerface}[1]{\prog{cointerface\;of}\;{#1}}
\newcommand{\vhide}[2]{\prog{hide}\;{#1}\;\prog{in}\;{#2}\,}
\newcommand{\vsmallproj}[2]{\prog{project}\;{#2}\;\prog{to}\;{#1}\,}
\newcommand{\vrestrict}[2]{\prog{restrict}\;{#2}\;\prog{to}\;{#1}\,}
\newcommand{\send}[3]{\prog{send}\;{#1}\;\prog{from}\;{#2}\;\prog{to}\;{#3}\,}
\newcommand{\vwiden}[3]{\prog{widen}\;{#2}\;\prog{to}\;{#1}\;\prog{using}\;{#3}\,}
\newcommand{\widen}[3]{{#2}^{\uparrow{#1}}_{#3}}
\newcommand{\extend}[3]{{#2}^{\uparrow{#1}}_{#3}}
\newcommand{\yfun}[1]{{\sf y}_{#1}}
\newcommand{\gfun}[1]{{\sf g}_{#1}}
\newcommand{\kfun}[1]{{\sf k}_{#1}}
\newcommand{\mfun}[1]{{\sf m}_{#1}}
\newcommand{\llength}[1]{\ell_{\!#1}}
\newcommand{\bcontrol}[1]{\beta_{\!#1}}
\newcommand{\aactive}[1]{{\rm Act}_{#1}}
\newcommand{\vextend}[3]{\prog{extend}~{#2}\;\prog{to}\;{#1}\;\prog{using}\;{#3}\,}
\newcommand{\vminextend}[2]{\prog{min-extend}~{#2}\;\prog{to}\;{#1}\,}
\newcommand{\xvminextend}[3]{\prog{min-extend}~{#2}\;\prog{to}\;{#1}\;\prog{using}\;{#3}\,}
\newcommand{\hhide}[2]{#2_{\downarrow#1}}
\newcommand{\rename}[3]{#3[#2\!/#1]}
\newcommand{\vrename}[3]{\prog{rename}\;{#1}\;\prog{to}\;{#2}\;\prog{in}\;{#3}\,}
\newcommand{\phide}[1]{\cI^{\sf p}_{{\downarrow}#1}}
\newcommand{\dhide}[1]{\cI^{\sf d}_{{\downarrow}#1}}
\newcommand{\smallproj}[2]{\mathbf{proj}_{#1}(#2)}
\newcommand{\coproj}[2]{\coPi_{#1\!}(#2)}
\newcommand{\wcoproj}[2]{\coPi_{#1}\!\left[#2\right]}
\newcommand{\smallfib}[2]{\mathbf{fib}_{#1}(#2)}
\newcommand{\margin}[2]{#2^{\downarrow#1}}
\newcommand{\cond}[2]{#2^{\uparrow#1}}
\newcommand{\euler}[2]{\mathcal{E}_{#1#2}}
\newcommand{\neuler}[3]{\mathcal{E}^{#1}_{#2#3}}
\newcommand{\seuler}[2]{S^{\mathcal{E}}_{#1#2}}
\newcommand{\solution}[1]{\mu(#1)}
\newcommand{\Solutions}[1]{{\sf sol}(#1)}
\newcommand{\Allsolutions}{\mathbb{S}{\sf ol}}
\newcommand{\minsolution}[1]{\mu^\star(#1)}
\newcommand{\minextension}[2]{\mu^\star(#1|\,#2)}
\newcommand{\aaltpath}[3]{#1\,\block^{#3}#2}
\newcommand{\altpath}[3]{\block^{#3}_{#1#2}}
\newcommand{\altpathh}[4]{\block^{#4#3}_{#1#2}}
\newcommand{\projaltpath}[4]{#1 \ra^{\!*{#4}}_{#3} #2}
\newcommand{\fmessage}[2]{m^{\sf f}_{#1{\ra}#2}}
\newcommand{\bmessage}[2]{m^{\sf b}_{#1{\ra}#2}}
\newcommand{\Offsets}{\mathbb{M}}
\newcommand{\Rescalingdiag}{\mathbb{D}}
\newcommand{\rref}[2][]{\prettyref{#2}}

\newenvironment{proof}{\paragraph{\it Proof}}{\eproof \\ }

\def\deff{\,\overset{\textup{def}}{=}\,}
\def\Jacobian{\mathbf{J}}
\def\ssuccessful{complete}
\def\J{\,\mathcal{J}\,}
\def\M{\,\mathcal{M}\,}
\def\bfm{\mathbf{m}}
\def\bfR{\mathbf{R}}
\def\bfc{\mathbf{c}}
\def\bfe{\mathbf{e}}
\def\bfd{\mathbf{d}}
\def\bfi{\mathbf{i}}
\def\bfD{\mathbf{D}}
\def\bfw{\mathbf{w}}
\def\bff{\mathbf{f}}
\def\bfF{\mathbf{F}}
\def\bfG{\mathbf{G}}
\def\bfB{\mathbf{B}}
\def\bfJ{\mathbf{J}}
\def\bfC{\mathbf{C}}
\def\bfL{\mathbf{L}}
\def\bfB{\mathbf{B}}
\def\bfH{\mathbf{H}}
\def\bfK{\mathbf{K}}
\def\bfW{\mathbf{W}}
\def\bfV{\mathbf{V}}
\def\bbV{\mathbb{V}}
\def\bbE{\mathbb{E}}
\def\bM{\mathbb{M}}
\def\Ld{{\mathop{Ld}}}
\def\path{\pi}
\def\bbr{{\mathbb R}}
\def\bbn{{\mathbb N}}
\def\bbz{{\mathbb Z}}
\def\reals{\bbr}
\def\rationals{{\mathbb Q}}
\def\bbc{{\mathbb C}}
\def\orbit{{\mathcal{O}(\vec{x}_\iota)}}
\def\Iorbit{I(\orbit)}
\def\borbit{{\bar{\mathcal{O}}(\vec4{x}_\iota)}}
\def\deg{{\operatorname{deg}}}
\def\argmax{{\operatorname{argmax}}}
\def\ker{{\operatorname{ker}}}
\def\sfd{{\operatorname{SF}}}
\def\vars{{\operatorname{vars}}}
\def\V{V_\bbr}
\def\diffradI{\negthinspace\sqrt[\leftroot{-3}\uproot{3}\scriptscriptstyle \lie]{I}}
\def\DM{{\rm DM}}
\def\SD{{\rm SD}}

\def\vbeta{{\vec{\beta}}}
\def\valpha{{\vec{\alpha}}}
\def\vgamma{{\vec{\gamma}}}
\def\ti{{t_\iota}}
\def\sriram{Sankaranarayanan}
\def\IsamDAE{IsamDAE}
\def\straight{k}
\def\dualscore{K}
\def\score{J}
\def\interfscore{\mathsf{J}}
\def\optscore{\mathsf{J}}


\def\sfT{{\sf T}}
\newcommand{\interior}[1]{{#1}^{\circ}}
\newcommand{\ctree}[1]{\sfT_{\!#1}}
\newcommand{\inctree}[1]{{\sf T}^0_{\!#1}}
\newcommand{\eqqbf}[1]{\bf #1}
\newcommand{\prog}[1]{{\small\texttt{\textbf{#1}}}}
\newcommand{\imp}[1]{\widehat{#1}}
\newcommand{\compl}[1]{#1^{\sf i}}
\newcommand{\impuls}[1]{#1^{\sf i}}
\newcommand{\ancestor}[1]{#1^{\uparrow}}
\newcommand{\successors}[1]{#1^{\Downarrow}}
\newcommand{\ancestors}[1]{#1^{\Uparrow}}
\newcommand{\prev}[1]{#1^{\!<}}
\newcommand{\facts}[1]{\Phi_{\!#1}}
\newcommand{\nonimpuls}[1]{#1}
\newcommand{\wemph}[1]{\mbox{{\color{white}$#1$}}}
\newcommand{\pmatch}[1]{{\color{blue}#1}}
\newcommand{\prevmode}[1]{{\color{gray}#1}}
\newcommand{\remph}[1]{{\color{red}#1}}
\newcommand{\gremph}[1]{{\color{green}#1}}
\newcommand{\memph}[1]{{\color{magenta}#1}}
\newcommand{\gemph}[1]{\mbox{{\color{gray}$#1$}}}
\newcommand{\bemph}[1]{\mbox{{\color{blue}$#1$}}}
\newcommand{\blemph}[1]{\mbox{{\color{black}$#1$}}}
\newcommand{\llong}{long}
\newcommand{\transient}{transient}

\newtheorem{nochange}{\bemph{\textbf{Reused}}}
\newtheorem{exec}{Constructive Semantics}
\newtheorem{principle}{Principle}
\newtheorem{conjecture}{Conjecture}
\newtheorem{definition}{Definition}
\newtheorem{proposition}{Proposition}
\newtheorem{procedure}{Procedure}
\newtheorem{property}{Property}
\newtheorem{assumption}{Assumption}
\newtheorem{corollary}{Corollary}
\newtheorem{protocol}{Protocol}
\newtheorem{strategy}{Strategy}
\newtheorem{ccomment}{Comment}
\newtheorem{ffinding}{Finding}
\newtheorem{rrule}{Rule}
\newtheorem{require}{Requirement}
\newtheorem{convention}{\textbf{Convention}}
\newtheorem{notation}{Notations}
\newtheorem{example}{Example}
\newtheorem{remark}{Remark}
\newtheorem{pprinciple}{Principle}
\newtheorem{problem}{Problem}

\newcommand{\parconsistency}[1]{{\!\!}\wemph{(}\overline{#1}\wemph{)}{\!\!}}
\newcommand{\body}[1]{{[\![}#1{]\!]}}
\newcommand{\der}[1]{#1^{\sf der}}
\newcommand{\seloptscore}[2]{\interfscore_{#1}^{#2}}
\newcommand{\argseloptscore}[2]{\Xi_{#1}^{#2}}
\newcommand{\adot}[2]{^\mathbf{#1}#2'}
\newcommand{\dder}[2]{#2^{\prime#1}}
\renewcommand{\dot}[1]{#1'}
\newcommand{\dummydot}[1]{#1^{\prime{d}}}
\renewcommand{\ddot}[1]{#1''}

\newcommand{\tension}{{\lambda}}

\newcommand{\myparagraph}[1]{\smallskip\noindent\textit{{#1.}}}
\newcommand{\bea}{\begin{array}}
\newcommand{\eea}{\end{array}}
\newcommand{\negesp}{\!\!\!}
\newcommand{\beq}{\begin{eqnarray}}
\newcommand{\eeq}{\end{eqnarray}}
\newcommand{\beqq}{\begin{eqnarray*}}
\newcommand{\eeqq}{\end{eqnarray*}}
\newcommand{\status}{\sigma}
\newcommand{\sstatus}{\mathbf{s}}
\newcommand{\ttimes}{}
\newcommand{\alphabet}{{{\sf SD}}}
\newcommand{\pproperty}{P}
\newcommand{\Ra}{\Rightarrow}
\newcommand{\ra}{\rightarrow}
\newcommand{\la}{\leftarrow}
\newcommand{\Un}{\mathbf{1}}
\newcommand{\branch}{\pi}
\newcommand{\base}[1]{\lfloor#1\rfloor}
\newcommand{\revert}[1]{\stackrel{\mbox{\tiny\raisebox{-1mm}{$\leftarrow$}}}{#1}}
\newcommand{\srevert}[1]{\bar{#1}}
\newcommand{\wh}[1]{{\color{white}#1}}
\newcommand{\change}[1]{{\sf\color{blue}#1}}

\newcommand{\vareq}[2]{#1(#2)}
\newcommand{\drop}{\mathbf{d}}
\newcommand{\condgraph}[1]{\,\frac{~~}{}\![#1]\!\frac{~~}{}\,}
\newcommand{\dircondgraph}[1]{\condgraph{#1}}
\newcommand{\condbranch}[3]{{#1}\condgraph{#2}\,{#3}}
\newcommand{\revcondbranch}[3]{{#1}\condgraph{#2}{{#3}}}
\newcommand{\dircondbranch}[4]{{#1}\dircondgraph{#2}{\puple{#3\,|\,#4}}}
\def\defined#1{#1^{\!\top}}
\def\unique#1{\bigtriangleup{#1}}
\def\free#1{\bigtriangledown{#1}}

\newcommand{\ttrans}[4]{
\underbrace{
\mbox{\fbox{~~$#1$~~}}
}_{{\rm location}}
~~
\underbrace{
\frac{\hspace{0.4cm} #3
\hspace{0.4cm}}{\hspace{0.4cm} #4
\hspace{0.4cm}}\!\!\!\raisebox{-0.1mm}{\mbox{{\tt >}}}
}_{{\rm transition}} ~
\underbrace{
\mbox{\fbox{~~$#2$~~}}
}_{{\rm new ~location}}
}

\newcommand{\trans}[3]{#1\;\frac{\hspace{2mm} #3
\hspace{3mm}}{}\!\!\!\raisebox{-0.07mm}{\mbox{{\tt >}}}\; #2}

\newcommand{\shorttrans}[3]{#1\,{\stackrel{#2}{\longrightarrow}}\,#3}
\newcommand{\rshorttrans}[3]{#1{\stackrel{#2}{\remph{\ra}}}#3}

\newcommand{\transition}[2]{
\frac{\hspace{0.4cm} #1
\hspace{0.4cm}}{\hspace{0.4cm} #2
\hspace{0.4cm}}\!\!\!\raisebox{-0.1mm}{\mbox{{\tt >}}}
}

\newcommand{\cConefun}{\cC_{\rm 1fun}}
\newcommand{\cCfun}{\cC_{\rm fun}}
\newcommand{\cdim}{\cC_{\rm dim}}
\newcommand{\nil}{\it nil}
\newcommand{\fun}{{\it fun}}
\newcommand{\cK}{\mathcal{K}}
\newcommand{\cJ}{\mathcal{J}}
\newcommand{\cR}{\mathcal{R}}
\newcommand{\mplus}{\oplus}
\newcommand{\notX}{\bar{X}}
\newcommand{\val}{\nu}
\newcommand{\bfval}[1]{\bfV_{\!#1}}
\newcommand{\vval}[1]{\val_{#1}}
\newcommand{\notvval}[1]{w_{\!\bar{#1}}}
\newcommand{\cmutex}[4]{
\neg[\priv{#1}{#2}\wedge\priv{#1}{#3}] \mbox{ and }
	\priv{#1}{#4}=\priv{#1}{#2}\vee\priv{#1}{#3}
	}
\newcommand{\Cmutex}[3]{C_\muvar^{#1,#2,#3}}

\newcommand{\mathbox}[1]{\fbox{\mbox{$#1$}}}
\newcommand{\dom}[1]{\mathbf{dom}\!\left(#1\right)}
\newcommand{\eeval}[1]{\mathbf{eval}[#1]}
\newcommand{\zzero}{\mathbf{0}}
\newcommand{\ScottVar}{S-variable}
\newcommand{\ScottVars}{\ScottVar s}
\newcommand{\Scott}{\mathscr{S}}
\newcommand{\cons}{{\it constr}}
\newcommand{\consist}{{\it coherent}}
\newcommand{\predeval}{predicate evaluation}
\newcommand{\systeq}{systems of equations}
\newcommand{\any}{-}
\newcommand{\ftop}{\top_{\!\!f}}
\newcommand{\btop}{\top_{\!b}}
\newcommand{\ctop}{\top_{\!c}}
\newcommand{\ccC}{\overline{\mathcal{C}}}
\newcommand{\numX}{\bX^{\prog{num}}}
\newcommand{\numE}{\bE^{\prog{num}}}
\newcommand{\cC}{\mathcal{C}}
\newcommand{\cD}{\mathcal{D}}
\newcommand{\sfD}{\mathsf{D}}
\newcommand{\cI}{\mathcal{I}}
\newcommand{\cV}{\mathcal{V}}
\newcommand{\cW}{\mathcal{W}}
\newcommand{\cG}{\mathcal{G}}
\newcommand{\vecc}[1]{\vec{\,#1}}
\newcommand{\veccG}{\vecc{\mathcal{G}}}
\newcommand{\un}{{1\!\!{\rm I}}}
\newcommand{\vcG}{\vec{\cG}}
\newcommand{\cGg}{\mathcal{G}_{{\bG}}}
\newcommand{\cGc}{\mathcal{G}_{{\bC}}}
\newcommand{\cZ}{\mathcal{Z}}
\newcommand{\success}{b}
\newcommand{\Id}{I}
\newcommand{\ldot}[1]{#1'}
\newcommand{\nsldot}[1]{\frac{#1-{\preset{#1}}}{\vsmall}}
\newcommand{\eqq}{e}
\newcommand{\Eqqs}{\mathbb{E}}
\newcommand{\bfEqqs}{\mathbf{E}}
\newcommand{\FEqqs}{\mathbb{F}}
\newcommand{\rouge}[1]{{\color{red}#1}}
\newcommand{\bleu}[1]{{\color{blue}#1}}
\newcommand{\unf}{1,{\rm f}}
\newcommand{\unc}{1,{\rm c}}
\newcommand{\und}{1,{\rm d}}
\newcommand{\ODE}{{{\rm ODE}}}
\newcommand{\Blocks}[1]{\mathbb{B}_{#1}}
\newcommand{\block}{\beta}
\newcommand{\noDAE}{{{\rm noDAE}}}
\newcommand{\DAE}{{{\rm DAE}}}
\newcommand{\noODE}{{{\rm noODE}}}
\newcommand{\privede}[1]{_{\downarrow{#1}}}
\newcommand{\wnotafter}{\sched}
\newcommand{\wnotbefore}{\lsched}
\newcommand{\notafter}{\bbind}
\newcommand{\notbefore}{\llind}
\newcommand{\atomic}{\dbind}
\newcommand{\shortatomic}{\!\!\sdbind\!\!}
\newcommand{\causes}{{\,<\!\!\!\!\!\lhd\,}}
\newcommand{\selfcauses}{{\,<\!\!\!\!\!\lhd\!\rhd\!\!\!\!\!>\,}}
\newcommand{\atomiceq}{\mbox{$\unlhd\!\unrhd$}}
\newcommand{\notaftereq}{\unrhd}
\newcommand{\albert}[1]{\textsf{\small\color{blue}Albert: {#1}}}
\newcommand{\benoit}[1]{\textsf{\small\color{blue}Beno\^{\i}t: {#1}}}
\newcommand{\proofof}[1]{\paragraph*{\textsc{Proof of #1}}}
\newcommand{\conn}[1]{{\it Con}\left(#1\right)}
\newcommand{\eq}[1]{{\it Eq}_{#1}}
\newcommand{\solve}[1]{\;\mathbf{out}\;#1}
\newcommand{\tuple}[1]{[#1]}
\newcommand{\puple}[1]{(#1)}
\newcommand{\const}[1]{[\![#1]\!]}
\newcommand{\Clocks}[1]{{\it Clocks}(#1)}
\newcommand{\pant}[1]{\hat{#1}}
\newcommand{\assert}{{\bf assert}\;}
\newcommand{\every}{\,{\bf every}\,}
\newcommand{\reset}{\,{\bf reset}\,}
\newcommand{\eval}{\prog{eval}}
\newcommand{\bblock}{\prog{block}}
\newcommand{\rread}{\prog{read}}
\newcommand{\sskip}{\prog{skip}}
\newcommand{\hybrid}{{\sc Hybrid}}
\newcommand{\daehybrid}{{\sc DaeHybrid}}
\newcommand{\bhybrid}{{\bf Hybrid}}
\newcommand{\shybrid}{{\sc SimpleHybrid}}
\newcommand{\lhybrid}{{\sc TryHybrid}}
\newcommand{\barF}{\bar{F}}
\newcommand{\dependent}[1]{#1^{\sf dep}}
\newcommand{\ffree}[1]{#1^{\sf free}}
\newcommand{\bbot}[1]{#1^\bot}
\newcommand{\bbbot}[1]{{\bf #1}^\bot}
\newcommand{\previous}[1]{#1^-}
\newcommand{\preset}[1]{{^{\bullet\!}{#1}}}
\newcommand{\wpostset}[1]{#1^\circ}
\newcommand{\wpreset}[1]{^{\circ\,}\!#1}
\newcommand{\postset}[1]{#1^\bullet}
\newcommand{\pplus}[1]{#1^+}
\newcommand{\mmoins}[1]{#1^-}
\newcommand{\spreset}[1]{#1^-}
\newcommand{\spostset}[1]{\dot{#1}}
\newcommand{\dont}[3]{D_{[#1,#2,#3]}}
\newcommand{\dep}[2]{#1{\,\ra\,}#2}
\newcommand{\adep}[2]{#1{-\!\!\!-}#2}
\newcommand{\ddep}[3]{#1{\ra}#2{\ra}#3}
\newcommand{\cont}[2]{C_{[#1,#2]}}
\newcommand{\controlled}[2]{#2_{#1}}
\newcommand{\contt}[3]{C_{#1,[#2,#3]}}
\newcommand{\dontt}[4]{D_{#1,[#2,#3,#4]}}
\newcommand{\ppreset}[2]{^{\bullet#1}#2}
\newcommand{\ppostset}[2]{#2^{\bullet#1}}
\newcommand{\pprime}[2]{#2^{\prime#1}}
\newcommand{\pprimeppostset}[3]{{\pprime{#1}{#3}}{\ppostset{#2}{}}}
\newcommand{\pprimeppreset}[3]{{\ppreset{#2}{}}{\pprime{#1}{#3}}}
\newcommand{\wppostset}[2]{#2^{\circ#1}}
\newcommand{\scalar}[2]{\left\langle{#1},{#2}\right\rangle}
\newcommand{\paths}[1]{\cGc^{#1}}
\newcommand{\Paths}{\Pi}
\newcommand{\kernel}{K}
\newcommand{\Kernels}{\mathbb{K}}
\newcommand{\ppath}{\pi}
\newcommand{\vpath}[1]{V\!(#1)}
\newcommand{\nst}[1]{{^{\star{\!}}#1}}
\newcommand{\anst}[1]{{^{\mathbf{a}\star\!}#1}}
\newcommand{\Nanst}[1]{{^{N\mathbf{a}\star\!}#1}}
\newcommand{\rond}{\circ}
\newcommand{\xiz}{\xi}
\newcommand{\xit}{\tau}
\newcommand{\leftlim}[1]{\stackrel{\leftarrow}{#1}}
\newcommand{\rightlim}[1]{\stackrel{\rightarrow}{#1}}
\newcommand{\wcause}{\,\mbox{-\,-\,-$\!>$}\,}
\newcommand{\rightgraph}[3]{#1\notafter #2}

\newcommand{\leftrightgraph}[3]{#1 ~  \raisebox{-0.1mm}{\mbox{{\tt
<}}}\!\!\!\frac{\hspace{0.4cm} #3
\hspace{0.4cm}}{}\!\!\!\raisebox{-0.1mm}{\mbox{{\tt >}}} ~ #2}

\newtheorem{theorem}{Theorem}
\newtheorem{lemma}{Lemma}
\newtheorem{approach}{Approach}
\newcommand{\eproof}{\hfill$\Box$}
\newcommand{\vsmall}{\varepsilon}
\newcommand{\dvsmall}{{\vsmall^2}}
\newcommand{\st}[1]{{\mathsf{st}\!\left(#1\right)}}
\newcommand{\vX}{\mbox{\tt X}}
\newcommand{\vY}{\mbox{\tt Y}}
\newcommand{\vZ}{\mbox{\tt Z}}
\newcommand{\guarded}[2]{\mathbf{if}\;{#1}\;\mathbf{then}\;{#2}}
\newcommand{\ifteguarded}[3]{\mathbf{if}\;{#1}\;\mathbf{then}\;{#2}\;\mathbf{else}\;{#3}}
\newcommand{\regularU}[2]{U_{\!#1\!}\left(#2\right)}
\newcommand{\regularV}[2]{V_{\!#1\!}\left(#2\right)}
\newcommand{\negguarded}[2]{\mathbf{if}\;{#1}\;\mathbf{else}\;{#2}}
\newcommand{\bguarded}[2]{{\bf if}\;{#1}\;{\bf then}\;{#2}}
\newcommand{\beaguarded}[2]{#1&{\vdash}&#2}
\newcommand{\po}[2]{D^{\mbox{{\footnotesize #2}}}_{#1}}
\newcommand{\eqdef}{=_{{\rm def}}}
\newcommand{\depends}{\preceq}
\newcommand{\out}[1]{[#1]_{\rm out}}
\newcommand{\height}{h}
\newcommand{\onedef}{\,{\raisebox{-0.25mm}{\mbox{\large $\mapsto\!\!\!\!\!\ra$}}}\,}
\newcommand{\bbind}{\,{\raisebox{-0.25mm}{\mbox{\large $\ra\!\!\!\!\!\ra$}}}\,}
\newcommand{\llind}{\,{\raisebox{-0.25mm}{\mbox{\large $\la\!\!\!\!\!\la$}}}\,}
\newcommand{\dbind}{\,{\raisebox{-0.25mm}{\mbox{\large $\leftrightarrow\!\!\!\!\!\leftrightarrow$}}}\,}
\newcommand{\sbbind}{\,\ra\!\!\!\!\!\ra\,}
\newcommand{\sdbind}{\,\leftrightarrow\!\!\!\!\!\leftrightarrow\,}
\newcommand{\sched}{\,{\raisebox{-0.25mm}{\mbox{\large $\ra$}}}\,}
\newcommand{\lsched}{\,{\raisebox{-0.25mm}{\mbox{\large $\la$}}}\,}
\newcommand{\dsched}{\,{\raisebox{-0.25mm}{\mbox{\large $\leftrightarrow$}}}\,}
\newcommand{\inT}{{\cvar\!\!=\!\ttt}}
\newcommand{\notinT}{{\cvar\!\!=\!\fff}}
\newcommand{\comp}{\mbox{{\large \bf $\|$}} ~}
\newcommand{\leftcomp}{\mbox{{\Large \bf $($}}}
\newcommand{\rightcomp}{\mbox{{\Large \bf $)$}}}
\newcommand{\vecbranch}[1]{({#1})^\ra}
\newcommand{\pre}[1]{{\bf pre}\left(#1\right)}
\newcommand{\last}[1]{{\bf last}\left(#1\right)}
\newcommand{\sundae}{\textsc{Sundae}}
\newcommand{\lift}[2]{{#1}^{\uparrow{#2}}}
\newcommand{\restrict}[2]{#2_{\left|#1\right.}}
\newcommand{\allVars}{\pprimeppostset{}{}{\BaseVars}}
\newcommand{\allVarsVsmall}{{\BaseVars}_\vsmall^{\prime\bullet}}
\newcommand{\cVars}{\cT}
\newcommand{\cvars}{T}
\newcommand{\clock}{\tau}
\newcommand{\clk}[1]{h_{#1}}
\newcommand{\lead}[1]{#1^\uparrow}
\newcommand{\leading}[1]{#1^{\sf lead}}
\newcommand{\state}[1]{#1^{\sf state}}
\newcommand{\other}[1]{#1^\downarrow}
\newcommand{\mode}{m}
\newcommand{\Modes}{M}
\newcommand{\guard}{\gamma}
\newcommand{\Guards}{\Gamma}
\newcommand{\cvar}{\tau}
\newcommand{\zero}[1]{{\it zero}(#1)}
\newcommand{\xcvar}{\xi}
\newcommand{\nstimes}{\mathfrak{s}}
\newcommand{\changetime}{t_*}
\newcommand{\nstime}{\changetime}
\newcommand{\nstimeeoolt}{{t}}

\newcommand{\snstime}{\sigma}
\newcommand{\nstr}{\nst{\,\bR}}
\newcommand{\nstz}{\nst{\,\bZ}}
\newcommand{\nstR}{\nst{\,\bR}}
\newcommand{\nstn}{\nst{\,\bN}}
\newcommand{\nstN}{\nst{\,\bN}}
\newcommand{\continuation}[1]{\mathsf{Continuations}\left(#1\right)}
\providecommand{\vabs}[1]{\lvert#1\rvert}
\newcommand{\abs}{\ensuremath{\text{\scriptsize{\ding{54}}}}}
\newcommand{\squared}{\mathfrak{E}}
\newcommand{\underapprox}{\mathfrak{U}}
\newcommand{\overapprox}{\mathfrak{O}}
\newcommand{\irrelevant}{\mbox{\sc i}}
\newcommand{\dt}{dt}
\newcommand{\du}{du}
\newcommand{\zcvar}{\zeta}
\newcommand{\OK}{{\,|\!\!\models\,}}
\newcommand{\enabled}[2]{{#1\,{:}\;!#2}}
\newcommand{\eend}{{\sf end}}
\newcommand{\zc}{{\sc zc}}
\newcommand{\cB}{{\cal B}}
\newcommand{\fullB}{\overline{\cB}}
\newcommand{\cX}{{\cal X}}
\newcommand{\fullX}{\mathbb{X}}
\newcommand{\cF}{{\cal F}}
\newcommand{\archis}{\mathbb{A}}
\newcommand{\archi}{\mathcal{A}}
\newcommand{\cA}{{\cal A}}
\newcommand{\cM}{{\cal M}}
\newcommand{\cN}{{\cal N}}
\newcommand{\cY}{{\cal Y}}
\newcommand{\fullY}{\overline{\cY}}
\newcommand{\cE}{\mathbb{E}}
\newcommand{\EE}{\mathbf{E}}
\newcommand{\cT}{{\cal T}}
\newcommand{\cU}{{\cal U}}
\newcommand{\cS}{{\cal S}}
\newcommand{\CS}{\mathbf{ER}}
\newcommand{\cQ}{{\cal Q}}
\newcommand{\cP}{{\cal P}}
\newcommand{\cH}{{\cal H}}
\newcommand{\scott}{\mathbb{D}}
\newcommand{\useful}[2]{\lfloor#1\rfloor_{#2}}
\newcommand{\when}{\prog{if}}
\newcommand{\doo}{\prog{then}}
\newcommand{\mdAE}{mdAE}
\newcommand{\deltaAE}{$\vsmall\!${A\!E}}
\newcommand{\bdae}{mdAE}
\newcommand{\mDAE}{mDAE}
\newcommand{\hDAE}{hDAE}
\newcommand{\dae}{\mathrm{dAE}}
\newcommand{\bE}{\mathbf{E}}
\newcommand{\bI}{\mathbb{I}}
\newcommand{\bR}{\mathbb{R}}
\newcommand{\bT}{\mathbb{T}}
\newcommand{\bfT}{\mathbf{T}}
\newcommand{\bfM}{\mathbb{M}}
\newcommand{\bfb}{\mathbf{b}}
\newcommand{\bfx}{\mathbf{x}}
\newcommand{\bfu}{\mathbf{u}}
\newcommand{\bfv}{\mathbf{v}}
\newcommand{\bfy}{\mathbf{y}}
\newcommand{\bfa}{\mathbf{a}}
\newcommand{\bfA}{\mathbf{A}}
\newcommand{\cbT}{\overline{\bT}}
\newcommand{\sfj}{{\sf j}}
\newcommand{\sfA}{{\sf A}}
\newcommand{\bJ}{{\bf J}}
\newcommand{\bj}{{\bf j}}
\newcommand{\bS}{{\bf S}}
\newcommand{\bSs}{\mathcal{S}}
\newcommand{\bY}{{\bf Y}}
\newcommand{\bN}{\mathbb{N}}
\newcommand{\bB}{\mathbb{B}}
\newcommand{\bbS}{\mathbb{S}}
\newcommand{\bX}{\mathbb{X}}
\newcommand{\bU}{\mathbf{U}}
\newcommand{\bG}{\mathbf{G}}
\newcommand{\bC}{\mathbf{C}}
\newcommand{\bV}{\mathbf{V}}
\newcommand{\bZ}{\mathbb{Z}}
\newcommand{\bfZ}{\mathbf{Z}}
\newcommand{\bQ}{\mathbb{Q}}
\newcommand{\bD}{\mathsf{D}}
\newcommand{\bbD}{\mathbf{D}_b}
\newcommand{\barR}{\overline{{\bR}}}
\newcommand{\iindex}{{\it Ind}}
\newcommand{\oon}{\mbox{$\bf on$ }}
\newcommand{\iinit}[1]{#1_{\sf init}}
\newcommand{\iiinit}[3]{#1^{\sf init}_{\controlled{#3}{#2}}}
\newcommand{\ttt}{\mbox{\sc t}}
\newcommand{\fff}{\mbox{\sc f}}
\newcommand{\win}{\in_{\!_{_{+}}}\!\!}
\newcommand{\consistentcausality}{consistent causality}
\newcommand{\Consistentcausality}{Consistent causality}
\newcommand{\precond}{\mbox{pre-condition}}
\newcommand{\postcond}{\mbox{post-condition}}
\renewcommand{\time}{\partial}
\newcommand{\true}{{\bf when}}
\newcommand{\Bool}{{\sf Bool}}
\newcommand{\boolexp}{{\it boolexp}}
\newcommand{\zerocross}{{\it zero}}
\newcommand{\cchange}{{\bf change}}
\newcommand{\uup}{{\bf up}}
\newcommand{\aand}{\mathtt{and}}
\newcommand{\nneg}{\mathtt{not}}
\newcommand{\nnot}{{\bf not}}
\newcommand{\durtype}{\Delta}
\newcommand{\evt}{\mathbf{evt}\,}
\newcommand{\dure}{\mathbf{con}}
\newcommand{\system}{S}
\newcommand{\SOM}{\textsc{som}}
\newcommand{\SOW}{\textsc{sow}}
\newcommand{\Terms}{\mathcal{T}}
\newcommand{\false}{{\bf whennot}}
\newcommand{\Esterel}{{\sc Esterel}}
\newcommand{\esterel}{{\sc Esterel}}
\newcommand{\lustre}{{\sc Lustre}}
\newcommand{\lucy}{{\sc Lucid Synchrone}}
\newcommand{\signal}{{\sc Signal}}
\newcommand{\indexpara}[1]{{\|}_{_{#1}}}
\newcommand{\para}{{\,\|\,}}
\newcommand{\compat}{{\;\bowtie\;}}
\newcommand{\join}{{\;\sqcup\;}}
\newcommand{\vjoin}{\prog{join}}
\newcommand{\spara}{{\|\,}}
\newcommand{\abspara}{{|}}
\newcommand{\mutexpara}{{\|_\mu\,}}
\newcommand{\dpara}{{\!\vec{\,\|}\,}}
\newcommand{\paracs}{\,\|_{_{\rm CS}}}
\newcommand{\good}{\xi}
\newcommand{\nsx}{\bar{x}}
\newcommand{\nsv}{\bar{v}}
\newcommand{\nsw}{\bar{w}}
\newcommand{\involved}{\,\overline{\in}\,}
\newcommand{\vecinvolved}{\,{\in}_{\rm out}\,}
\newcommand{\notvecinvolved}{\,{\not\in}_{\rm out}\,}
\newcommand{\notinvolved}{\,\overline{\not\in}\,}
\newcommand{\involvedg}[1]{\,\overline{\in}_{#1}\,}
\newcommand{\Out}{{\bf IsOut}}
\newcommand{\invar}[1]{#1^{\sf in}}
\newcommand{\outvar}[1]{#1^{\sf out}}
\newcommand{\svar}{s}
\newcommand{\scalarp}[2]{\langle{#1},{#2}\rangle}
\newcommand{\entail}[2]{\frac{#1}{#2}}
\newcommand{\optw}[2]{{\sf J}^{#2}_{#1}}
\newcommand{\murule}[2]{\displaystyle\frac{#1}{#2}}

\newcommand{\Marc}[1]{{\sf Marc: #1}}

\newcommand{\ifequation}{{\tt if}-equation}
\newcommand{\whenequation}{{\tt when}-equation}
\newcommand{\ifGuards}{\Gamma^{\tt if}}
\newcommand{\whenGuards}{\Gamma^{\tt wh}}
\newcommand{\deronde}{\partial}
\newcommand{\gtype}{\tau}
\newcommand{\head}{{\sf h}}
\newcommand{\tail}{{\sf t}}
\newcommand{\followup}{+}
\newcommand{\lra}{\leftrightarrow}

\newcommand{\Vars}{\mathcal{X}}
\newcommand{\Scottvars}{\mathcal{S}}

\newcommand{\atomicact}[1]{\mathsf{#1}}
\newcommand{\asgnt}[1]{{#1}}
\newcommand{\asgnf}[1]{\overline{#1}}
\newcommand{\asgnd}[1]{\asgnf{#1}}
\newcommand{\asgnw}[1]{\sharp{#1}}
\newcommand{\nxt}[1]{\postset{#1}}
\newcommand{\rnd}[1]{#1^{\deronde}}
\newcommand{\prm}[1]{#1'}
\newcommand{\priv}[2]{#1_{\downarrow{#2}}}

\newcommand{\nan}{\mbox{NaN}}
\newcommand{\mlast}[1]{#1^-}
\newcommand{\shift}[1]{\mbox{FS}({#1})}
\newcommand{\latent}[1]{\mbox{LE}({#1})}
\newcommand{\REMOVE}[1]{}
\newcommand{\redundent}[1]{\Delta\gets{#1}}


\newcommand{\syntax}[1]{\mathtt{#1}}
\newcommand{\enbld}[2]{\mathit{Enab}({#2})}
\newcommand{\disbld}[2]{\mathit{Disab}({#2})}
\newcommand{\ndef}[1]{\mathit{Undef}({#1})}
\newcommand{\neval}[1]{\mathit{Uneval}({#1})}

\newcommand{\cnfa}{$\asgnt{\omega_1},\asgnt{\omega_2}$}
\newcommand{\cnfb}{$\begin{array}{c}\asgnf{\gamma},\asgnt{\omega_1},\asgnt{\omega_2}, \\
    \asgnd{e_3},\asgnd{e_4}\end{array}$}
\newcommand{\cnfc}{$\begin{array}{c}\asgnf{\gamma},\asgnt{\omega_1},\asgnt{\omega_2}, \\
    \asgnt{\tau_1},\asgnt{\tau_2},\asgnt{\nxt{\omega_1}},\asgnt{\nxt{\omega_2}}, \\
    \asgnt{\rnd{e_1}},\asgnt{\rnd{e_2}},\asgnd{e_3}, \\
    \asgnd{e_4},\asgnt{e_5},\asgnt{e_6}\end{array}$}
\newcommand{\cnfd}{$\begin{array}{c}\asgnt{\gamma},\asgnt{\omega_1},\asgnt{\omega_2}, \\
    \asgnd{e_5},\asgnd{e_6} \end{array}$}
\newcommand{\cnfe}{$\begin{array}{c}\asgnt{\gamma},\asgnt{\omega_1},\asgnt{\omega_2}, \\
    \asgnt{\tau_1},\asgnt{\tau_2},\asgnt{\nxt{\omega_1}},\asgnt{\nxt{\omega_2}}, \\
    \asgnt{\rnd{e_1}},\asgnt{\rnd{e_2}},\asgnt{\nxt{e_3}}, \\
    \asgnt{e_4},\asgnd{e_5},\asgnd{e_6} \end{array}$}
\newcommand{\cnff}{$\asgnt{\omega_1},\asgnt{\omega_2},\asgnw{e_3}$}
\newcommand{\cnfg}{$\begin{array}{c}\asgnf{\gamma},\asgnt{\omega_1},\asgnt{\omega_2}, \\
    \asgnd{e_3},\asgnd{e_4}\end{array}$}
\newcommand{\cnfh}{$\begin{array}{c}\asgnf{\gamma},\asgnt{\omega_1},\asgnt{\omega_2}, \\
    \asgnt{\tau_1},\asgnt{\tau_2},\asgnt{\nxt{\omega_1}},\asgnt{\nxt{\omega_2}}, \\
    \asgnt{\rnd{e_1}},\asgnt{\rnd{e_2}},\asgnd{e_3}, \\
    \asgnd{e_4},\asgnt{e_5},\asgnt{e_6}\end{array}$}
\newcommand{\cnfi}{$\begin{array}{c}\asgnt{\gamma},\asgnt{\omega_1},\asgnt{\omega_2}, \\
    \asgnt{e_3},\asgnd{e_5},\asgnd{e_6} \end{array}$}
\newcommand{\cnfj}{$\begin{array}{c}\asgnt{\gamma},\asgnt{\omega_1},\asgnt{\omega_2}, \\
    \asgnt{\tau_1},\asgnt{\tau_2},\asgnt{\nxt{\omega_1}},\asgnt{\nxt{\omega_2}}, \\
    \asgnt{\rnd{e_1}},\asgnt{\rnd{e_2}},\asgnt{e_3},\asgnt{\nxt{e_3}}, \\
    \asgnt{e_4},\asgnd{e_5},\asgnd{e_6} \end{array}$}

\newcommand{\lbla}{$\asgnf{\gamma};\asgnd{e_3};\asgnd{e_4}$}
\newcommand{\lblb}{$\asgnt{\gamma};\asgnd{e_5};\asgnd{e_6}$}
\newcommand{\lblc}{$\begin{array}{c}\asgnt{e_5};\asgnt{e_6}; \\ \asgnt{\rnd{e_1}};\asgnt{\rnd{e_2}}\end{array}$}
\newcommand{\lbld}{$\asgnt{\rnd{e_1}}+\asgnt{\rnd{e_2}}+\asgnt{\nxt{e_3}}+\asgnt{e_4}$}
\newcommand{\lble}{$\asgnf{\gamma};\asgnd{e_3};\asgnd{e_4}$}
\newcommand{\lblf}{$\asgnt{\gamma};\asgnd{e_5};\asgnd{e_6};\redundent{e_3}$}
\newcommand{\lblg}{$\asgnt{e_5};\asgnt{e_6};\asgnt{\rnd{e_1}};\asgnt{\rnd{e_2}}$}
\newcommand{\lblh}{$\begin{array}{c}\asgnt{\rnd{e_1}}+\asgnt{\rnd{e_2}}+ \\ \asgnt{\nxt{e_3}}+\asgnt{e_4}\end{array}$}


\newcommand{\gencodea}{\begin{array}{l}\tau_1 = 0; \; \tau_2 = 0; \\ \omega_1' = a_1(\omega_1) + b_1(\omega_1) \tau_1; \\ \omega_2' = a_2(\omega_2) + b_2(\omega_2) \tau_2\end{array}}

\newcommand{\gencodeb}{\begin{array}{l}\tau_1 = \nan; \; \tau_2 = \nan; \\ \omega_1^+ = \frac{b_2(\omega_2^-) \mlast{\omega_1} + b_1(\omega_1^-) \mlast{\omega_2}}{b_1(\omega_1^-)+b_2(\omega_2^-)}; \\ \omega_2^+ = \omega_1^+\end{array}}

\newcommand{\gencodec}{\begin{array}{l}\tau_1 = (a_2(\omega_2) - a_1 (\omega_1)) / (b_1(\omega_1) + b_2(\omega_2)); \tau_2 = - \tau_1; \\ \omega_1' = a_1(\omega_1) + b_1(\omega_1) \tau_1; \; \omega_2' = a_2(\omega_2) + b_2(\omega_2) \tau_2;  \\ \mbox{\textbf{constraint}} \; \omega_1 - \omega_2 = 0\end{array}}

\newcommand{\gencoded}{\begin{array}{l}\tau_1 = 0; \; \tau_2 = 0; \\ \omega_1^+ = \mlast{\omega_1}; \\ \omega_2^+ = \mlast{\omega_2}\end{array}}


\newcommand{\ttrue}{\mbox{\sc t}}
\newcommand{\ffalse}{\mbox{\sc f}}
\newcommand{\cstar}{\centering $\star$}
\newcommand{\rel}{\rightsquigarrow}


\newcommand{\bNinfty}{\underline{\bN}}
\newcommand{\bZinfty}{\overline{\bZ}}
\newcommand{\Vvars}{\mathcal{V}}
\newcommand{\iset}{\bar{X}}
\newcommand{\jset}{\bar{Y}}
\newcommand{\vvar}{v}
\newcommand{\wvar}{w}
\newcommand{\Ooptscore}[1]{\optscore_{#1}}
\newcommand{\aargmax}[1]{{\sf A}_{#1}}
\newcommand{\bL}{\mathbb{L}}
\newcommand{\opp}{\diamond}

\newcommand{\scoproj}[2]{\coPi^s_{#1}\!\left[#2\right]}




%
%
%
%
%
%

\newcommand{\skernel}{io\sinterface}
\newcommand{\smethod}{$\Sigma$-method}
\newcommand{\ssystem}{$\Sigma$-system}
\newcommand{\scomponent}{$\Sigma$-component}
\newcommand{\ssolution}{$\Sigma$-solution}
\newcommand{\smessage}{$\Sigma$-message}
\newcommand{\sinterface}{$\Sigma$-interface}
\newcommand{\spreinterface}{$\Sigma$-abstraction}
\newcommand{\scointerface}{$\Sigma$-cointerface}
\newcommand{\sdecomposition}{$\Sigma$-decomposition}
\newcommand{\sproblem}{$\Sigma$-problem}
\newcommand{\smatrix}{$\Sigma$-matrix}
\newcommand{\smatrices}{$\Sigma$-matrices}
\newcommand{\sstructure}{$\Sigma$-structure}
\newcommand{\bool}{\mathbb{B}}
\newcommand{\nat}{\mathbb{N}_{{-}\infty}}
\newcommand{\FF}{\mathbb{F}}
\newcommand{\XX}{\mathbb{X}}
\newcommand{\Selectors}{{\sf s}}
\newcommand{\Xselectors}{{\sf X}}
\newcommand{\Allselectors}[1]{\Selectors[#1]}
\newcommand{\Interfaces}[1]{{\sf Interfaces}(#1)}
\newcommand{\dbmof}[1]{\LPD{}[#1]}
\newcommand{\YY}{\mathbb{Y}}
\newcommand{\precomp}{\,{|\!|\!|}\,}
\newcommand{\edge}{{f\!,x}}
\newcommand{\zedge}{{f\!,z}}
\newcommand{\gedge}{{g,x}}
\newcommand{\aalpha}[1]{\alpha_{\!#1}}
\newcommand{\bbeta}[1]{\vartheta_{\!#1}}
\newcommand{\hatedge}{{f\!,\equivsim{x}}}
\newcommand{\fshiftdiff}[3]{f_{#1,#2,#3}}
\newcommand{\yedge}{{f\!,y}}
\newcommand{\redge}{{x,f}}
\newcommand{\LP}[1]{\mathcal{L}_{#1}}
\newcommand{\LPsel}[2]{\mathcal{L}^{#2}_{#1}}
\newcommand{\CPP}[1]{{C}^{\sf p}_{#1}}
\newcommand{\CPPsel}[2]{{C}^{\sf p{#2}}_{#1}}
\newcommand{\CPD}[1]{{C}^{\sf d}_{#1}}
\newcommand{\LPP}[1]{\mathcal{L}^{\sf p}_{#1}}
\newcommand{\LPPsel}[2]{\mathcal{L}^{{#1}{\sf p}}_{#2}}
\newcommand{\LPDsel}[2]{\mathcal{L}^{{#1}{\sf d}}_{#2}}
\newcommand{\LPD}[1]{\mathcal{L}^{\sf d}_{#1}}
\newcommand{\ILPD}[1]{\widehat{\mathcal{L}}^{\sf d}_{#1}}
\newcommand{\LPDmatch}[2]{\mathcal{L}^{{\sf d}#2}_{#1}}
\newcommand{\loc}[1]{#1^{\!\ell}}
\newcommand{\vis}[1]{#1^{\!s}}
\newcommand{\rrows}[1]{{\sf rows}(#1)}
\newcommand{\impulsion}[1]{\vec{#1}}
\newcommand{\indexreduced}[1]{#1_\Downarrow}
\newcommand{\consistency}[1]{#1_\Uparrow}
\newcommand{\complete}[1]{#1_\Updownarrow}
\newcommand{\newmode}[1]{#1^+_\Downarrow}
\newcommand{\ssem}[1]{{\left\llbracket{#1}\right\rrbracket_s}}
\newcommand{\wsem}[1]{{\left\llbracket{#1}\right\rrbracket_w}}
\newcommand{\sequiv}{\equiv_s}
\newcommand{\wequiv}{\equiv}
\newcommand{\locprim}[1]{#1^{\prime\,\sf loc}}
\newcommand{\visprim}[1]{#1^{\prime\,\sf vis}}
\newcommand{\preinterf}[1]{\mathbf{A}_{#1}}
\newcommand{\interf}[1]{\mathbf{I}_{#1}}
\newcommand{\pinterf}[1]{\mathbf{I}^{\sf p}_{#1}}
\newcommand{\dinterf}[1]{\mathbf{I}^{\sf d}_{#1}}
\newcommand{\btfinterf}[1]{\mathbf{I}^{\sf btf}_{#1}}
\newcommand{\ppreinterf}[1]{\mathbf{A}^{\sf p}_{#1}}
\newcommand{\dpreinterf}[1]{\mathbf{A}^{\sf d}_{#1}}
\newcommand{\btfpreinterf}[1]{\mathbf{A}^{\sf btf}_{#1}}
\newcommand{\cointerf}[1]{\mathbf{K}_{#1}}
\newcommand{\decomp}[2]{\mathbf{D}_{#1,#2}}
\newcommand{\jinterf}[1]{\mathbf{J}_{#1}}
\newcommand{\IP}[1]{\mathit{IP}_{\!#1}}
\newcommand{\Edges}{E}
\newcommand{\vecEdges}{\vec{E}}
\newcommand{\Edge}{{F\!,X}}
\newcommand{\bfEdges}{\mathbf{E}}
\newcommand{\bEdges}{\mathbb{F}{\times}\mathbb{X}}
\newcommand{\bfedge}{\mathbf{e}}
\newcommand{\weight}[1]{\sigma_{\!#1}}
\newcommand{\mcweight}[1]{\impulsion{\sigma}_{\!#1}}
\newcommand{\bfun}{\mathbf{1}}
\newcommand{\bfzero}{\mathbf{0}}
\newcommand{\selector}{Y}
\newcommand{\tselector}{Y}
\newcommand{\openst}{\mathbf{S}}
\newcommand{\selectortriv}{\sigma_{\rm triv}}
\newcommand{\II}{\mathcal{I}}
\newcommand{\cL}{\mathcal{L}}
\newcommand{\sigmamatrix}{\Delta}
\newcommand{\bNinf}{\bN_{-\infty}}
\newcommand{\bZinf}{\bZ_{-\infty}}
\newcommand{\parts}[1]{\cP(#1)}
\newcommand{\with}[2]{#2^{\wedge#1}}

\newcommand{\Block}{\beta}
\newcommand{\neint}{\!\!\!}
\newcommand{\gguard}{g}
\newcommand{\ffun}{\mathit{exp}}
\newcommand{\numexp}{{\it numexp}}
\newcommand{\BaseVars}{\mathbf{\Xi}}
\newcommand{\YVars}{\mathcal{Y}}
\newcommand{\Vals}{\mathcal{V}}
\newcommand{\sem}[1]{[\![#1]\!]}
\newcommand{\semantique}[1]{{\sf Sem}(#1)}
\newcommand{\bigsemantique}[1]{{\sf Sem}\bigl(#1\bigr)}
\newcommand{\modetraj}{\vec{\mode}}
\newcommand{\prelead}[1]{#1^-}
\newcommand{\postlead}[1]{#1^+}
\newcommand{\Subarray}[1]{\widehat{#1}}
\newcommand{\Impulsives}[2]{\mathbf{Impuls}^{#2}(#1)}
\newcommand{\impulsorder}[1]{\iota\!\left(#1\right)}
\newcommand{\impulsecause}{\Ra_\iota}

\newcommand{\hookdownarrow}{\mathrel{\rotatebox[origin=t]{90}{\reflectbox{$\hookrightarrow$}}}}
\def\Circlearrowleft{\ensuremath{\rotatebox[origin=c]{180}{$\circlearrowleft$}}}

\begin{titlepage}
\makeRR
\end{titlepage}
\hypersetup{pageanchor=true}
\pagenumbering{arabic}
\clearpage
{
	\hypersetup{linkcolor=black}
	\tableofcontents
}
\clearpage
\section{Introduction}
\subsubsection*{Motivation and related work}
In this paper, we consider \emph{multimode DAE systems}, i.e., collections of finitely many (possibly nonlinear) equations of the form
\beq
\prog{if}\;b\;\prog{then}\;f(\mbox{the } x_i \mbox{'s and their derivatives})=0 \,, \label{htrgfshytngf}
\eeq
where $x_i,i{=}1,\dots,n$ are (time-dependent) numerical variables, $f$ is a sufficiently differentiable numerical function,\footnote{What ``sufficiently'' means will be made precise in Section~\ref{nhmgfcvkjh}. \label{skdjhsgksujg}} and $b$ is a Boolean expression of predicates  over system variables and time. Boolean $b$ guards equation $f=0$, meaning that this equation is enabled when $b$ evaluates to $\ttrue$ (the constant ``true''), disabled otherwise.
We assume that each predicate takes the form $g(x^-){\geq} 0$, where
 $g$ is a smooth function of system variables, and $x^-(t)=\lim_{s\nearrow{t}}x(s)$ is the left-limit of $x$ at instant $t$. Function $g$ is often called a \emph{zero-crossing} function: its successive switches from negative to positive values trigger the onsets of the predicate.

Multimode DAE systems constitute a generic form for the hybrid systems specified by Modelica~\cite{MODELICA12,dlr11508}, a popular DAE-based modeling language dedicated to the modeling of multiphysical systems.\footnote{\href{https://modelica.org/}{Modelica} is the most well-known instance of such languages; \href{https://plm.sw.siemens.com/fr-FR/simcenter/systems-simulation/amesim/}{Amesim} and \href{https://fr.mathworks.com/products/simscape.html}{Simscape} are other languages used in the industry.} Modelica  is by itself physics-agnostic, which allows the user to include in the model a specification of the software code controlling the system.\footnote{This is in contrast to physics-oriented modeling methodologies such as bond graphs~\cite{thoma1975introduction} or port-Hamiltonian modeling~\cite{vanderSchaft2014PortHamiltonian}.} It is  important to establish such modeling languages on solid mathematical grounds. 

Multimode DAE systems, however, exhibit significant difficulties. First, there is no reference mathematical definition of the solution of a multimode DAE in general, but only for specific subclasses of physics, e.g., contact mechanics~\cite{brogliato2012nonsmooth}. Second, modes can be numerous---the number of modes typically grows exponentially with the number of subsystems. Due to these difficulties, existing tools do not implement mode-aware compilation:
some models, although clearly valid, fail to get correctly simulated. Such pathological models are by no means exceptional, nor are they difficult to exhibit~\cite{electronics11172755}. This work addresses the first difficulty by: 
\begin{enumerate}
	\item giving a mathematically sound definition of mode change events in multimode DAEs, and 
	\item proposing an effective algorithm computing the hot restart after a mode change. 
\end{enumerate}

\paragraph*{\normalsize Related work:}
An extensive review of the litterature on multimode DAEs is available in~\cite{DBLP:series/lncs/BenvenisteCEGOP19};
in this brief discussion, we only collect its major findings.
For selected physics, e.g., multi-body systems with contacts and electrical circuits with idealized switches and diodes, dedicated methods are proposed to handle possible impulses~\cite{PfeifferGlocker2008,Pfeiffer2012,Schoeder2013,Heemels2002,Barela2016}. It is not clear if the special methods used in these areas extend to general models.
Mehrmann et al.~\cite{HAMANN2008693} propose numerical techniques to handle chattering between modes. In Zimmer's PhD thesis~\cite{Zimmer2010}, multimode DAEs are considered with varying structure and index; however, impulsive behaviors are not supported. Both references assume that consistent reset values are explicitly given for each mode: this assumption is not suited to a compositional framework where one wants to assemble predefined physical components.
Elmqvist et al.~\cite{MultiMode2014,VaryingIndex2015} propose a high level description of multimode models as an extension to the Modelica 3.3 state machines. However, mode changes with impulsive behavior are not supported and not all types of multimode systems can be handled,
as mentioned above.

In our previous work~\cite{DBLP:series/lncs/BenvenisteCEGOP19}, we proposed for the first time an alternative approach for handling multimode DAEs. A restricted class of multimode DAEs, possibly involving impulsive behaviors, was considered (since then known as \emph{semi-linear} multimode DAEs~\cite{DBLP:journals/arc/BenvenisteCM20}). For this subclass, two alternative approaches were investigated and shown to be equivalent: a Gear-Gupta-Leimkuhler method~\cite{gear1985automatic}, implemented in the Julia package \href{https://modiasim.github.io/Modia.jl/stable/}{Modia}, and a novel approach relying on nonstandard analysis~\cite{Lindstrom}. The latter approach is the root of the one presented in this paper.

\paragraph*{\normalsize Zooming in:}
Key references for our work are the important contributions by Stephan Trenn and coworkers~\cite{TrennPhD2009,Tren09b,LibeTren12,KausTren17b,ChenTren23}. 
The pioneering work~\cite{TrennPhD2009} investigates the inconsistent DAE initialization problem: a regular linear DAE $E\dot{x}{=}Ax{+}f$ is considered for $t{\geq}{0}$, with an inconsistent initialization $x_{(-\infty,0)}$, meaning that the left limit $\lim_{s\nearrow 0}x(s)$ does not satisfy the consistency constraints imposed by the DAE. This possibly causes impulsive behavior for some variables at instant 0. The classical theory of distributions over the timeline $\bR$ is inadequate to handle this problem, as the value of a distribution at a given instant is undefined in general. The use of \emph{piecewise-smooth distributions} was proposed in~\cite[Chapter 2]{TrennPhD2009} to give an explicit solution for the inconsistent initialization of a linear DAE, by which the behavior of the impulsive variables is described as a linear combination of Dirac derivatives. A key pillar in computing this solution is the \emph{quasi-Weierstrass form} for a linear DAE, which consists in constructing a state basis and recombining equations, such that the impulsive order (differentiation degree of the Dirac function $\delta$) of each variable is well identified, and pure constraints are separated from the ODE part.

An interesting step forward is~\cite{KausTren17b}, in which the inconsistent initialization problem for DAE $E\dot{x}{=}Ax{+}g(x){+}f$ is investigated, where $g$ is a smooth nonlinear function. The authors assume that the nonlinear part $g(x)$ remains ``foreign to impulsions'' (see~\cite{KausTren17b} for a formalization)\,---\,the reason being that, for $x$ impulsive and $g$ nonlinear, $g(x)$ is undefined in general ($g(\delta)$ being an example).
Under this assumption, an appropriate time-invariant change of coordinates and equations allows the authors to statically decompose the DAE system into three subsystems: 1) a nonlinear ODE, 2) a nonlinear static constraint, and 3) a linear DAE carrying the impulsive part of the system. Using this decomposition and reusing the background from linear DAEs, the authors propose an explicit solution to the  inconsistent initialization problem. 
These results comply with the linear control systems vision: states are internal and are only defined up to a change of basis.

\subsubsection*{Objectives and approach}
Our work differs from the works of Trenn \emph{et al.} in a number of aspects. \

First, a key difference is our overall objective, namely: \emph{to provide mathematical soundness for the compilation of DAE-based modeling languages such as Modelica.} In Modelica models, parameters occurring in models are entered right before launching a simulation, i.e., after the simulation code was compiled. Therefore, compilation involves the structure of the model, not the actual values of its parameters. 
In particular, a key transformation of a DAE model is the \emph{structural analysis,} mainly consisting of the \emph{index reduction,} which transforms the DAE into an ODE-like system by suitably differentiating selected equations. The technique used in Modelica tools to perform index reduction uses the model structure, not the numerical details of the model. 
Our work complies with this philosophy: 
\begin{enumerate}
	\item 
\emph{Changes of state basis are prohibited}\footnote{This prohibits numerical steps such as the construction of the above mentioned quasi-Weierstrass form. }\,---\,such changes are little relevant in the context of physical modeling (unlike in black-box modeling).
\item
Whereas in classical linear systems theory, regularity (of matrices or linear pencils) plays a central role, \emph{we give up numerical regularity and replace it by {structural regularity}} (also called structural nonsingularity). 
\end{enumerate}
A square matrix is \emph{structurally nonsingular} if all of its diagonal entries can be made nonzero by pre- and post-multiplying it by permutation matrices. A static system of nonlinear equations is \emph{structurally nonsingular} if so is its Jacobian around its solution. Structural nonsingularity holds if and only if a one-to-one mapping exists between equations and dependent variables, i.e., a \emph{perfect matching}.
Structural nonsingularity is necessary, and generically sufficient,\footnote{Generically means that the matrix remains nonsingular almost everywhere when its non-zero entries vary over some neighborhood. See~\cite{benveniste:hal-03104030} for a short tutorial on structural methods.} for exact regularity of a static system. 
Structural regularity is checked on the incidence graph of the system (an abstraction of its structure), relying on graph-based algorithms that can scale up much better than numerical ones: those are used in high-performance computing, and in all the compilers of DAE-based modeling tools~\cite{Fritzson_ModelicaStructAna_02}. 
Moreover, as structural methods do not rely on parameter values, they allow us to identify, prior to simulation, if the model is over-, under-, or well-determined, providing useful information to model designers.

The second difference is that we do not describe the full trajectories of all the variables, including impulsion events. Instead, \emph{we detect impulsive behaviors, without identifying their exact nature,\footnote{That is, without describing them as a suitable linear combination of Dirac measures. Recall that the occurrence of $g(\delta)$, for $g$ nonlinear, prevents us from providing a full definition of solutions in this case.} and rescale impulsive variables to make them non-impulsive.} Our rescaling algorithm bears similarities with J. Pryce's \sigmamethod\ for the structural analysis of (single-mode) DAEs~\cite{Pryce01}.

\subsubsection*{Contribution}
Based on the above approach, our contributions are the following:
\begin{itemize}
\item {A general definition of \emph{hot restart}} (Problem~\ref{dsjkchgfckajycgk});
\item {A system of \emph{rescaling equations}} (Problem~\ref{likgusefhdrlpouhiouh});
\item {A procedure to generate hot restarts} (Procedure~\ref{skedfjuyghwsaliukgh});
\item {A proof of correctness of Procedure~\ref{skedfjuyghwsaliukgh}} (Theorem~\ref{skldjhfgslkdufghikl});
\item {Bounds for design parameters of Problem~\ref{likgusefhdrlpouhiouh}} (Theorem~\ref{sekdfuysageloiug});
\item An effective algorithm for solving rescaling equations.
\end{itemize}
Procedure~\ref{skedfjuyghwsaliukgh} generates the system of equations by which restart values for the states are determined from the values just before the mode change. Uniqueness of the restart system of equations is guaranteed, when it exists, by Theorem~\ref{skldjhfgslkdufghikl}, and its non-existence typically corresponds to a lack of determinism, expressing that the model is insufficiently specified for hot restart.

In our previous work~\cite{DBLP:journals/arc/BenvenisteCM20}, the same problem was addressed, by building on top of nonstandard analysis~\cite{Lindstrom}. The resulting algorithm, however, was not completely specified and was difficult to analyze. The method introduced in the present paper fixes both issues. In particular:
\begin{itemize}
	\item 
It handles general nonlinear systems; mode changes can be state-based, not only time-based.\footnote{Consequently, all multimode DAE models are supported by our compilation method.} Yet, hot restart succeeds only if impulsive variables are involved \emph{linearly} in the system model---this restriction mirrors the one formulated in~\cite{KausTren17b}, where it was formulated in Definition 6 as an extra and complex ``condition (G$_p$)''. In our case, there is no need for a separate check: our algorithm discovers by itself if this condition is satisfied or violated.
\item 
It is physics-agnostic, yet is able to reveal hidden physical invariants: in the cup-and-ball example of Section~\ref{jhtgfkuytjrsdfghk}, the hot restart we generate preserves angular momentum, although no such law was explicitly stated in the model. 
\item 
By being graph-based, it has the potential to scale up to very large systems, unlike previous approaches.
\end{itemize}
{This report is organized as follows.} Section~\ref{jhtgfkuytjrsdfghk} introduces our approach by means of an illustrative example. This example is simple; yet, state-of-the-art DAE-based modeling tools fail to simulate it correctly---they actually crash when reaching mode change events. Background material is recalled in Section~\ref{jhgfvkjhhtrgfdhgfedbsv}: Sections~\ref{jhtgfkjhseghfaerv} and~\ref{nhmgfcvkjh} focus on the structural analysis of both static systems of equations and (single-mode) DAE systems, and Section~\ref{skdjfhsgkldjfgh} states a slight reformulation of the implicit function theorem. The core problem addressed by this paper, namely, the hot restart problem, is stated in Section~\ref{jashcgfackjhamgfjh}. Sections~\ref{hgrfdkhgfcvkjmh} and~\ref{dfkjshlksgfskdjfmg} develop our key result, namely the rescaling analysis. Main theorems regarding our approach are collected in Section~\ref{hgrfdjhtrsbdgfb}, and details on the resulting algorithm are developed in Section~\ref{jsdhgcfvsdjhcg}. Finally, a mathematical characterization of the hot restarts generated by our method is presented in Section~\ref{kaerjfhgkaewrjyfhg}. Main proofs are collected in Section~\ref{skduhvgskjvghkjh}. 

Additional proofs are collected in Appendices~\ref{hjgfgkkjmhnm}, \ref{kcsjgvcsldjghkj}, and~\ref{slviwsjrhfvluikvhk}. Additional examples illustrating the features of our approach are presented in Appendices~\ref{ksdfiygskdyjgf} and~\ref{hgrdfhreagfdsuj}. Finally, a comparison with the approach by Trenn et al. is discussed in Appendix~\ref{sldkvjshdblvckujsh}.

In a first version of this report, we only considered mode changes between two successive ``long modes'' (lasting for a positive duration with the same DAE dynamics), with no additional hot restart constraint. The case of finite cascades of successive transient modes (of zero duration) is now supported in Section~\ref{sdlkjhfcklsujdg} of this revised version, with limited objectives. This is achieved by considering additional hot restart constraints.

\section{A cup-and-ball example}
\label{jhtgfkuytjrsdfghk}
We develop a cup-and-ball game example to give an intuitive presentation of our approach; missing background is introduced informally and will be developed in Section~\ref{jhtgfkjhseghfaerv}. This example illustrates the main challenges to be addressed: mode changes are state-based, and they involve impulsive behaviors. As a matter of fact, major tools fail to simulate this model.

A ball, modeled by a point mass, is attached to one end of a rope, while the other end of the rope is fixed to the origin of the plane in the model. The ball is subject to the unilateral constraint set by the rope, but moves freely while the distance between the ball and the origin is less than its length. A model for a 2D version of this example is:
\beq
\mbox{\raisebox{-10mm}{\includegraphics[width=0.1\textwidth]{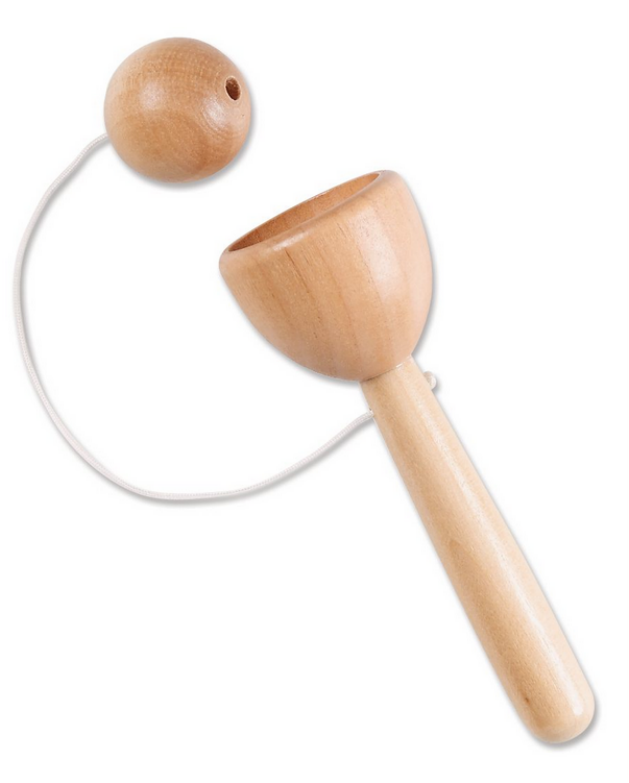}}}&\hspace*{2cm}
\left\{\bea{lll}
 0= \ddot{x}+{\tension}x && (\eqq_1) \\
 0= \ddot{y}+{\tension}y+g  && (\eqq_2) \\
0\leq{L^2}{-}(x^2{+}y^2)    && (\kappa_1) \\
0\leq\tension   && (\kappa_2) \\
0=\left[{L^2}{-}(x^2{+}y^2)\right]\times\tension   && (\kappa_3) 
\eea\right.
\label{reiutyekrtyk}
\eeq
where the {dependent variables} are the position $(x,y)$ of the ball in Cartesian coordinates and the rope tension $\tension$.

Subsystem $(\kappa_1,\kappa_2,\kappa_3)$ expresses that the tension is nonnegative, the distance between the ball and the origin is less than or equal to $L$, and one cannot have a nonzero tension and a
distance less than $L$ at the same time: (\ref{reiutyekrtyk}) is a \emph{complementarity system}, and such systems are key in \emph{non-smooth mechanics}~\cite{Pfeiffer2012}. This model is not of the form~(\ref{htrgfshytngf}) yet because of unilateral constraints $(\kappa_1,\kappa_2)$. Using the technique presented in~\cite{dlr11508}, we redefine the graph of $(\kappa_1,\kappa_2,\kappa_3)$ as a parametric curve, represented by the following three equations: 
\beq
\left\{\bea{rcl} 
\guard &=& [s\leq 0] 
\\
0 &=& \prog{if}\;\guard\;\prog{then} \;{L^2}{-}(x^2{+}y^2) \; \prog{else} \; \tension
\\
s &=& \prog{if}\; \guard \;\prog{then} -\tension \;\prog{else} \;{L^2}{-}(x^2{+}y^2)
\eea \right.
\label{dhwsatdrfshtr} 
\eeq
The model now constitutes a logico-numerical fixpoint equation with dependent variables $\ddot{x},\ddot{y},\tension,\guard,s$. Such equation can have zero, one, or infinitely many solutions. No characterization exists that could serve as a basis for a structural analysis. \emph{We thus decide to refuse solving such mixed logico-numerical systems.}

To break the fixpoint equation defining $\guard$, we choose to base guards on left-limits of signals. This yields the modified model (\ref{loeifuhpwoui}), where the modification is highlighted in {\color{red}red}. For convenience, we also grouped the equations that are only active in modes $\guard=\ttt$ and $\guard=\fff$, respectively:
\beq\left\{\bea{rll}
& 0= \ddot{x}+{\tension}x & (\eqq_1) \\
& 0= \ddot{y}+{\tension}y+g  & (\eqq_2) \\
& \remph{{\guard}= [s^-\leq{0}]; \guard(0)=\fff}  & \remph{(\straight_0)} \\
\when \; \guard\; \doo& 0={L^2}{-}(x^2{+}y^2)   & (\straight_1) \\
\prog{and}& 0=\tension+s   & (\straight_2) \\
\prog{else}& 0=\tension   & (\straight_3) \\
\prog{and}& 0=({L^2}{-}(x^2{+}y^2))-s   & (\straight_4) \\
\eea\right.
\label{loeifuhpwoui}
\eeq
where ``$\prog{else}$'' is a shorthand for ``$\when\;\prog{not}\;\guard\;\doo$''.
Let us assume that the ball is initially under free motion ($\guard{\,=\,}\fff$). Eventually, the rope will get straight ($\guard:\fff{\ra}\ttt$), with two possible continuations, depending on whether the impact is assumed \emph{elastic} (the ball bounces inward when the rope gets straight) or \emph{inelastic} (the rope remains straight).
\begin{itemize}
	\item If the impact is inelastic, then the radial velocity of the ball becomes zero, hence the rope gets straight ($\guard{\,=\,}\ttt$) until the gravity pushes the ball back to free motion ($\guard{\,=\,}\fff$). 
	\item If the impact is elastic, then the radial velocity changes its sign (the simplest model is that it gets opposite and equal in magnitude) and the ball bounces back to free motion ($\guard{\,=\,}\fff$); in this case, the ball spends zero time in $\guard{\,=\,}\ttt$ mode, we call it \emph{transient}. 
\end{itemize}
The respective timing views are the following:
\beq\bea{rcccc}
\mbox{inelastic}&:&\underbrace{\mbox{free motion}}_{\rm long} & \underbrace{{\mbox{straight rope}}}_{\rm long} & \underbrace{\mbox{free motion}}_{\rm long} 
\\ [6mm]
\mbox{elastic}&:&\underbrace{\mbox{free motion}}_{\rm long} & \underbrace{{\mbox{straight rope}}}_{\rm transient} & \underbrace{\mbox{free motion}}_{\rm long} 
\eea  \label{hythfjnhgfedshgf}
\eeq
Note that subsystem $(\kappa_1,\kappa_2,\kappa_3)$ leaves the impact law at mode change insufficiently specified: the model does not state whether the impact is elastic or inelastic..
We focus here on the inelastic case, where only long modes occur. An impulsive behavior is expected at the mode change from free motion to {straight rope}, $\guard:\fff{\,\ra\,}\ttt$; we will illustrate our method on this mode change.

We first consider each mode separately in (\ref{loeifuhpwoui}). Under the mode $\guard=\fff$, free motion occurs, which yields an ODE. In contrast, under $\guard=\ttt$, the rope is straight, thus enabling the algebraic constraint $(k_1)$: this yields a DAE. Index reduction is performed~\cite{CampbellGear1995,Pantelides1988}, which consists here in adding to the model the two \emph{latent equations} $(\dot{k_1},\ddot{k_1})$ obtained by successive differentiations of $(k_1)$. The resulting model splits into two parts~\cite{Pantelides1988}:
\beq\bea{rcl}
\bea{r} \mbox{leading}\\ \mbox{equations}\eea&\negesp:\negesp&\left\{\bea{lc}
0=\pmatch{\ddot{x}}+\tension x & (e_1) \\
0={\ddot{y}}+\pmatch{\tension} y & (e_2) \\
0=x\ddot{x}{+}\dot{x}^2{+}\dot{y}^2{+}y\pmatch{\ddot{y}} & (\ddot{k_1})
\eea\right.
\\ [5mm]
\bea{r} \mbox{consistency}\\ \mbox{equations}\eea&\negesp:\negesp&\,\left\{\bea{lc}
0={L^2}{-}(x^2{+}y^2) & \hspace*{9.2mm} ({k_1}) 
\\ 0=x\dot{x}+y\dot{y} & \hspace*{9.2mm} (\dot{k_1})
\eea\right.
\eea
\label{htdfkuytdshgd}
\eeq
In (\ref{htdfkuytdshgd}) and in the sequel, we highlight in {\color{blue}blue} an injective \emph{matching} of equations to variables. It is one-to-one for the leading equations, showing structural nonsingularity\,---\,see the discussion on structural nonsingularity in the introduction, and Section~\ref{jhtgfkjhseghfaerv} for a formal presentation. Indeed, the leading equations generically determine the leading variables $\ddot{x},\ddot{y},\tension$, while the consistency equation constraints the state variables $x,y,\dot{x},\dot{y}$ at initialization, leaving two degrees of freedom.

Focus next on mode change $\guard:\fff{\ra}\ttt$, and let $t$ be the instant when it occurs. We wish to identify how the previous mode $\guard{=}\fff$ influences the restart of the new mode $\guard{=}\ttt$. 
\begin{approach}\rm
		\label{isduyfsegdiksuygu} Our approach unfolds as follows. In a neighborhood of the mode change instant:
		\begin{enumerate}
			\item  \label{jhgfdytrhgf}    
			Each derivative is discretized using a simple forward Euler scheme with time step $\vsmall$, i.e., any derivative $\dot{x}$ is \emph{identified} with $\vsmall^{-1}({x}(t{+}\vsmall){-}{x}(t))$; we call the added equation $\dot{x}=\vsmall^{-1}({x}(t{+}\vsmall){-}{x}(t))$ an \emph{Euler identity};
			\item    \label{jytrgdhtsgdyntgfv} Using this interpretation of derivatives, we stack finitely many successive transitions of the resulting discrete-time system around the mode change instant, thus building an array of equations called \emph{mode change array}\,---\,the last transition stacked in this array will serve to initialize the next mode. The identification and compensation of impulsive behaviors is performed on this array, leading to its \emph{rescaling};
			\item \label{hgrfdjtrgdhtfjhgsdrf} Finally, we let $\vsmall{\,:=\,}0$ in the rescaled array.\footnote{Throughout this paper, symbol $:=$ means the assignment of a value to a variable or parameter.\label{iskdufygseukyg}} As a result, with the expansion of $\dot{x}(t)$ as $\vsmall^{-1}({x}(t{+}\vsmall){-}{x}(t))$, we retrieve the derivative of $x$ at $t$.	The resulting \emph{restart system} relates the latest state variables of this system, to the right limits of the trajectories of the previous mode; it is used for hot restart.\eproof
		\end{enumerate}
	\end{approach}
The approximations errors introduced by enforcing the Euler identity in step~\ref{jhgfdytrhgf} are handled in step~\ref{hgrfdjtrgdhtfjhgsdrf}; we prove later that, under certain conditions, the limit exists when $\vsmall\searrow{0}$ and is indeed obtained by setting $\vsmall{:=}0$.
To be properly used for the hot restart of the new mode, this restart system should satisfy the following requirements:
\begin{require}[for the restart system]
	\label{ksdjyhfgaskjfhgfkj} \
	\begin{enumerate}
		\item \label{ksjdfhsgkfjhg} It should be \emph{deterministic,} i.e., it should uniquely determine the restart state variables of the new mode;
		\item \label{mdfsjhfgvksjmhdgkj} The restart state variables should satisfy the consistency conditions of the new mode.
	\end{enumerate}
\end{require}
Requirement~\ref{ksdjyhfgaskjfhgfkj}.\ref{ksjdfhsgkfjhg} means that the mode change is determined. Requirement~\ref{ksdjyhfgaskjfhgfkj}.\ref{mdfsjhfgvksjmhdgkj} expresses that a hot restart should at least satisfy the conditions for being a consistent start.

Performing Step~\ref{jhgfdytrhgf} of Approach~\ref{isduyfsegdiksuygu} results in a discretization of the dynamics of the clutch in a neighborhood of the mode change. We focus on the neighborhood of the mode change, and we detail Step~\ref{jytrgdhtsgdyntgfv}. To simplify the writing, we use the following notation using the instant $t$ of mode change as a reference: for any variable $x$, let
\[
\preset{x}\equiv{x(t{-}\vsmall)}\,,\;
x\equiv{x(t)} \mbox{ and } 
\postset{x}\equiv{x(t{+}\vsmall)}\,.
\]
Using this notation, we display in  Fig.\,\ref{fdkghldkjhkjj} the dynamics of the new mode at instant $t$, as completed following (\ref{htdfkuytdshgd}). We regard it as a static system of equations $A_0$ and call it \emph{mode change array}. Variables $x,y,\dot{x},\dot{y}$ are fully determined by the previous mode, due to the Euler identities $x=\preset{x}+\vsmall{\times}{\preset{\dot{x}}}$ and $\dot{x}=\preset{\dot{x}}+\vsmall{\times}\preset{\ddot{x}}$ (and similarly for $y$): we call them \emph{past variables}. The dependent variables of $A_0$ are $\ddot{x},\ddot{y},\tension$ (past variables are excluded). 

\begin{figure}[!ht]

 \beqq
A_0:\left\{\bea{cclcll}
0&\!\!\!\!=\!\!\!\!& {\pmatch{\ddot{x}}+{\tension}x} 
&\hspace*{10mm}  ({\eqq_1}) & \guard{=}\ttt,{{\nstimeeoolt}}
\\
 0&\!\!\!\!=\!\!\!\!& {\ddot{y}+\pmatch{\tension}y+g}
&\hspace*{10mm}  ({\eqq_2}) &  
\\
\gremph{0} &\!\!\!\!\gremph{=}\!\!\!\!& \gremph{{L^2}{-}(x^2{+}y^2)} 
& \hspace*{10mm}\gremph{(\straight_1)}  &   \mbox{\gremph{fact}}
\\
\remph{0} &\!\!\!\!\remph{=}\!\!\!\!& \remph{x\dot{x}{+}y\dot{y}} 
& \hspace*{10mm}\remph{(\dot{\straight_1})}  &   \mbox{\remph{disabled}}
\\
{0} &\!\!\!\!{=}\!\!\!\!& x\ddot{x}{+}\dot{x}^2{+}\dot{y}^2{+}y\pmatch{\ddot{y}} \hspace*{-10mm}& \hspace*{10mm}{(\ddot{\straight_1})}
\eea\right.
\eeqq
\caption{\sf \textbf{Cup-and-ball example:} Mode change array $A_0$. In the last column we point the \gremph{{facts}} and the \remph{disabled} conflicting equations; black equations are enabled. Triple $(e_1,e_2,\ddot{\straight_1})$ is structurally nonsingular, with a one-to-one matching $\pmatch{\cM}=\{(e_1,\pmatch{\ddot{x}}),(e_2,\pmatch{\tension}),(\ddot{k_1},\pmatch{\ddot{y}})\}$ between enabled equations and variables. 
}
\label{fdkghldkjhkjj}

 \beqq
A_1:\left\{\bea{rclcll}
 0&\!\!\!\!=\!\!\!\!& \pmatch{\ddot{x}}+{\tension}x 
& (\eqq_1) 
& \mbox{instant }\nstimeeoolt
\\
 0&\!\!\!\!=\!\!\!\!& \ddot{y}+\pmatch{\tension}y+g 
& (\eqq_2)  & 
\\
\gremph{0} &\!\!\!\!\gremph{=}\!\!\!\!& \gremph{{L^2}{-}(x^2{+}y^2)} 
& \gremph{(\straight_1)}  
& \mbox{\gremph{fact}}
\\
\remph{0} &\!\!\!\!\remph{=}\!\!\!\!& \remph{x\dot{x}{+}y\dot{y}} 
& \remph{(\dot{\straight_1})}    
& \mbox{\remph{disabled}}
\\
\remph{0} &\!\!\!\!\remph{=}\!\!\!\!&  \remph{x\ddot{x}{+}\dot{x}^2{+}\dot{y}^2{+}y\ddot{y}}
& \remph{(\ddot{\straight_1})}    
& \mbox{\remph{disabled}}
\\ 
[2mm]
0 &\!\!\!\!{=}\!\!\!\!&  \ddot{x}-\vsmall^{-1}(\pmatch{\postset{\dot{x}}}-\dot{x}) & (\euler{\ddot{x}}{})
\\ 
0 &\!\!\!\!{=}\!\!\!\!&  \pmatch{\ddot{y}}-\vsmall^{-1}({\postset{\dot{y}}}-\dot{y}) & (\euler{\ddot{y}}{})
\\ [2mm]
0&\!\!\!\!=\!\!\!\!& \postset{(\pmatch{\ddot{x}}+{\tension}x)}
& (\postset{\eqq_1}) 
& \mbox{instant }\nstimeeoolt{+}\vsmall
\\
 0&\!\!\!\!=\!\!\!\!& \postset{(\ddot{y}+\pmatch{\tension}y+g)}
& (\postset{\eqq_2})  & 
\\
\gremph{0} &\!\!\!\!\gremph{=}\!\!\!\!& \gremph{\postset{({L^2}{-}(x^2{+}y^2))} }
& \gremph{(\postset{\straight_1})}     
& \mbox{\gremph{fact}}
\\
{0} &\!\!\!\!{=}\!\!\!\!& {\postset{{(x\dot{x}{+}y\pmatch{\dot{y}})}}}
& {(\postset{\dot{\straight_1}})}  & 
\\
{0} &\!\!\!\!{=}\!\!\!\!&  \postset{({x\ddot{x}{+}\dot{x}^2{+}\dot{y}^2{+}y\pmatch{\ddot{y}}})}
& {(\postset{\ddot{\straight_1}})}   
\eea\right.
\eeqq
\caption{\sf \textbf{Cup-and-ball example:} Mode change array $A_1$. \gremph{{Facts}} and \remph{disabled} conflicting equations are pointed. Euler identities $(\euler{\ddot{x}}{}),(\euler{\ddot{y}}{})$, relating second derivatives to shifts of first derivatives, were added. The enabled subsystem (in black) is structurally nonsingular, with a one-to-one matching $\pmatch{\cM}$ between equations and variables, highlighted in \pmatch{blue}.
}
\label{kefygfkjfgjhgf}
 \[
A_1:\left\{\bea{rclc}
 0&\!\!\!\!=\!\!\!\!& \pmatch{\ddot{x}}+{\tension}x 
& (f_1) 
\\
 0&\!\!\!\!=\!\!\!\!& \ddot{y}+\pmatch{\tension}y+g 
& (f_2) 
\\ [2mm]
{0} &\!\!\!\!{=}\!\!\!\!& {\postset{{(x\dot{x}{+}y\pmatch{\dot{y}})}}}
& {(f_3)}  
\\ [2mm]
0 &\!\!\!\!{=}\!\!\!\!&  \ddot{x}-\vsmall^{-1}(\pmatch{\postset{\dot{x}}}-\dot{x}) & (f_4)
\\ 
0 &\!\!\!\!{=}\!\!\!\!&  \pmatch{\ddot{y}}-\vsmall^{-1}({\postset{\dot{y}}}-\dot{y}) & (f_5)
\eea\right.
\]
\caption{\sf \textbf{Cup-and-ball example:} The black subsystem of Fig.\,\ref{kefygfkjfgjhgf}, in which leading equations of instant $\postset{t}=t{+}\vsmall$ were omitted, and equations renumbered for convenience.}
\label{hgrtduhtraytrgfeads}
\end{figure}

\begin{ffinding}[facts]\rm
	\label{jdcfghafmhayfg} The equation $\gremph{(k_1)}$ is highlighted in \gremph{green} to indicate that it is satisfied up to $O(\vsmall)$, despite it involving only past variables; the reason for this is that the mode change was detected at instant $\preset{\nstimeeoolt}$ by the \emph{zero-crossing} event $0{=}\preset{(x^2{+}y^2{-}L^2)}$, hence, ${x^2{+}y^2{-}L^2}{\,=\,}O(\vsmall)$.	We call \emph{fact} an equation that involves past variables only and is satisfied up to an $O(\vsmall)$. Being nearly satisfied and involving no dependent variable, \emph{\textbf{facts can be ignored}}.\footnote{ \,The reader is kindly advised against inferring anything about the human world from this statement.}\eproof
\end{ffinding}
This is in contrast to the next equation $\remph{(\dot{k_1})}$: it also involves only past variables but it is not satisfied. \emph{We decide to disable it at this instant}, which amounts to postponing its consideration for a while. We are left with the three equations $(e_1,e_2,\ddot{k_1})$, which determine the leading variables $\ddot{x},\ddot{y},\tension$. Hence, array $A_0$ uniquely determines the leading variables from the values of past variables: Requirement~\ref{ksdjyhfgaskjfhgfkj}.\ref{ksjdfhsgkfjhg} is met. In contrast, Requirement~\ref{ksdjyhfgaskjfhgfkj}.\ref{mdfsjhfgvksjmhdgkj} is not met since consistency equation $\remph{(\dot{k_1})}$ was violated.

Let us extend array $A_0$ by one more instant, by considering array $A_1$, shown in Fig.\,\ref{kefygfkjfgjhgf}. Euler identities $(\euler{\ddot{x}}{})$ and $(\euler{\ddot{y}}{})$ were added, to make explicit the relations between $\ddot{x}$ and $\postset{\dot{x}}$, and  $\ddot{y}$ and $\postset{\dot{y}}$. Variables $x,y,\dot{x},\dot{y}$ are still determined by the previous mode. Equation $\gremph{(k_1)}$ remains a fact, and so does $\gremph{(\postset{k_1})}$, for the same reasons. Let us disable the conflicting equations $\remph{(\dot{k_1},\ddot{k_1})}$. The remaining system is structurally nonsingular, as evidenced by the perfect matching highlighted in \pmatch{blue}. Array $A_1$ determines the values of all its dependent variables knowing the past variables: Requirement~\ref{ksdjyhfgaskjfhgfkj}.\ref{ksjdfhsgkfjhg} is met. In addition, Requirement~\ref{ksdjyhfgaskjfhgfkj}.\ref{mdfsjhfgvksjmhdgkj} is now satisfied since all the consistency equations at instant $\postset{\nstimeeoolt}$ are enabled. 
\begin{ffinding}\rm
	 \label{jsdhcgfanhgfnh}
	We solved the hot restart of the new mode, in discretized setting, i.e. Step~\ref{jytrgdhtsgdyntgfv} of Approach~\ref{isduyfsegdiksuygu} was successful when using $A_1$.\eproof
\end{ffinding}	
 What happens when performing Step~\ref{hgrfdjtrgdhtfjhgsdrf}, i.e., when letting $\vsmall{\,:=\,}{0}$? Here, we only consider the black subsystem of Fig.\,\ref{kefygfkjfgjhgf}, which yields the system shown in Fig.\,\ref{hgrtduhtraytrgfeads}, where the same perfect matching is highlighted in \pmatch{blue}. 
\begin{ffinding}[trying $\vsmall{\,:=\,}{0}$]\rm
	\label{djhdgafkjhgf} Letting $\vsmall{\,:=\,}{0}$ in $A_1$ causes trouble, due to the occurrence of $\vsmall^{-1}$ in equations $(f_4,f_5)$. Multiplying these two equations by $\vsmall$ does not solve the problem either, since letting $\vsmall{\,:=\,}{0}$ in this case erases $\pmatch{\ddot{y}}$ from equation $(f_5)$, which makes $A_1$ singular.\eproof
\end{ffinding} 
Only one action can get rid of this difficulty: properly identifying impulsive variables and rescaling them.
\begin{notation}\rm	\label{kducjhgskjsyhg} 
In the sequel, we use the convention that $(f_i)$ denotes an equation of array $A_1$, whereas $f_i$ shall denote the corresponding function defined by its right-hand side.\eproof
\end{notation}
We propose a simple approach to identify and quantify impulsive behaviors:\footnote{What follows is an informal definition; the formalization will be given in Section~\ref{dfkjshlksgfskdjfmg}.}
\beq
\mbox{
\begin{minipage}{12cm}	
	\begin{enumerate}
		\item Say that variable $x$ has \emph{rescaling offset} $k$, written $\rescaling{x}{=}k$, if its value in the solution of system $A_1$ is $0(\vsmall^{-k})$; 
		\item \label{klsdjfvhsgklijug}
	$\rescaling{x}{=}k$ results in the rescaling $\rescaled{x}\eqdef\vsmall^k{\times}x$.
\end{enumerate}
\end{minipage}
}
\label{gtrfdkyjrgs}
\eeq
Focus on equation $(f_5)$.
Since $\dot{y}$ and $\postset{\dot{y}}$ are both non impulsive, equation $(f_5)$ in array $A_1$ implies $\rescaling{\ddot{y}}{=}1{+}\max(\rescaling{{\dot{y}}},\rescaling{\postset{\dot{y}}}){=}1$. In addition, we also associate, to function $f_5$, a rescaling offset indicating the magnitude order w.r.t. $\vsmall^{-1}$ of $f_5$ in the neighborhood of a solution of this equation. That is: 
\[
\rescaling{\!f_5}=\max\bigl(\rescaling{\ddot{y}},1{+}\max(\rescaling{{\dot{y}}},\rescaling{\postset{\dot{y}}})\bigr)=1\,.
\]
Performing this kind of reasoning for all the equations of array $A_1$ yields the following \emph{rescaling analysis:}
\[
\left\{\bea{rlllllllllll}
\rescaling{\ddot{x}} &\negesp=\negesp& \rescaling{\!f_1} &\negesp=\negesp& \max(\rescaling{\ddot{x}},\rescaling{\tension}) \\
\rescaling{\tension} &\negesp=\negesp& \rescaling{\!f_2} &\negesp=\negesp& \max(\rescaling{\ddot{y}},\rescaling{\tension}) \\
\rescaling{\postset{\dot{y}}} &\negesp=\negesp& \rescaling{\!f_3} &\negesp=\negesp& \max(\rescaling{\postset{\dot{x}}},\rescaling{\postset{\dot{y}}}) \\
1+\rescaling{\postset{\dot{x}}} &\negesp=\negesp& \rescaling{\!f_4} &\negesp=\negesp& \max(\rescaling{{\ddot{x}}},1+\rescaling{\postset{\dot{x}}}) \\
\rescaling{{\ddot{y}}} &\negesp=\negesp& \rescaling{\!f_5} &\negesp=\negesp& \max(\rescaling{{\ddot{y}}},1+\rescaling{\postset{\dot{y}}}) \\
\eea\right.
\]
We search for solutions of the rescaling analysis, satisfying the following requirement, expressing that state variables for restart should not be impulsive:
\begin{require}[on states]
	\label{sjhagcdfasjhngf}  State variables of the last instant possess rescaling offset zero.
\end{require}
The only solution making restart variables non-impulsive is the following:
\beq
\rescaling{\ddot{x}}=\rescaling{\tension}=\rescaling{\ddot{y}}= 1 = \rescaling{\!f_5}= \rescaling{\!f_2}= \rescaling{\!f_1}
\label{skksdjfhsgclugy}
\eeq
whereas other variables and functions have rescaling offset zero. Rescaling variables and equations using (\ref{gtrfdkyjrgs}), Step~\ref{klsdjfvhsgklijug}, yields:
\beq\bea{rclcrcrcl}
\rescaled{\ddot{x}} &\negesp\eqdef\negesp& \vsmall{\times}\ddot{x} &\negesp=\negesp& \postset{\dot{x}}{-}\dot{x} &;&
\rescaled{f_1} &\negesp\eqdef\negesp& \vsmall{\times}{f_1} 
\\ [1mm]
\rescaled{\ddot{y}} &\negesp\eqdef\negesp& \vsmall{\times}\ddot{y} &\negesp=\negesp& \postset{\dot{y}}{-}\dot{y}  &;&
\rescaled{f_5} &\negesp\eqdef\negesp& \vsmall{\times}{f_5} 
\\ [1mm]
\rescaled{\tension} &\negesp\eqdef\negesp& \vsmall{\times}\tension && &;&
\rescaled{f_2} &\negesp\eqdef\negesp& \vsmall{\times}{f_2} 
\eea \label{hgtgdfkjhglkj}
\eeq
By using (\ref{hgtgdfkjhglkj}), we have
\beq
\rescaled{A_1}:\left\{\bea{cclcll}
 0&\!\!\!\!=\!\!\!\!& 
\pmatch{\rescaled{\ddot{x}}} +{\rescaled{\tension}}x 
\\ [0mm]
 0&\!\!\!\!=\!\!\!\!& 
\rescaled{\ddot{y}} +\pmatch{\rescaled{\tension}}y
\\ [1mm]
{0} &\!\!\!\!{=}\!\!\!\!& {\postset{{(x\dot{x}{+}y\pmatch{\dot{y}})}}}
\\
[1mm]
0&\!\!\!\!=\!\!\!\!& \rescaled{\ddot{x}}-(\pmatch{\postset{\dot{x}}}-\dot{x})
\\
0&\!\!\!\!=\!\!\!\!& \pmatch{\rescaled{\ddot{y}}}-({\postset{\dot{y}}}-\dot{y})
\eea\right.
\label{skeufygsdkjyhdhg}
\eeq
\begin{ffinding}[retrying $\vsmall{\,:=\,}{0}$]\rm
	\label{kjdfhgsdkmjhg} At this point, letting $\vsmall{\,:=\,}{0}$ in $\rescaled{A_1}$ preserves regularity, structurally\,---\,actually, $\vsmall$ already disappeared from $\rescaled{A_1}$. \eproof
	\end{ffinding}
We are now ready to conclude. New positions were determined by the previous mode: $\postset{x}{\,=\,}x$ and $\postset{y}{\,=\,}y$, expressing that positions are continuous at mode change. The restart system is finally obtained by renaming, in (\ref{skeufygsdkjyhdhg}):
\begin{enumerate}
	\item the variables set by the previous mode $x,y,{\dot{x}},{\dot{y}}$, by the left-limits at mode change $x^-,y^-,\mmoins{\dot{x}},\mmoins{\dot{y}}$;
	\item the tail variables $\postset{x},\postset{y},\postset{\dot{x}},\postset{\dot{y}}$, by the restart values for the new mode $x^+,y^+,\pplus{\dot{x}},\pplus{\dot{y}}$; whereas
	\item other variables are auxiliary and are not renamed.
\end{enumerate}
The second renaming action is sensible only if the following requirement is satisfied by rescaled array (\ref{skeufygsdkjyhdhg}):
\begin{require}
	\label{hgfhfkjghm} Performing rescaling shall bring, in array $\rescaled{A_1}$, no variable attached to instants later than $t{+}\vsmall$.
\end{require} 
This requirement is indeed satisfied by rescaled array (\ref{skeufygsdkjyhdhg}). Finally, performing the above legitimate renaming yields the
\begin{equation}
\mbox{restart system}:\left\{\bea{ccll}
 0&\!\!\!\!=\!\!\!\!& \pmatch{\pplus{{x}}}-\mmoins{x}
\\ 
 0&\!\!\!\!=\!\!\!\!& \pmatch{\pplus{{y}}}-\mmoins{y}
\\ 
 0&\!\!\!\!=\!\!\!\!& \pmatch{\pplus{\dot{x}}}-\dot{x}^-+{\rescaled{\tension}}x^-
\\ 
 0&\!\!\!\!=\!\!\!\!& \pplus{\dot{y}}-\dot{y}^-+\pmatch{\rescaled{\tension}}y^-
\\ 
{0} &\!\!\!\!{=}\!\!\!\!& {x^-}\pplus{\dot{x}}{+}{y^-}\pmatch{\pplus{\dot{y}}}
\eea\right.
\label{keufygkuaygkuyg}
\end{equation}
Note that eliminating $\rescaled{\tension}$ from (\ref{keufygkuaygkuyg}) yields
\beq
\pplus{\dot{x}}\pplus{y}-\pplus{\dot{y}}\pplus{x} &=& \mmoins{\dot{x}}\mmoins{y}-\mmoins{\dot{y}}\mmoins{x}\,,  \label{dkfvjshdgfkjhy}
\eeq
which expresses the preservation of angular momentum, a physical invariant that was not explicitly specified in the original modeling: our approach ``discovered'' it. 
Also note that our solution does not give a meaning to variables $\ddot{x},\ddot{y},\tension$ beyond their status of being impulsive. Only the rescaled variable $\rescaled{\tension}$ is well defined. 
This approach is formalized and generalized in Sections~\ref{jashcgfackjhamgfjh} and~\ref{hgrfdkhgfcvkjmh}; before that, we shall proceed with background material.

\section{Background}
\label{hgfdjhgljhgb}
\label{jhgfvkjhhtrgfdhgfedbsv} \label{jhtfdjghfagfdv}
In this section, we recall our background, consisting of the structural analysis of DAE systems.
Throughout this section we consider DAE systems 
\beq
S&:&f_i(\mbox{the } x_j \mbox{'s and their derivatives}) = 0
\label{sepguiohp}
\eeq
where $x_1,\dots,x_p$ are the variables and $f_1{=}0,\dots,f_n{=}0$ are the equations. The functions $f_i:\bR^p\ra\bR$ are \emph{smooth}, i.e., of class $\mathcal{C}^m$ for sufficiently large integer $m$. System (\ref{sepguiohp}) is called \emph{static} (also called \emph{algebraic}) if no derivatives are involved. 
\begin{itemize}
	\item Call \emph{leading variables} of System~(\ref{sepguiohp}) the $d_j$-th derivatives\footnote{The notation $\pprime{k}{x}$ is adopted throughout this paper, instead of $x^{(k)}$, for the $k$-th derivative of $x$.} $\pprime{d_j}{x_j}$ for $j{=}1,\dots,p$, where $d_j$ is the maximal differentiation degree of variable $x_j$ through $f_1{=}0,\dotsc,f_n{=}0$. 
	\item Remaining variables $\pprime{m}{x_j}$ for $j=1,\dots,m$ and $0{\leq} m{<}d_j$ are called \emph{state variables.} 
	\item For $f{=}0$ an equation and $x$ a variable of $S$, we denote by $\sigma_{\!\edge}$ the maximal differentiation degree of $x$ in $f$; by convention, $\sigma_{\!\edge}=-\infty$ if $x$ does not occur in $f$.
\end{itemize}
System (\ref{sepguiohp}) is denoted by 
\beq
S=(F,X), &\mbox{where}& 
F =\bigl\{\,f_i\mid \, i=1,\dots,n\bigr\}, \mbox{ and }
X = \bigl\{{x_j}\mid j=1,\dots,p\bigr\}
\label{hgrdhtrjgfjynhtdg}
\eeq
In the following, \emph{we regard $X$ and $F$ as sets of variables and functions. In particular, the conjunction of systems $S_1$ and $S_2$ is  $(F_1{\cup} F_2,X_1{\cup} X_2)$}, which we write also $S_1{\cup} S_2$. 
\begin{convention}\rm
	\label{htjfgjhtrgfd} 
By abuse of language, we will identify the set of functions $F$ and the set of equations $F{=}0$ it defines. Hence, we will often call ``equation'' an $f\in{F}$.\eproof
\end{convention}
\begin{ccomment}\rm
	\label{skdjcsghcksdujgk} 
We can alternatively regard $S$ in (\ref{sepguiohp}) as a system of static equations with the leading variables as dependent variables (unknowns). The state variables are free variables, however subject to certain consistency conditions, as we shall see in the  Section~\ref{nhmgfcvkjh} on the \sigmamethod.\eproof
\end{ccomment}

\subsection{Structural analysis of systems of equations}

\label{jhtgfkjhseghfaerv}
In this section, we consider the subcase of \emph{static systems,} i.e., systems of the form (\ref{sepguiohp}) involving no derivative. We will need to consider additional \emph{free} variables collected in $Y$, whose value is set by some environment in the system $F(X,Y){=}0$. By contrast, we call \emph{dependent} variables the variables collected in $X$, meaning that $F(X,Y){=}0$ is to be solved for $X$.

The generic term of \emph{structural analysis} refers to any analysis of a system of equations that relies on its incidence graph only\,---\,such analyses are therefore much cheaper than numerical ones. 
To $S=(F,X)$ (since $Y$ is irrelevant here), we associate its \emph{incidence graph}, which is a nondirected bipartite graph $\cG_S=(F{\cup}X,E)$, where $(f,x)\in F{\times}X$ is an edge of $\cG_S$ if and only if variable $x$ occurs in equation $f$. 

A \emph{matching} is a subset $\matching\subseteq{\Edges}$ involving at most once each equation and variable of $S$. Matching $\matching$ is \emph{perfect} (we also say \emph{complete}) if all equations and all variables of $S$ are involved (implying $p{=}n$). Matching $\cM$ is \emph{variable-complete} (respectively \emph{equation-complete}) if all variables (respectively equations) are involved in it. 
A vertex of $\cG_S$ is called \emph{unmatched} in $\matching$ if it does not occur in $\cM$.
\begin{definition}
	Say that system $S$ is \emph{structurally nonsingular} (or \emph{structurally regular}) if its incidence graph $\cG_S$ possesses a \emph{perfect matching} $\matching$.
 \label{sjduygvjyst}
\end{definition}
Structural nonsingularity is a necessary (and generically sufficient) condition for the existence and uniqueness of solutions, see~\cite{benveniste:hal-03104030}. This condition is widely used in all software dealing with very large, but sparse, systems of equations. Structural analysis only deals with the incidence graph of a system of equations. Hence, in structural analysis, the following rule is enforced:
\beq\mbox{
\begin{minipage}{14cm}
	 if variable $x$ occurs with differentiation degree $m$ in equation $f{=}0$, then it occurs with differentiation degree $m{+}1$ in equation $\dot{f}{=}0$.
\end{minipage}
}  \label{jkdfshfgsdjhfg}
\eeq
For $S$ a (possibly nonsquare, i.e., $p{\neq}{n}$) system of equations, the \emph{Dulmage-Mendelsohn (DM) decomposition}~$\mbox{\cite{PothenF90,DBLP:journals/arc/BenvenisteCM20,benveniste:hal-02521747}}$ of $\cG_S$ uniquely partitions $S$ into three pairwise disjoint subsystems: 
\beq
S=S^{\sf o}\cup S^{\sf r}\cup S^{\sf u}\,, \label{eujhydwtye}
\eeq
 where $S^{\sf r}$ is \emph{structurally regular}, whereas $S^{\sf o}$ and $S^{\sf u}$ are the structurally \emph{over-determined} and \emph{under-determined} parts of $S$.
System $S$ is structurally regular if and only if its Dulmage-Mendelsohn decomposition yields $S^{\sf o}=S^{\sf u}=\emptyset$. 
\begin{notation}\rm
\label{htgfkjthndfj}
If $\cM\subseteq{F}{\times}X$ is a matching, the notation $f\in\cM$ means that there exists $x\in X$ such that $(f,x)\in\cM$; the corresponding meaning holds for the notation $x\in\cM$.\eproof
\end{notation}

\subsection{The \sigmamethod\ for DAEs}
\label{nhmgfcvkjh} 
Here we consider square DAE systems, meaning that $p{=}n$ in (\ref{sepguiohp}). We recall the so-called \emph{index reduction}, by which suitably differentiating the different equations of a DAE system makes it ODE-like. Suppose we have a solution to the following problem:
\beq
\mbox{
\begin{minipage}{14cm}
	Find a perfect matching $\cM$ for $\cG_S$ and integer valued \emph{equation offsets} $\{c_{\!f}\mid{f\in{F}}\}$ and \emph{variable offsets} $\{d_{x}\mid{x\in{X}}\}$,  satisfying the following conditions:
\end{minipage}
}
	\label{liftuerhpituhepu8} 
\eeq
\[\bea{rcl}
d_{x}-c_{f} &\!\!\!\geq\!\!\!& \weight{\edge} \; \mbox{ with equality if } (\edge)\in\cM
 \\
c_{f} &\!\!\!\geq\!\!\!& 0
\eea\]
Following Comment~\ref{skdjcsghcksdujgk}, we denote by $\indexreduced{S}$ the static system collecting the leading equations, i.e., such that:
\beq
\indexreduced{S}=\bigl(\indexreduced{F},\indexreduced{X}\bigr)&\mbox{where}& \left\{\!\!\!\!\bea{rcl}
\indexreduced{F}&\!\!\!\!=\!\!\!\!&\{f^{\prime c_{\!f}}\mid f{\in}F\} \\
\indexreduced{X}&\!\!\!\!=\!\!\!\!&\{x^{\prime d_{x}}\mid x{\in}X\}
\eea\right.  \label{kfjhvbkdfjhb}
\eeq
and we call it the \emph{index reduced system}. System $\indexreduced{S}$ has the following properties: inequality $d_{x}\geq c_{f}+\weight{\edge}$ holds for each $x\in{X}$, and $d_{x}=c_{f}+\weight{\edge}$ holds for the unique $f$ such that  $(\edge)\in\cM$. Consequently, the set of pairs 
\beq
\indexreduced{\cM} &\eqdef& \left\{
\left.
(\pprime{c_{f}}{f},\pprime{d_{x}}{x}) 
\;\right|
(f,x)\in\cM
\right\}
\label{dskfjhsfbhjn}
\eeq
 is a perfect matching for the system $\indexreduced{S}$, seen as a system of static equations having $\pprime{d_{x}}{x}$ as dependent variables, hence this system is structurally nonsingular, see Definition~\ref{sjduygvjyst}. Hence, having reduced the index of $S$ makes it ODE-like: $\indexreduced{S}$ (generically) uniquely determines the values of leading variables, knowing the values of state variables (belonging to $\consistency{X}$).\footnote{The very elegant \sigmamethod, proposed by J.~Pryce in 2001~\cite{Pryce01} translates Problem~(\ref{liftuerhpituhepu8}) into a pair of primal/dual Linear Programs, of which the primal determines $\cM$ and  the dual determines the offsets. This proves that variable and equation offsets are independent of the particular choice for $\cM$. Furthermore, it is proven that offsets can be selected pointwise minimal, and that such a minimal solution is unique. Index reduction is performed by all major DAE based modeling tools, by using the \sigmamethod\ or the original method by C.~Pantelides~\cite{Pantelides1988}. See~\cite{DBLP:journals/arc/BenvenisteCM20,benveniste:hal-02521747} for omitted details.} As announced in footnote~\ref{skdjhsgksujg}, the consideration of latent equations governs how many times each $f\in{F}$ sould be differentiable.

The static system 
\beq
\consistency{S}=\bigl(\consistency{F},\consistency{X}\bigr)&\!\!\!\!\mbox{where}\!\!\!\!& \left\{\!\!\!\!\bea{rcl}
\consistency{F}&\!\!\!\!=\!\!\!\!&\{\pprime{l}{f}\mid f{\in}F,0{\leq} l{<}c_{\!f}\} \\
\consistency{X}&\!\!\!\!=\!\!\!\!&\{\pprime{k}{x}\mid x{\in}X,0{\leq} k{<}d_{x}\}
\eea\right.  \label{dlfgukjhsldkjfh}
\eeq
defines the \emph{consistency conditions.} System $\consistency{S}$ sets the constraints that must be satisfied by any consistent initial condition for the considered DAE~\cite{Pantelides1988}. Assuming that the \sigmamethod\ has a solution, $\consistency{S}$ possesses no overdetermined subsystem; hence, consistent valuations for the state variables exist, and each of them determines a value for the leading variables. Actually, 
\beq
\consistency{\cM} &\eqdef& \left\{
(\pprime{m}{f},\pprime{d}{x}) 
\,\left|\bea{l}
f\in{F} \\ 0\leq m<c_{\!f} \\ d=d_x{-}c_{\!f}{+}m
\eea\negesp\right.\right\}
\label{kujyfglukhoih}
\eeq
defines an equation-complete matching for $\consistency{S}$.

Finally, the static system
\beq
\complete{S}=(\complete{F},\complete{X}\bigr)&\eqdef&\bigl(\indexreduced{F}\cup\consistency{F},\indexreduced{X}\cup\consistency{X}\bigr)
\label{ldfkgjhdlkufgh}
\eeq
is called the \emph{completion} of DAE $S$. Referring to the cup-and-ball example, (\ref{htdfkuytdshgd}) shows the decomposition $\complete{S}=\indexreduced{S}\cup\consistency{S}$.
\begin{ccomment}\rm 
	\label{hgrfdgtbfd} Systems $\indexreduced{S},\consistency{S},$ and $\complete{S}$ are all static, meaning that $x,\dot{x},\dots,\pprime{k}{x},\dots$ are seen as \emph{different variables,} not as successive derivatives of the same variable $x$---they are called \emph{dummy derivatives} in the DAE litterature~\cite{DummyDerivatives1993}.\eproof
\end{ccomment}

\subsection{Implicit function theorem}
\label{skdjfhsgkldjfgh}
This section reports useful results on the structural analysis of static systems of equations involving a small parameter. More precisely, we consider a system 
\beq
F(X,Z,\vsmall)=0\,, \label{sajxchsagfvckghf}
\eeq
where $F$ is a finite set of $\cC^1$-functions, $X$ and $Z$ are the sets of dependent and free variables, and $\vsmall{>}0$ is a small parameter. 
Recall the following result, which is a rephrazing of the Implicit Function Theorem (see, e.g., Theorem 10.2.2 in~\cite{DieudonneEA1}):
\smallskip
\begin{proposition}[Implicit Function Theorem]
	\label{solifujhsgldikugh} Assume that the valuation $(x,z)$ for the pair $(X,Z)$ satisfies $F(x,z,0)=0$ and the Jacobian matrix $\partial{F}/\partial{X}$ at $(x,z,0)$ is nonsingular. Then:
	\begin{enumerate}
		\item \label{jytghflkjhlk} There exists a neighborhood of $(z,0)$ in which, for every pair $(\hat{z},\vsmall)$, there exists $\hat{x}$ such that $F(\hat{x},\hat{z},\vsmall)=0$.
		\item \label{sdkfhskdjg} For any sequence $(z_n,\vsmall_n)$ converging to $(z,0)$, the solution $x_n$ of system $F(X,z_n,\vsmall_n){=}0$ converges to $x$.
	\end{enumerate}
\end{proposition}

\section{Problem setting: hot restart of mode changes}
\label{jashcgfackjhamgfjh}
In this section we formalize the hot restart problem, that is, how to properly generate deterministic restart conditions for a new mode, knowing the previous mode. We consider mode changes in isolation.
Cascades of transient modes (of zero duration, as in the cup-and-ball with elastic impact) and Zeno phenomena (in which events of mode change accumulate in a finite duration of time) are not addressed.
\begin{definition}[mode change]
	\label{sdlifjvghsaldukfjgh} A \emph{mode change} is defined as a triple $(S^-,\changetime,S)$, where $S^-=(F^-,X^-)$ and $S=(F,X)$ are square\footnote{With as many equations as dependent variables.} DAE systems of the form $(\ref{sepguiohp})$ and $\changetime{\,\in\,}\bR_+$. $S^-$ is the \emph{previous mode}, $\changetime$ is the instant of mode change, and $S$ is the \emph{new mode}. 
\end{definition}
Without loss of generality, we can assume that: 
\begin{enumerate}
	\item $F^-$ is index-reduced, i.e., $F^-=\indexreduced{F^-}$; and  
	\item $F$ is completed with its latent equations, i.e., $F=\complete{F}$. 
\end{enumerate}
Recall that $\complete{X^-}$ (resp., $\complete{X}$) denotes the set collecting the leading and state variables of the previous (resp., new) mode. Then, we consider the subset $X^+\subseteq\complete{X}$ of state variables of the new mode. The hot restart consists in finding values for all of them.

As announced before, mode changes can be time- or state-based. More precisely, the following assumption holds throughout this work regarding mode changes:
\begin{assumption}[zero-crossing]
	\label{hytdrfjytrjytr} 
Mode change $(S^-,\changetime,S)$ is caused by a \emph{zero-crossing}, i.e., the crossing of zero from below by a $C^1$-function $g(\complete{X^-})$. Formally, setting $g(t)\eqdef g(\complete{X^-}(t))$, there exists some duration $\delta{>}0$ such that, if the dynamics of $S^-$ is enforced everywhere, then, $g(t)<0$ holds for $t\in(\changetime{-}\delta,\changetime)$ and $g(t)>0$ holds for $t\in(\changetime,\changetime{+}\delta)$.
\end{assumption}
\begin{problem}[hot restart]
	\label{dsjkchgfckajycgk} 
	For $(S^-,\changetime,S)$ a mode change, construct its \emph{hot restart} system, which is a system of equations $R(\complete{X^-},X^+)=0$, relating the left limits of variables and their derivatives at the mode change to the state variables for the new mode, satisfying the following conditions:
	\begin{enumerate}
	\item \label{skdjfvhsdgkjfh} System $R(\complete{X^-},X^+)$ should be \emph{deterministic}, meaning that the values for $X^+$ should be uniquely determined from the values for $\complete{X^-}$;
		\item \label{sjdfghvcsakdujfcyg} The restart should be \emph{consistent}, meaning that the values for $X^+$ should satisfy the consistency conditions $\consistency{F}$;
	\item	\label{kaujfyhgakujyhg} After proper rescaling to compensate for possible impulsive behaviors, \emph{every invariant dynamics \emph{(i.e., common to previous and new mode)} should be satisfied at the mode change.}
	\end{enumerate}
\end{problem}
Condition~\ref{skdjfvhsdgkjfh} reflects Requirement~\ref{ksdjyhfgaskjfhgfkj}.\ref{ksjdfhsgkfjhg}; condition~\ref{sjdfghvcsakdujfcyg} reflects Requirement~\ref{ksdjyhfgaskjfhgfkj}.\ref{mdfsjhfgvksjmhdgkj}. Condition~\ref{kaujfyhgakujyhg} is partly informal, as ``impulsive'' and ``rescaling'' are not formally defined yet. This condition will be formalized and studied in Section~\ref{kaerjfhgkaewrjyfhg}. It is vacuously true if no invariant dynamics exists; otherwise, it indicates that hot restart is a strict refinement of cold start.

Our solution for hot restart will actually take the form of a system of equations 
$$\overline{R}(\complete{X^-},Z,X^+)=0$$
 where $Z$ collects auxiliary variables. Eliminating $Z$ from $\overline{R}$ will yield a system 
$$R(\complete{X^-},X^+)=0$$
 satisfying Requirements~\ref{skdjfvhsdgkjfh},~\ref{sjdfghvcsakdujfcyg}, and~\ref{kaujfyhgakujyhg}. All of this was illustrated in the cup-and-ball example by (\ref{keufygkuaygkuyg},\ref{dkfvjshdgfkjhy}), where $Z=\{\rescaled{\tension}\}$.

\section{Mode change arrays}
\label{hgrfdkhgfcvkjmh}
In this section, we introduce mode change arrays, as the main data structure of our approach. We first formalize the discrete time setting we introduced in Approach~\ref{jhgfdytrhgf}. For a fixed \emph{time step} $\vsmall{\,>\,}0$, we use the time line 
\beq
\bT \eqdef \{n\vsmall\mid n\in\bN\} \, ,  \label{jhtfdkujughkjn}
\eeq
in a neighborhood of the mode change instant, where $\bN$ denotes the set of nonnegative integers. Assuming $\vsmall$ small enough and with reference to Assumption~\ref{hytdrfjytrjytr}, we redefine the \emph{instant of detection of the mode change} as being the first instant belonging to $\bT\cap(\changetime{-}\delta,\changetime{+}\delta)$, such that the zero-crossing function $g(\complete{X^-})$ is positive. To ensure that mode change detection remains causal, the new mode begins at instant $\changetime{+}\vsmall$; by abuse of notation, we redefine $\changetime$ to be this instant. To summarize,
\beq
\mbox{
\begin{minipage}{14cm}
	 in our discretized setting, the zero-crossing is detected at instant $\changetime{-}\vsmall$ and the new mode is effective at instant $\changetime$ (see Fig.\,\ref{seldoifuhsglsig}).
\end{minipage}
} \label{sjdcytsafcjsyutf}
\eeq
\begin{figure}[!h]
\centerline{\includegraphics[width=0.6\textwidth]{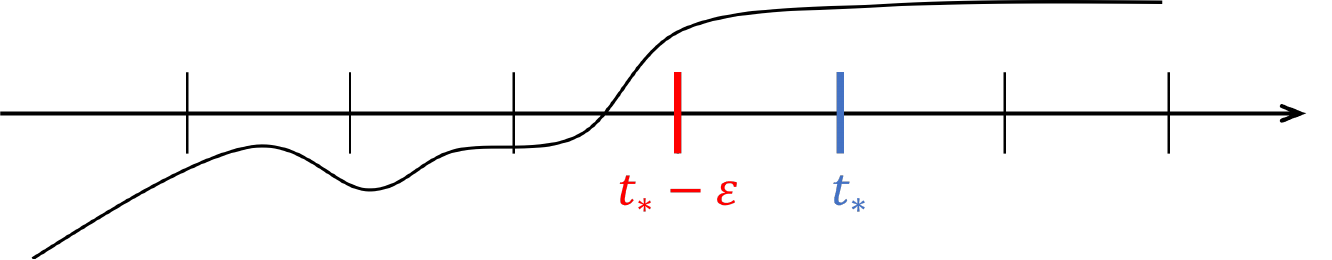}}
\caption{\sf Zero-crossing in discrete time. The detection instant is in {\color{red}red} and the first instant of the new mode is in {\color{blue}blue}.}
	\label{seldoifuhsglsig} 
\end{figure}
\subsection{Shifting and differentiating}
\label{jghtsdhgfsdhfgsd}
Since we are interested in the system behavior around the instant of mode change, we take $\changetime$ as time reference: 
\clearpage
\begin{notation}[variables and functions]\rm \
	\label{dksfjvhsgkvj} 
	\begin{enumerate}
		\item We assume an underlying set $\BaseVars$ of \emph{base variables}. 
		For $x\in\BaseVars$ and integers $k\in\bZ,m\in\bN$:
\beq\bea{rcl}
x &\mbox{shall denote}& x(\changetime) \\
\ppostset{k}{x} &\mbox{shall denote}& x(\changetime{+}k\vsmall) \\
\pprime{m}{x} &\mbox{shall denote}& \pprime{m}{x}(\changetime) \\
\pprimeppostset{m}{k}{x} &\mbox{shall denote}& \pprime{m}{x}(\changetime{+}k\vsmall) 
\eea \label{kduhgvskdug}\eeq
and we denote by $\allVars$ the set of all $\pprimeppostset{m}{k}{x}$. Notations (\ref{kduhgvskdug}) extend to $\allVars$. 
\item 
We also assume an underlying set $\baseFunctions$ of \emph{base functions}. For $f\in\baseFunctions$ we define $\pprimeppostset{m}{k}{f}$ by $\pprimeppostset{m}{k}{f}(X)\eqdef f(\pprimeppostset{m}{k}{X})$ and $\allFunctions$ follows accordingly.
\item \label{kfdjshfgdskmjhg} 
For $S$ a static system of equations over $\allVars$, we denote by $\Vars_S$ the set of variables occurring in $S$. When $S=(\{f\},X)$, we simply write $\Vars_{\!f}$.\eproof
\end{enumerate}
\end{notation}
Following step~\ref{jhgfdytrhgf} of Approach~\ref{isduyfsegdiksuygu}, derivatives are discretized by enforcing, for every $x\in\allVars$, the \emph{Euler identity}
\beq
\dot{x} =\vsmall^{-1}(\postset{x}-x)\,. \label{skdufhgwkuy}
\eeq
 \begin{notation}\rm
	 \label{kjuygfkjhjglkj}
 For $y\in\BaseVars$ and  $g\in\baseFunctions$, we define: 
\[\bea{rclcl}
x=\pprimeppostset{m}{k}{y}\in\allVars &\Ra&(\yfun{x},\mfun{x},\kfun{x}) &\negesp\eqdef\negesp& (y,m,k) \\
f=\pprimeppostset{m}{k}{g}\in\allFunctions&\Ra&
(\gfun{f},\mfun{f},\kfun{f}) &\negesp\eqdef\negesp& (g,m,k) \hspace*{5mm}\mbox{\eproof}
\eea\]
\end{notation}
\begin{definition} 
\label{jhtgfjhgrsgdfk} Let $\sim\;\subseteq\;\allVars{\times}\allVars$ be defined by:
\beqq
x_1 \sim x_2 &\mbox{iff}&
\yfun{x_1}=\yfun{x_2} \mbox{ and } \sizesim{x_1}=\sizesim{x_2}\,,
\eeqq
where $\sizesim{x}=\mfun{x}{+}\kfun{x}$ is the \emph{total degree} of $x$. Define the partial order $\preceq$ on $\allVars$ by $x_1\preceq x_2$ iff $\yfun{x_1}=\yfun{x_2}$ and $\sizesim{x_1}\leq\sizesim{x_2}$, and write $\prec~\eqdef~\preceq\land\neq$.
\end{definition}  
As mentioned in Section~\ref{jhtfdjghfagfdv}, 
\beq
\mbox{
\begin{minipage}{14cm}
	 DAEs are seen as systems of static equations over base variables, their derivatives, and their shifts, i.e., as \emph{systems of static equations over $\allVars$}.
\end{minipage}
} \label{hgrshtrsabrefgf}
\eeq
We now investigate some consequences of identity~(\ref{skdufhgwkuy}). Repeatedly applying this identity yields the \emph{Euler identity}
\beq
\pprime{n}{y} 
=\vsmall^{-n}{\sum_{i=0}^{n}}\left(\negesp\bea{c}n \\ i\eea\negesp\right)(-1)^i\ppostset{(n-i)}{y}\,, 
&\mbox{where}&\left(\negesp\bea{c}n \\ i\eea\negesp\right) = \frac{n!}{i!(n-i)!} \label{grsfbgfbdnbv}
\eeq
is the binomial coefficient. 

Consider two variables $x{\sim}z$. We have $x{=}\pprimeppostset{m_x}{k_x}{y}$ and $z{=}\pprimeppostset{m_z}{k_z}{y}$, where $m_x{+}k_x=m_z{+}k_z$. If, in addition, $x{\neq}z$ holds, then two cases can occur: either $m_x{>}m_z$, or $m_x{<}m_z$. We assume $m_x{>}m_z$, set $n=m_x{-}m_z=k_z{-}k_x>0$, and apply (\ref{grsfbgfbdnbv}) with this pair $(y,n)$. Differentiating $m_z$-times and shifting $k_x$-times the two sides of (\ref{grsfbgfbdnbv}) yields the (generalized) \emph{Euler identity}
\beq
0&\negesp=\negesp&x-\vsmall^{-n}\sum_{i=0}^{n}\left(\!\!\!\bea{c}n \\ i\eea\!\!\!\right)(-1)^i\ppostset{(-i)}{z} 
\mbox{, denoted by}
\nonumber
\\
0&\negesp=\negesp&\neuler{n}{x}{z}(x,z,U), \mbox{ where }U=\bigl(\ppostset{(-i)}{z}\bigr)_{i{=}1,\dots,n}
\label{htgrsfdjyrnwsefgdjt} \label{utrwshjyhtrf}
\eeq
where tuple $U$ collects the variables that are $\prec x$. 
\begin{convention}\rm 
	\label{cjdhgfcjsdhgf} We write $\neuler{n}{x}{z}$ when the arguments need not be mentioned, or even $\euler{x}{z}$ when only $x$ and $z$ matter.\eproof
\end{convention}
We summarize the above discussion in the following lemma, where Notations~\ref{kjuygfkjhjglkj} are used:
\begin{lemma}
	\label{ujytsagrehd} Any two different variables $x\sim z$ such that $\mfun{x}>\mfun{z}$ are related by  Euler identity $\euler{x}{z}$.
\end{lemma}
\begin{definition}[$\sim$\,--\,closed system]
	\label{uyytyfgkutfjhh} Say that static system $S{=}(F,X)$, where $X{\subseteq}\allVars$, is \mbox{\emph{$\sim$\,--\,closed}} if it includes all the Euler identities $\euler{x}{z}$ for any two $x, z {\in}\Vars_S$ such that $x{\sim}z$ and $\mfun{x}{>}\mfun{z}$. For any such system $S=(F,X)$, its \emph{$\sim$\,--\,closure},  denoted by $\simclosure{S}$, is defined by:
\beq
\simclosure{S} = \bigcap\left\{ T
\,\left|\negesp\;\bea{l}
T \supseteq S \\ 
\forall x,z{\in}\Vars_T,\; 
\left[\negesp\bea{c} x{\sim} z , x{\neq}z \\ \mfun{x}{>}\mfun{z}
\eea\negesp\right]\Ra\;\euler{x}{z}\in{T}
\eea\negesp\right.\right\}
\label{htrgfjrgdfhtrgdffn}
\eeq
where $T$ ranges over the set of all extensions of system $S$.

\end{definition}
The following result justifies the above definition:
\begin{lemma}
	\label{hngtedgrefsdg} $\simclosure{S}$ always exists and is unique. It can be obtained by performing the following recursion, until fixpoint:
	\begin{enumerate}
		\item Initialization: $F_0\gets F,X_0\gets X$;
		\item While $X_{n}{\supset} X_{n-1}$, update $(F_{n+1},X_{n+1}){\gets}(F_n,X_n)$ by
		\begin{itemize}
			\item adding to $F_n$ every Euler identity $\euler{x}{z}$ such that 
			$x,z\in X_n$, $x\sim z$, $x\neq z$, and $\mfun{x}{>}\mfun{z}$;
			\item adding to $X_n$ every $v$ involved in $\euler{x}{z}$.
		\end{itemize}
		\end{enumerate}
\end{lemma}
\begin{proof}
	The recursion terminates in finitely many steps, since the total degree of the added variables is strictly decreasing. 	The limit of the recursion satisfies the fixpoint equation (\ref{htrgfjrgdfhtrgdffn}).
\end{proof}
%
%
\begin{convention}\rm 
	\label{skdjhasdgfckjsadhg} Unless the joint consideration of $S$ and its $\sim$\,--\,closure $\simclosure{S}$ is needed,  superscript ``$\;\simclosure{~}\;$'' shall be omitted.\eproof
\end{convention}

\subsection{Mode change arrays and facts}
\label{geadshgshdjrgnsf}
In this section, we start addressing Step~\ref{jytrgdhtsgdyntgfv} of Approach~\ref{isduyfsegdiksuygu}, for mode change $(S^-,\changetime,S)$ following Definition~\ref{sdlifjvghsaldukfjgh}.
We use the notations from Section~\ref{nhmgfcvkjh}: for $F$ a DAE system, $\indexreduced{F},\consistency{F}$, and $\complete{F}$ refer to the ``index-reduced'', ``consistency'', and ``completion'' of $F$ once index reduction has been performed. When considering mode changes, we use notations (\ref{sjdcytsafcjsyutf},\ref{kduhgvskdug}).
\begin{notation}\rm 
	\label{sdolifkjshgikldu} 
	Let $S=(F,X)$ be any DAE system, and let $\changetime$ be a time origin. Then, for $k\in\bZ$, $(\ppostset{k}{F},\ppostset{k}{X})$ denotes the dynamics of the mode $S$ snapshot at instant $\changetime{+}k\vsmall$.\eproof
\end{notation}
For example, if function $f:x{+}\dot{y}$ belongs to $F$, then $f$ refers to $x(\changetime){+}\dot{y}(\changetime)$, and $\ppostset{k}{f}$ refers to $x(\changetime{+}k\vsmall){+}\dot{y}(\changetime{+}k\vsmall)$.
\begin{convention}\rm
	\label{slrikfgsjhfglsiugh} 
	Throughout this section, arrays are implicitly $\sim$\,--\,closed by completing them with the Euler identities (\ref{htgrsfdjyrnwsefgdjt}) following Lemmas~\ref{ujytsagrehd} and~\ref{hngtedgrefsdg}, and Convention~\ref{skdjhasdgfckjsadhg} applies.\eproof
\end{convention}
To the new mode $S{=}(F,X)$ we associate the infinite array $A$ of equations obtained by stacking successive shifts of the completion $\complete{F}$:	
	\beq
A \eqdef \bigcup_{k\geq 0}\ppostset{k}{\complete{F}} 
= \left\{\left.
\pprimeppostset{m}{k}{f} \right|
f\in{F},0\leq m\leq c_{\!f}, k\geq 0
\right\}.
\label{fdkjhsfkjhhj} 
\label{skdujcgskdjgyvckuj} 
\eeq
The set of \emph{past variables} of $A$ is its subset of variables whose value is determined by the previous mode:
\beq
\bea{rcl}
\Vars^- \eqdef \bigl\{
x\in\Vars_A
\left| \;
\exists x^-\in\bigcup_{k<0}\ppostset{k}{(X^-)} \mbox{~ s.t. ~} x^-{\sim\,} x
\right.
\bigr\}\,.
\eea
\label{htdjhtdmhg}
\eeq
The previous mode $S^-$ acts on the new mode $S$ only through the set $\Vars^-$ of past variables. 
To illustrate past variables, if $\pprime{m}{y}{\in}{X^-}$, then $x=\pprimeppostset{m}{(-k)}{y}\sim\ppostset{(m-k)}{y}\in\Vars^-$ whenever $0{<}k{\leq}m$. 

Referring to (\ref{htdjhtdmhg}), since $x^-{\sim\,} x$, $x$ is computed from $x^-$ via the Euler identity relating them (Lemma~\ref{ujytsagrehd}). Hence, the set of dependent variables of $A$ is
\beq
\dependent{\Vars_{A}} &\eqdef& \Vars_{A} \setminus \Vars^-\,,    \label{kljfghbkdjhb}
\eeq
where $\Vars_{A}$ is the set of \emph{all} the variables involved in $A$ (both leading and state variables).  The reason for taking the dependent variables of $A$ as in (\ref{kljfghbkdjhb}) is that we follow the causality principle that the past cannot be undone. Variables that were set by the previous mode cannot be modified by the new mode, thus should not be considered as dependent variables of $A$. Besides this, all the variables involved in $A$ need to be evaluated when performing the mode change (including the state variables).

At mode changes, additional equations hold, due to zero-crossings (Assumption~\ref{hytdrfjytrjytr}):
\begin{lemma}
	\label{jhgffdhghkjh} 
If the mode change is caused by a zero-crossing of $g(X^-)$, where the tuple $X^-$ of variables possesses continuous trajectories in the previous mode, then $g(X^-){=}0$ holds, up to an $O(\vsmall)$, at instants $\nstime,\dots,{\nstime{+}\ell\vsmall}$, where $\ell$ is the largest integer such that $\ppostset{\ell}{(X^-)}\subseteq\Vars^-$.
\end{lemma}
 \begin{proof}
	 The continuity of trajectories of  $X^-$ within the previous mode and the smoothness of function $g$, together imply that the value of $g(\complete{X^-})$ changes by at most $O(\vsmall)$ during a finite number of discrete time steps. 
 \end{proof}
\begin{definition}[facts]
	\label{dkjhgdkhgl} For $g(\complete{X^-})$ a zero-crossing function:
	\begin{enumerate}
	\item Call \emph{root fact} any equation that is entailed, up to $O(\vsmall)$, by $g(X^-({t_*})){=}0$;	the set $\facts{}^-$ of all root facts is independent from the new mode;
	\item With $g$ a root fact and $\ell$ as in Lemma~\ref{jhgffdhghkjh}, call \emph{fact} any equation of $A$ that is entailed, up to $O(\vsmall)$, by the system
	\[
g(X^-){=}0\;,\;\dots\;,\;g\ppostset{\ell}{(X^-)}{=}0\;.
\]
The set of all facts is denoted by $\facts{A}$.
\end{enumerate}
\end{definition}
Since facts hold true up to an $O(\vsmall)$ and do not involve dependent variables, we set the following convention:
\begin{convention}\rm
	\label{conv:remove-facts} In the sequel, facts are removed from $A$, i.e.,
$A$ is replaced by $A \,\setminus\, \facts{A}$.\eproof
\end{convention}
We continue addressing Step~\ref{jytrgdhtsgdyntgfv} of Approach~\ref{isduyfsegdiksuygu} by introducing finite arrays, based on the notion of height. A \emph{height} for mode $S=(F,X)$ denotes a function\footnote{Reasons for considering \emph{nonconstant} heights are stated in Comment~\ref{hjgfjhgfkjhgkjh}.}
\beq
\mbox{$K:F\ra\bN$,~~ denoted ~~$f\mapsto K_{\!f}$}\,. \label{ksdjfghvblkzjch}
\eeq
In the following definition, $A$ denotes the infinite array (\ref{fdkjhsfkjhhj}), and its set $\dependent{\Vars_A}$ of dependent variables is defined by (\ref{kljfghbkdjhb}):
\begin{definition}[mode change array]
	\label{jsahgdfkjagdf} 
		 Any height $K$ for mode $S{=}(F,X)$ defines the $K$-\emph{mode change array}
		\beq
		A_K=\left\{\left.
		\pprimeppostset{m}{k}{f}
		\,\right|\,
		f{\in}{F}\,,\,0{\leq} m{\leq} c_{\!f}\,,\,0{\leq} k{\leq} K_{\!f}
		\right\}\subset A\,.
		\label{skudjcghskdjhcg}
	\eeq
The set of dependent variables of $A_K$ is $\dependent{\Vars_{A_K}}=\dependent{\Vars_A}{\cap}\Vars_{A_K}$. Define	the tail of $A_K$ and its set of dependent variables by
	\beq
	\ttail{A_K} = \,\left\{ \pprimeppostset{m}{K_{\!f}}{f} \left|\bea{l}f{\in}{F} \\ 0\leq m\leq c_{\!f} \eea\negesp\right.\right\}\,,&\negesp\negesp&
	\ttail{\Vars_{A_K}} = \left\{
	\pprimeppostset{d}{K_{\!f}}{x} \;\left|\bea{l}
	f{\in}{F},(\edge)\in\cM \\
	0\leq m\leq c_{\!f} \\
	d=\sigma_{\!\edge}+m
	\eea\negesp\right.\right\}
\nonumber \\
	\mbox{ \emph{and} }&\negesp\negesp&\hhead{\Vars_{A_K}}\eqdef {\Vars_{A_K}}\setminus\ttail{\Vars_{A_K}}
	 \label{hjgtflukkjhjhlkj}
	\eeq
	where $\cM$ is the matching found by solving the \sigmamethod\ for $S$. We define the subset of equations that must be enabled by
		\beq\bea{rcl}
\enable{A_K}&=& 
\ttail{A_K} \cup \bigl(
A_K\cap\;\bigcup_{k\geq 0} \,\ppostset{k}{(F^-)}
\bigr)
\eea
\label{jytfjytjfgyluk}
\eeq
and we set $\disable{A_K} \eqdef A_K\,\setminus\enable{A_K}$.
\end{definition}
The rationale for (\ref{jytfjytjfgyluk}) is that 1) following Requirement~\ref{ksdjyhfgaskjfhgfkj}, we want the tail to be enabled, and 2) the invariant equations (belonging to both modes) should never be disabled. 
\begin{lemma}
	\label{sldkjfhsblkj} For any given base variable $x$ and $0{\leq}d{\leq}d_x$, the set $\{\pprimeppostset{m}{k}{z}\in\ttail{\Vars_{A_K}}\mid z{=}x,m{=}d\}$ possesses cardinality $1$.
\end{lemma}
\begin{proof}
Variable $x$ uniquely determines $f$ by the condition $(\edge)\in\cM$. Hence, the height $K_{\!f}$ is uniquely defined. The lemma follows immediately.
\end{proof}

\begin{example}\rm
	\label{kvdjhfgvkfdujgh} 
Fig.\,\ref{kefygfkjfgjhgf} shows the $\sim$\,-\,closed array ${A}_1$ (see Definition~\ref{uyytyfgkutfjhh}) associated to the mode change $\guard:\fff\ra\ttt$ for the cup-and-ball. Disabled equations (in {\color{red}red}) are included; $K=1$ is a constant; $\Vars^-=\{x,y,\dot{x},\dot{y},\postset{x},\postset{y}\}$; and $\Vars_{A_1}$ collects all the variables involved in blocks occurring at $t$ and $t{+}\vsmall$. $\ttail{A_1}$ is the block of equations attached to instant $t{+}\vsmall$ and $\ttail{\Vars_{A_1}}=\postset{\{\ddot{x},\ddot{y},\tension,\dot{x},\dot{y}\}}$. Note that $\ttail{\Vars_{A_1}}$ has cardinality $5$, whereas the cardinality of $\ttail{A_1}$ is $4$. Adding, to $\ttail{A_1}$, Euler identity $(\euler{\ddot{x}}{})$ makes it structurally nonsingular.\eproof
\end{example}
\subsection{Structural Implicit Function Theorem}
\label{ksdjhfgksjhg}
We now provide the main mathematical tool justifying step~\ref{hgrfdjtrgdhtfjhgsdrf} of  Approach~\ref{isduyfsegdiksuygu}, when we set $\vsmall$ to zero. This is achieved through a structural interpretation of Proposition~\ref{solifujhsgldikugh} (Implicit Function Theorem). 

We assume a reference instant $\changetime$, a time step $\vsmall$, and the resulting set $\allVars$ of variables, associated with them following Notations~\ref{dksfjvhsgkvj}. On top of this, we consider a static system $S=(F,X{\cup}Z),X{\cup}{Z}{\,\subseteq\,}\allVars$, where $X$ and $Z$ are the dependent and free variables of $S$. We assume that $S$ is structurally nonsingular and $\sim$\,--\,closed, see Definition~\ref{uyytyfgkutfjhh}.\footnote{An example is the array $A_1$ shown in Fig.\,\ref{hgrtduhtraytrgfeads}.} System $S$ involves dependent and free variables, as well as the parameter $\vsmall$, due to the occurrence of Euler identities. $S$ is thus of the form $F(X,Z,\vsmall){=}0$, to which Proposition\,\ref{solifujhsgldikugh} applies. As its main  smoothness condition, Proposition\,\ref{solifujhsgldikugh} requires that $\partial F/\partial X$ shall be invertible at a solution $(x,z,0)$ of $F$. 

As now usual in our approach, we weaken invertibility to generic invertibility (equivalently, structural nonsingularity). Accordingly, Proposition\,\ref{solifujhsgldikugh} weakens as follows:
\smallskip
\begin{proposition}[Structural Implicit Function Theorem]
	\label{jsdhgfaksjyhcg} 
		If setting $\vsmall$ to zero in the system $F(X,Z,\vsmall)=0$ yields a structurally nonsingular system, then:
	\begin{enumerate}
		\item \label{ksjdfhgskdjyhg} System $F(X,Z,0)=0$ structurally determines $X$ as a function of $Z$; and
		\item \label{jesdhfgkjshg} For a same valuation of $Z$, the so obtained valuation of $X$ is the limit, when $\vsmall\searrow{0}$, of the valuation of $X$ defined by system $F(X,Z,\vsmall)=0$.
	\end{enumerate}
\end{proposition}
\smallskip
Statement~\ref{ksjdfhgskdjyhg} is equivalent to assuming that $F(X,Z,0)=0$ is structurally nonsingular, hence only Statement~\ref{jesdhfgkjshg} has added value. Proposition\,\ref{jsdhgfaksjyhcg} induces the following rule governing the setting of $\vsmall$ to zero in the system $F(X,Z,\vsmall){=}0$:
\smallskip
\begin{rrule}
	\label{ksjhcgskdjgh} 	
It is legitimate to set $\vsmall$ to zero in system $F(X,Z,\vsmall)=0$, provided that the resulting system remains structurally nonsingular.
\end{rrule}
\smallskip
We will use Rule~\ref{ksjhcgskdjgh} as a validity criterion when performing Step~\ref{hgrfdjtrgdhtfjhgsdrf} of Approach~\ref{isduyfsegdiksuygu}.

\section{Rescaling analysis of mode change arrays}
\label{dfkjshlksgfskdjfmg}
\label{jdhsgcvfsjhtdfj}
In this section we develop a rescaling analysis for mode change arrays. 
Conventions~\ref{skdjhasdgfckjsadhg} and~\ref{slrikfgsjhfglsiugh} are in force.
This section focuses on exploiting array $A$ introduced in (\ref{fdkjhsfkjhhj}), recalling that facts were removed from it (Convention~\ref{conv:remove-facts}).

We first specify more precisely the class of ``model functions'', on top of which $A$ is built. 
\begin{definition}[model functions]
	\label{skldifmjghslnkj} 
	The set $\allFunctions$ of \emph{model functions} is defined as the set of all smooth functions defined on the set of variables $\allVars$, augmented with the set of all Euler identities relating the variables belonging to $\allVars$.
\end{definition}
\begin{lemma}
	\label{juyhfhtgnfhg} $\allFunctions$ contains all the Euler identities, as well as the closure of $\baseFunctions$ under time differentiation $f\mapsto\dot{f}$.
\end{lemma}
\begin{proof} The first statement is immediate.
The second statement follows from the chain rule for differentiation.
\end{proof}

\noindent We now zoom on the features of model functions that we will be using in the rescaling analysis:
\begin{itemize}
	\item For $f{\in}\allFunctions$ a smooth function defined on $\allVars$, we distinguish the subset $\Lin{f}$ of its variables that enter linearly: 
\beq
f=\sum_{x\in\Lin{f}} x \times f_x\,,
\label{skdufchsagvdlcsivcg}
\eeq
where each $f_x$ is a smooth nonlinear function of a finite subset of variables belonging to ${\allVars}\,\setminus\,\Lin{f}$. 
\item 
Euler identity (\ref{htgrsfdjyrnwsefgdjt}), where $x\sim{z}$ and $n\eqdef\mfun{x}{-}\mfun{z}{>}0$, has the following form, where $a_i$ are constants:
\beq\euler{x}{z}&=&
x-\vsmall^{-n}\sum_{i=0}^{n}a_i\times\ppostset{(-i)}{z}\,.
\label{jsdhgfsjhydfg}
\eeq
\item Classes (\ref{skdufchsagvdlcsivcg}) and (\ref{jsdhgfsjhydfg}) of functions are subsumed by the following larger class collecting all functions of the form
\beq
f=\sum_{x\in\Lin{f}} \vsmall^{-n_x} \times  x \times f_x\,,
\label{jdhsgfsdkjhfg}
\eeq
where each $n_x$ is a nonnegative integer, and each $f_x$ is a smooth nonlinear function of a finite subset of variables belonging to ${\allVars}\,\setminus\,\Lin{f}$.
\end{itemize}
In the sequel, we use the form (\ref{jdhsgfsdkjhfg}) for our model functions. 
%
%
\begin{definition}[rescaling offset]
	\label{hgtdfheafsgrefa} 
	Call \emph{rescaling offset} a function
	$\rescaling{}:\allFunctions\ra\bN\cup\{+\infty\}$, 
	denoted $\polyn\mapsto\rescaling{\!\polyn}$, satisfying the following axioms, where we write $\rescaling{x}$ when $f\equiv x$, and $\Vars_{\!f}$ is defined in Notations~$\ref{dksfjvhsgkvj}$:
	\beq
	\forall n\in\bN&\negesp\Ra\negesp& \rescaling{\vsmall^{-n}}=n
	\label{hgfxhngfdxcnjhcg}
	\\ [1mm]
f\mbox{ \emph{of the form (\ref{jdhsgfsdkjhfg})}} &\negesp\Ra\negesp& \rescaling{\!\polyn}=\displaystyle\max_{x\in\Lin{f}}\,(n_x{+}\rescaling{x}{+}\rescaling{f_x}) \label{hgfdjhgfmjm}
	\\ 
	f\mbox{ \emph{non-linear} }&\negesp\Ra\negesp&\rescaling{\!f}=
~\mathbf{if}\;\rescaling{\Vars_{\!f}}{=}0\;\mathbf{then}\;0 \;\mathbf{else}\;{+}\infty
	\label{grsjutrwskjthea}
	\eeq
	where $\rescaling{\Vars_{\!f}}=0$ means $\rescaling{x}=0$ for every $x\in{\Vars_{\!f}}$.
\end{definition}
\begin{ccomment}\rm
	\label{jshgcfjahtgf} Axioms (\ref{hgfxhngfdxcnjhcg}--\ref{grsjutrwskjthea}) are supported by the following intuition. The rescaling offset aims to represent the ``impulse order'', namely $\rescaling{x}=n$ captures that $x$ is $O(\vsmall^{-n})$. Axiom (\ref{hgfxhngfdxcnjhcg}) expresses that intent, and axiom (\ref{hgfdjhgfmjm}) is immediate considering (\ref{jdhsgfsdkjhfg}). Axiom (\ref{grsjutrwskjthea}) formalizes that, if $f$ is nonlinear, we do not know to assign an ``impulse order'' to it, unless it is known that all the variables involved in $f$ are not impulsive.\eproof
\end{ccomment}
Say that a function is \emph{impulsive} if its rescaling offset is positive. The \emph{rescaling} of function $\polyn$ is defined by
\beq
\rescaled{f}&\eqdef&\mathbf{if}~{\rescaling{\!f}}{<}\infty~\mathbf{then}~\vsmall^{\rescaling{\!f}} f~\mathbf{else}~\mbox{``undefined''}\,.
\label{grshrefrsriyyk}
\eeq
Of course, the intent of rescaling is that, if defined, $\rescaled{\polyn}$ will be non impulsive:
\begin{lemma} \label{kerdjfhsgdkfcjsduh} For every $f$ such that $\rescaling{\polyn}{<}\infty$: we have $\rescaling{\rescaled{\polyn}}{=}0$ and $\rescaled{({\rescaled{\polyn}})}{=}\rescaled{\polyn}$.
\end{lemma}
\begin{proof}
	Using (\ref{grshrefrsriyyk}), we get ${\polyn}{=}\vsmall^{-\rescaling{\!\polyn}}{\times}\rescaled{\polyn}$. Hence, by axiom (\ref{hgfdjhgfmjm}), we get $\rescaling{\!\polyn}{=}{\rescaling{\!\polyn}}{+}\rescaling{\rescaled{\polyn}}$, showing that $\rescaling{\rescaled{\polyn}}{=}0$. 
\end{proof}
%
We are given a mode change $(S^-,\changetime,S)$. To formalize Approach~\ref{isduyfsegdiksuygu}, we consider the following \emph{rescaling problem,} where we use Notations~\ref{htgfkjthndfj}. For $K$ a height, let $\cM_K$ be a matching of $(A_K,\dependent{\Vars_{A_K}})$, and let $(\edge_{\!f})\in\cM_K$, where $f$ has the form (\ref{jdhsgfsdkjhfg}). Define
\beq
f_{\downarrow\cM_K}  
\label{hgfdjyhtgfkjmhgb}
\eeq
by restricting, in (\ref{jdhsgfsdkjhfg}), the sum to the subset of its terms involving the variable $x_{\!f}$ (note that $x_{\!f}$ may or may not belong to $\Lin{f}$). We denote this subset by $$\LinM{f}{\cM_K}\,.$$
For example, referring to array $A_1$ of Fig.\,\ref{kefygfkjfgjhgf} and equation $(\postset{\dot{k_1}})$ in it, this equation is matched with variable $\bemph{\postset{\dot{y}}}$ in matching $\cM_1$ highlighted in blue, hence $(\postset{\dot{k_1}})_{\downarrow\cM_1}=\postset{(y\dot{y})}$.
\begin{problem}[rescaling problem]
	\label{likgusefhdrlpouhiouh} For a given pair $(A,\Vars^-)$ following $(\ref{fdkjhsfkjhhj},\ref{htdjhtdmhg})$, find:
\begin{enumerate}
\item 
a {height $K$} following $(\ref{ksdjfghvblkzjch})$;
\item \label{srtjhsrthsrth} 
a {variable-complete matching $\cM_K$} for $A_K$,	satisfying (see $(\ref{jytfjytjfgyluk})$ for the definition of $\enable{A_K}$)
	\beq
	\forall f\in\enable{A_K}\;\,\Ra \;\,f\in\cM_K
	\label{srrftjhtr} 
	\eeq
	and such that {the following system of equations over rescaling offsets has a solution} $\rescaling{}=(\rescaling{x})_{x\in{X}}$ that is finite for every $x$:
		\beq
		\forall f \in \cM_K \Ra \rescaling{\!f}=\rescaling{(f_{\downarrow\cM_K})}\;, &\mbox{ where $f_{\downarrow\cM}$ was defined in $(\ref{hgfdjyhtgfkjmhgb})$.}
\label{dksjhfgksjhgfk}
\eeq
\end{enumerate}
Triple $(K,\cM_K,\rescaling{})$ is a \emph{solution} of the rescaling problem.
\end{problem}
Expanding equation (\ref{dksjhfgksjhgfk}) by using Axiom (\ref{hgfdjhgfmjm}) yields
\beq
\displaystyle\max_{x\in\Lin{f}}\,(n_x{+}\rescaling{x}{+}\rescaling{f_x}) = 
\displaystyle\max_{x\in\LinM{f}{\cM_K}}\,(n_x{+}\rescaling{x}{+}\rescaling{f_x})
\label{slfighlifghpuifhiu}
\eeq
which, by applying Axiom (\ref{grsjutrwskjthea}) to $\rescaling{f_x}$, reveals that (\ref{dksjhfgksjhgfk}) is indeed to be solved for $(\rescaling{x})_{x\in X}$.
\begin{example}\rm
	\label{kfsdjhfgskjdfgj} %
	These notions are illustrated on the cup-and-ball example (Section~\ref{jhtgfkuytjrsdfghk}). In Fig.~\ref{fdkghldkjhkjj}, $\enable{A_0}$ collects equations $(e_1,e_2)$ since they are common to both modes, and $(\dot{k_1})$ since it is a consistency equation ot the last instant of the array. Thus, $\disable{A_0}=\emptyset$; hence, disabling $(\dot{k_1})$ violates this requirement. In Fig.~\ref{kefygfkjfgjhgf}, $\ttail{A_1}$ collects all the equations attached to the last instant $t{+}\vsmall$; $\hhead{A_1}={{{A}_1}}\setminus\ttail{A_1}$; $\disable{A_1}$ collects equations $(\dot{k_1},\ddot{k_1})$, since they belong to the new mode only and do not sit in the last instant; $\enable{A_1}$ is the complement. The solution of the corresponding rescaling problem consists of the matching highlighted in blue, and the rescaling offsets listed in (\ref{skksdjfhsgclugy}).\eproof
\end{example}
Using a solution $(K,\cM_K,\rescaling{})$, we can apply the following procedure to (hopefully) generate the system of equations characterizing the hot restart. In this procedure, we make time step $\vsmall$ explicit, since we let $\vsmall:={0}$ in its step~\ref{ligujehdrglieugh}.
\begin{procedure}
	\label{skedfjuyghwsaliukgh}
	Let $(K,\cM_K,\rescaling{})$ be a solution of Problem~$\ref{likgusefhdrlpouhiouh}$, and let $A_\vsmall$ be the subarray of $A_K$ collecting the equations matched in $\cM_K$. Perform the following:
	\begin{enumerate}
		\item \label{kuhyglsdkjvhblk} erase from $A_\vsmall$ all equations not matched in $\cM_K$;
		\item rescale the remaining equations and variables using $(\ref{grshrefrsriyyk})$; the resulting rescaled array is denoted by $\rescaled{A}_\vsmall$;
		\item \label{ligujehdrglieugh} set $\vsmall :={0}$ in $\rescaled{A}_\vsmall$, and denote the result by $\rescaled{A}_0$;
		\item rename as follows all the variables occurring in $\rescaled{A}_0$:
		\beq\negesp\negesp\negesp\bea{rcll}
		x\,{\in}\,\Vars^- &\negesp\gets\negesp& \mmoins{x} &
		\\ [1mm]
		\pprimeppostset{m}{k}{x} \,{\in}\,\ttail{\Vars_{A_K}}&\negesp\gets\negesp& \pplus{(\pprime{m}{x})} & \mbox{\emph{(restart variables)}}
		\\ [1mm]
		\pprimeppostset{m}{k}{x} \,{\in}\,\hhead{\Vars_{A_K}}  &\negesp\gets\negesp& \pprimeppostset{m}{k}{x} & \mbox{\emph{(auxiliary variables)}}
		\\ [1mm]
		\pprimeppostset{m}{k}{x} \,{\not\in}\,{\Vars_{A_K}}&\negesp\gets\negesp& \mbox{\emph{undefined}} & 
		\eea  \label{erlfgiuehrliueh}
		\eeq
	\end{enumerate}	
	Let $A_\star$ be the array obtained from $A_\vsmall$ by performing the above sequence of operations: we call it the \emph{restart system}.
\end{procedure}
For the cup-and-ball, this procedure was illustrated in (\ref{hgtgdfkjhglkj},\ref{skeufygsdkjyhdhg},\ref{keufygkuaygkuyg}).

The restart system has a unique solution under certain conditions, as will be seen in Theorem~\ref{skldjhfgslkdufghikl} and its proof. These conditions define good solutions of the rescaling problem:
\begin{definition}[good solution]
	\label{ksuyhjgsakjhsg} 
Triple $(K,\cM_K,\rescaling{})$ is a \emph{good solution} of the rescaling problem if the following holds:
	\beq
	\forall f\in\cM_K &\Ra& \rescaling{\!f}<\infty
	\label{jsdyhgcdfasjyht}
	\\
\forall x\in\ttail{\Vars_{A_K}} &\Ra& \rescaling{x}=0 
\label{jkdsfhgfjsmdhfgjk}
\\ 
\forall x{=}\pprimeppostset{m}{k}{y}\in\hhead{\Vars_{A_K}} &\Ra& 
\pprimeppostset{(m-n)}{(k+n)}{y}\in{\Vars_{A_K}}
\label{klifgyegrujytfgeikuy}
\\ 
&&\mbox{\emph{where} }n\eqdef\min(\rescaling{x},m) \hspace*{7.5cm} \Box \hspace*{-2mm}
\nonumber
\eeq
\end{definition}
Condition (\ref{jsdyhgcdfasjyht}) expresses that the enabled array must be rescalable. Condition (\ref{jkdsfhgfjsmdhfgjk}) expresses that restart variables should not be impulsive. Condition (\ref{klifgyegrujytfgeikuy}) will ensure that the last case of renaming (\ref{erlfgiuehrliueh}) never occurs.
\begin{ccomment}\rm
	\label{hjgfjhgfkjhgkjh} Note that, if the new mode $S$ is the composition of two noninteracting DAE systems, then Problem~\ref{likgusefhdrlpouhiouh} also decomposes into two independent rescaling problems. Requiring heights to be constant would then be inadequate.\eproof
\end{ccomment}
The solution (\ref{skksdjfhsgclugy}) for the rescaling problem of the cup-and-ball was good. These results are further illustrated by a series of examples in Appendices~\ref{ksdfiygskdyjgf} and~\ref{hgrdfhreagfdsuj}.
Differences with the approach by Trenn et al. are illustrated in Appendix~\ref{sldkvjshdblvckujsh}.

\section{Main theorems}
\label{ldxkjvbvskljvh}
\subsection{Correctness of our approach}
\label{hgrfdjhtrsbdgfb}
\begin{theorem}
	\label{skldjhfgslkdufghikl} The hot restart system obtained by solving Problem~$\ref{likgusefhdrlpouhiouh}$ and then applying Procedure~$\ref{skedfjuyghwsaliukgh}$, satisfies the following properties:
	\begin{enumerate}
		\item \label{jkhtgfvkjhgkjh}
		Given $(K,\cM_K)$ satisfying $(\ref{srrftjhtr})$, system $(\ref{slfighlifghpuifhiu})$ possesses a solution $\rescaling{}$, generically.
		\item 	\label{drtjhsrtjhrsth}
		If solution $(K,\cM_K,\rescaling{})$ is good, then: 
		\begin{enumerate}
			\item rescaled array $\rescaled{A}_\vsmall$  involves only variables belonging to $\Vars_{A_K}$; 
			\item $\rescaled{A}_0$ is structurally nonsingular;
			\item $\rescaled{A}_0$ is the unique limit of $\rescaled{A}_\vsmall$ when $\vsmall{\,\searrow\,}{0}$; and
			\item restart system $A_\star$ uniquely determines the states $\pplus{(\pprime{m}{y})}$ for the hot restart of the new mode.
		\end{enumerate}

		\item \label{jcsdghsdvckjhsdgk} If good solutions $(K,\cM_K,\rescaling{})$ exist, then height $K$ and rescaling offsets $\rescaling{}$ are unique; the only possible source of non uniqueness is indeed the matching $\cM_K$.
	\end{enumerate}
\end{theorem}
In Statement~\ref{jkhtgfvkjhgkjh}, ``generically'' refers to the structural analysis setup. More precisely, we extend $\rescaling{}$ to take values in $\bR_{\geq 0}$ and exponents of $\vsmall^{-1}$ to belong to $\bR_{\geq 0}$. Problem~\ref{likgusefhdrlpouhiouh} extends to this situation. Then, ``generically'' means: for almost every value for the exponents of $\vsmall^{-1}$ in some neighborhood of their nominal values.
\begin{proof}
See Section~\ref{skduhvgskjvghkjh}.
\end{proof}

\noindent Theorem~\ref{skldjhfgslkdufghikl} allows us to reject a model at compile time if no good solution is found, but gives little feedback in this case, regarding how to correct the model.  The following corollary of Theorem~\ref{skldjhfgslkdufghikl} improves the feedback to the designer by identifying the subset of variables that are insufficiently determined. Notations of Theorem~\ref{skldjhfgslkdufghikl} are used.

Consider the array $A_\vsmall$ and assume that goodness conditions (\ref{jsdyhgcdfasjyht}) and (\ref{jkdsfhgfjsmdhfgjk}) are satisfied. Let $\YVars$ be the subset of variables $x$ violating the goodness condition (\ref{klifgyegrujytfgeikuy}). Remove, from $A_\vsmall$, the equations that are matched with a variable belonging to $\YVars$, and keep ${\dependent{\Vars_{A_\vsmall}}}$ unchanged. The resulting system $(A_\vsmall^{\YVars},{\dependent{\Vars_{A_\vsmall}}})$ is now rectangular, with more dependent variables than equations. Its Dulmage-Mendelsohn decomposition (\ref{eujhydwtye}) yields $A_\vsmall^{\YVars}=A_\vsmall^{\sf r}\uplus A_\vsmall^{\sf u}$ (the overdetermined subsystem is empty). Let $\cM_K^{\sf r}$ be a perfect matching for $\bigl(A_\vsmall^{\sf r},\Vars^{\sf r}\bigr)$.
\begin{corollary}
	\label{hgrfdjrgfdfkm} Consider the objects of Theorem~$\ref{skldjhfgslkdufghikl}$ and perform the substitutions 
	\beq
	A_\vsmall\gets A_\vsmall^{\sf r}\,,\, {\dependent{\Vars_A}}\gets \Vars^{\sf r}\,,\, \cM_K\gets\cM_K^{\sf r}
	\label{kjhyfgliukjk}
	\eeq
	whereas the rescaling offsets $\rescaling{x}$ and $\rescaling{\!f}$ are kept unchanged for $x\in\Vars^{\sf r}$ and $f\in{A_\vsmall^{\sf r}}$. Let $A_\star^{\sf r}$ be the array obtained by performing, mutatis mutandis, the operations listed in Procedure~$\ref{skedfjuyghwsaliukgh}$. Then, $A_\star^{\sf r}$ determines the values of state variables belonging to $\Vars^{\sf r}$ for a hot restart of the new mode. Other state variables may be undetermined.
\end{corollary}
\begin{proof}
See Appendix~\ref{hjgfgkkjmhnm}.
\end{proof}

\noindent In Appendix~\ref{klsdifgujhsjdhfh} we analyze a slight modification of the cup-and-ball game, by which the mode change is no longer state based, but rather an external input. For this example, the rescaling problem possesses no good solution. However, applying Corollary~\ref{hgrfdjrgfdfkm} shows that restart positions are determined, but restart velocities are not. 

So far nothing tells us how large the height $K$ for the array should get explored. The next result proposes bounds for it:
\begin{theorem}
	\label{sekdfuysageloiug} The exist two heights $K_*\leq K^*$, depending on the equation offsets of the new mode $S$ and the set $\Vars^-$ of past variables, such that:
	\begin{enumerate}
		\item To find a variable-complete matching $\cM_K$ satisfying requirement $(\ref{srrftjhtr})$, it is enough to select $K\geq K^*$, whereas no such matching can be expected unless $K\geq K_*$;
		\item Let $K$ be the variable height of a successful solution of Problem~$\ref{likgusefhdrlpouhiouh}$. Then, $K\leq K^*$ actually holds.
	\end{enumerate}
\end{theorem}
\begin{proof}
See Appendix~\ref{slviwsjrhfvluikvhk}.
\end{proof}

\noindent For the cup-and-ball example, we have $K_*{=}K^*{=}1$. Although not true in general, the equality $K_*=K^*$ is expected to hold in practice. Theorem~\ref{sekdfuysageloiug} is illustrated in Appendix~\ref{hgnhtdfiytujhgn}.

\subsection{Solving the rescaling problem}
\label{jsdhgcfvsdjhcg}
Recall our functions $f$ belong to class (\ref{jdhsgfsdkjhfg}). However, since our actual equations will be of the form either (\ref{skdufchsagvdlcsivcg}) or (\ref{jsdhgfsjhydfg}), the following assumptions can be formulated on class (\ref{jdhsgfsdkjhfg}):
\begin{assumption}
	\label{jsdhfgsdjhg} The functions $f$ of class $(\ref{jdhsgfsdkjhfg})$ that we consider, satisfy the following conditions:
	\beq
\exists x{\in}\Lin{f}\mbox{ s.t. }n_x>0 &\mbox{implies}& \Lin{f}=X
\label{jdsfdghsdjfg}
\\
x_{\!f}\in\Lin{f} &\mbox{implies}& \LinM{f}{\cM_K}=\{x_{\!f}\}
\label{fkjsdfhvskdjyhcgk}
\eeq
\end{assumption}
\paragraph*{\normalsize Justification}
The left hand side of (\ref{jdsfdghsdjfg}) holds only if $f$ is of class (\ref{jsdhgfsjhydfg}), i.e., a Euler identity; in this case, all the variables are linearly involved, whence the right hand side. Then, (\ref{fkjsdfhvskdjyhcgk}) holds for every function $f$ of class (\ref{jdhsgfsdkjhfg}).\eproof
\begin{theorem}
	\label{dsjkchagfvckjshg} Assumption~$\ref{jsdhfgsdjhg}$ is in force.
	\begin{enumerate}
		\item \label{jsdhgcfsjcyfg} Finding a good solution for $(\ref{slfighlifghpuifhiu})$ decomposes as the following two steps, where 
		$$\cM_\cL \eqdef \{(\edge_{\!f}){\in}\cM_K {\,\mid\,} x_{\!f}{\in}\Lin{\!f}\;\}$$
		\begin{enumerate}
		\item \label{ksduhfvgskdujfgy} Find a finite minimal nonnegative solution for $\rescaling{\cL}\eqdef\{\rescaling{x_{\!f}}\mid (f,x_{\!f}){\in}\cM_\cL\,\}$, of the system
		\beq
		\forall f{\in}\cM_\cL:
\displaystyle\max_{x\in\LinM{f}{\cM}}\,(n_x{+}\rescaling{x})
= \displaystyle\max_{x\in\Lin{f}}\,(n_x{+}\rescaling{x})
\label{jshgfcsdjhgsfj}
\eeq
\item \label{kjdfghksg} Extend this solution to a good solution for $(\ref{slfighlifghpuifhiu})$ by checking the satisfiability of the following property:
\beq
\forall f\in\cM_K, \forall x\not\in\Lin{\!f}&\Ra&
 \rescaling{x}=0
\label{ksdjfhsgdkfjgh}
\eeq
	\end{enumerate}
	\item \label{jsdhvgcsdkjcg} Rescaling equation $(\ref{jshgfcsdjhgsfj})$ simplifies as 
		\beq
		\forall f{\in}\cM_\cL:&
n_{x_{\!f}}{+}\rescaling{x_{\!f}}
= \displaystyle\max_{x\in\Lin{f}}\,(n_x{+}\rescaling{x})
\label{kdjvchsgkvjshdgv}
\eeq
\end{enumerate}
\end{theorem}
\begin{proof} We first prove Statement~\ref{jsdhgcfsjcyfg}. Consider, as part of the rescaling problem, the system of equations (\ref{slfighlifghpuifhiu}). Suppose it possesses a good solution. By (\ref{jsdyhgcdfasjyht}) we look for a solution ensuring that all the variables and enabled equations possess a finite rescaling offset. If such a solution $\mu$ exists, then:
	 \begin{enumerate}
		 \item For $(\edge_{\!f})\in\cM_\cL$, (\ref{jshgfcsdjhgsfj}) holds, and
		\item Property (\ref{ksdjfhsgdkfjgh}) holds.
	 \end{enumerate}
In addition, if a solution is good, then any smaller solution is also good. By Statement~\ref{jcsdghsdvckjhsdgk} of Theorem~\ref{skldjhfgslkdufghikl}, there is at most one good solution. Hence the good solution, if any, must be a minimal solution of (\ref{jshgfcsdjhgsfj}). This proves Statement~\ref{jsdhgcfsjcyfg}.

Focus on Statement~\ref{jsdhvgcsdkjcg}. By Property (\ref{fkjsdfhvskdjyhcgk}) of Assumption~\ref{jsdhfgsdjhg}, the max sitting on the left hand side of (\ref{jshgfcsdjhgsfj}) boils down to the singleton $x=x_{\!f}$. Thus (\ref{jshgfcsdjhgsfj}) rewrites as (\ref{kdjvchsgkvjshdgv}).
\end{proof}

\noindent System (\ref{kdjvchsgkvjshdgv}) is of max/$+$ type. Yet, it belongs to a subclass amenable of much more efficient algorithms. More precisely, straightforward manipulations bring (\ref{kdjvchsgkvjshdgv}) to a set of inequalities of the form $\rescaling{x}-\rescaling{y}\leq D_{xy}$, where $x,y\in{X}\cup\{\zzero\}$, the unknowns are $\rescaling{x}$ taking values in $\overline{\bZ}\eqdef\bZ\cup\{{\pm}\infty\}$, $\zzero$ is a distinguished element such that $\rescaling{\zzero}=0$, and $D_{xy}$ is a matrix of entries belonging to $\overline{\bZ}$. Such systems of inequality constraints are known as Difference Bound Matrices (DBM)~\cite{DBLP:conf/avmfss/Dill89,DBLP:conf/pado/Mine01}. DBM can also be used for the dual problem of the \sigmamethod, for which they provide a modular (compositional) algorithm, see~\cite{caillaud:hal-05257001}.

\subsection{Characterizing the solutions}
\label{kaerjfhgkaewrjyfhg}
In this section, we complete the proof that our approach solves the hot restart problem, by formalizing point~\ref{kaujfyhgakujyhg} of Problem~\ref{dsjkchgfckajycgk} and proving that our solution solves it.

\paragraph*{Background regarding rings and ideals:}
Our reference is the textbook~\cite{ideals2010}. The set of solutions of a (static) system of smooth equations is best studied by using the notions of \emph{ring} $\allFunctions$ of real-valued functions $f$, and of \emph{ideal} of this ring. For $S=(F,X)$  such a system,  we denote by $V_S$ its set of solutions, i.e., the set of all valuations $\vval{X}\in\dom{X}$ such that $F(\vval{X})=0$.  
If $S=S_1\cup{S_2}$, then $V_S=V_{S_1}\cap V_{S_2}$, where the intersection symbol first requires to take the inverse projections of $V_{S_1}$ and $V_{S_2}$, from $X_1$ and $X_2$ to $X_1{\cup}X_2$. The set of all functions $f$ such that $F(\vval{X})=0$ implies $f(\vval{X})=0$ is an ideal, i.e., is closed under addition and multiplication by an arbitrary element of $\allFunctions$, we denote this ideal by $\cI_S$. If $S=S_1\cup{S_2}$, then $\cI_S=\cI_{S_1}+\cI_{S_2}\eqdef\{f_1+f_2\mid f_i\in\cI_{S_i}\}$.

\subsubsection*{Analyzing the cup-and-ball example}
Consider the cup-and-ball example. The  subsystem common to the two modes is
\[S^\pm:\left\{\bea{l}
0=\ddot{x}+\tension{x} \\
0=\ddot{y}+\tension{y}+g
\eea\right.\]
The algebraic variable (the tension $\tension$) plays no role in hot restart. Eliminating it from $S^\pm$ yields the equation
\beq
f_\pm=0, \mbox{ where }f_\pm=\ddot{y}x-\ddot{x}y+gx\,, \label{sdkfvijhsgdbvkjgh}
\eeq
which belongs to the ideal generated by $S^\pm$. However, equation (\ref{sdkfvijhsgdbvkjgh}) exhibits impulses at the mode change, and we have $\rescaling{f_\pm}=\rescaling{\ddot{x}}=\rescaling{\ddot{y}}=1$. Therefore:
\[\bea{rclcl}
\rescaled{\ddot{x}} &=& \vsmall{}\ddot{x} &=& \postset{\dot{x}}-\dot{x} \\
\rescaled{\ddot{y}} &=& \vsmall{}\ddot{y} &=& \postset{\dot{y}}-\dot{y} \\
\rescaled{f_\pm} &=& \vsmall{}{f_\pm} &=&  (\postset{\dot{y}}-\dot{y})x-(\postset{\dot{x}}-\dot{x})y+{0(\vsmall)}
\eea\]
which, by letting $\vsmall:={0}$, yields
\beq
0&=&(\postset{\dot{y}}-\dot{y})x-(\postset{\dot{x}}-\dot{x})y\,. \label{jhgfkysztyregfd}
\eeq
Focus on our solution (\ref{keufygkuaygkuyg}) for the restart of the cup-and-ball, which yields (\ref{dkfvjshdgfkjhy}) by eliminating $\rescaled{\tension}$ from it. Observe that (\ref{jhgfkysztyregfd}) identifies with (\ref{dkfvjshdgfkjhy}) after renaming (\ref{erlfgiuehrliueh}) is performed. 

\subsubsection*{{Generalizing}}
We now generalize the analysis of the cup-and-ball example. Following Section~\ref{jdhsgcvfsjhtdfj}, we consider the ring $\allFunctions$ of functions of the variables belonging to $\allVars$ augmented with $\{\vsmall^{-n}\mid n{\geq}0\}$. We also consider the ring $\Offsets\eqdef(\bN{\cup}\{+\infty\},\max,+)$ of rescaling offsets equipped with the max/$+$\,--\,algebra. 

From now on, we consider $(S^-,\changetime,S)$ a mode change, and we assume that Problem~\ref{likgusefhdrlpouhiouh} possesses a solution $(K,\cM_K,\rescaling{})$ satisfying goodness conditions (\ref{jsdyhgcdfasjyht}--\ref{klifgyegrujytfgeikuy}). By  axioms (\ref{hgfxhngfdxcnjhcg}--\ref{grsjutrwskjthea}), the following result easily follows regarding $\rescaling{}$:
\begin{lemma}
	\label{htdfkfkjhgssssghch} 
$\rescaling{}:\allFunctions\ra\Offsets$ is a ring homorphism.
\end{lemma}
Consider the map $\Lambda:\allFunctions\ra\allFunctions$ consisting in setting $\vsmall{:=}{0}$ in $\rescaled{f}$:
\beq
\Restart{f} &\eqdef&\rescaled{f}[\vsmall{:=}{0}]\,.
\label{heagrehrsdjh}
\eeq
\begin{lemma}
	\label{jfgklujlkughlkj} 
Writing $\rescaling{i}$ instead of $\rescaling{f_i}$, map $\Lambda$ satisfies
\[\hspace*{10mm}\bea{rcl}
\Restart{f_{1}{+}f_{2}} &=& \sum_{i:\rescaling{i}=\max(\rescaling{1},\rescaling{2})}~\Restart{f_{i}}
\\ [1mm]
\Restart{f_{1}{\times} f_{2}}&=&\Restart{f_{1}}\times \Restart{f_{2}} 
\eea\]
\end{lemma}
\begin{proof}
See Appendix~\ref{kcsjgvcsldjghkj}.
\end{proof}

\noindent The following theorem shows that our solution meets Requirement~\ref{kaujfyhgakujyhg} of Problem~\ref{dsjkchgfckajycgk} (by the way, this requirement is now formalized):
\begin{theorem}
	\label{sjdufghafejahyfg} Every function belonging to the ideal spanned by the subsystem that is common to both modes, i.e., $S^\pm=S^-{\cap}S$, is mapped via $\Lambda$ to a function that belongs to the ideal generated by our solution for restart.
\end{theorem}
\begin{proof}
Convention~{\ref{slrikfgsjhfglsiugh}} is in force throughout this proof. Let $A_\vsmall$ be the mode change array associated to the considered solution $(K,\cM_K,\rescaling{})$, and let $A_\vsmall^\cM$ be the subarray collecting the enabled equations (i.e., matched in $\cM_K$). Rephrasing Theorem~\ref{drtjhsrtjhrsth} by using map $\Lambda$, and seeing arrays as sets of equations, we have:
\beq\bea{rcl}
\rescaled{A_0} &=& \RRestart{A_\vsmall^\cM} \mbox{ is structurally nonsingular, and}
\\
A_\star &=& \rescaled{A_0} \mbox{ followed by renaming }(\ref{erlfgiuehrliueh})\,.
\eea  \label{ujytfriukytloik}
\eeq
Define the ideals
\beq
\cI_\vsmall \eqdef \sspan(A_\vsmall^\cM) &\mbox{and}& \cI_0 \eqdef \sspan(\rescaled{A_0})\,. \label{htrgfsgraesrgjdthg}
\eeq
Using Lemma~\ref{jfgklujlkughlkj} yields
\beq
\forall f\in\cI_\vsmall &\Ra& \RRestart{f}\in\cI_0\,.  \label{htrsjythdfgmnbv}
\eeq
Let $A^\pm_\vsmall$ be the subarray of $A^\cM_K$ collecting shifted versions of the equations involved in both modes. By (\ref{srrftjhtr}), invariant equations cannot be disabled, i.e., $A^\pm_\vsmall\subseteq\enable{A_\vsmall}$. Hence equations belonging to $A^\pm_\vsmall$ are all matched in $\cM_K$. Therefore, $A^\pm_\vsmall \subseteq \cI_\vsmall$. Hence, by (\ref{htrsjythdfgmnbv}), $
\RRestart{A^\pm_\vsmall}\subseteq\cI_0$, which proves the theorem.
\end{proof}

\noindent Array $\rescaled{A_0}$ involves the following categories of variables:
\begin{enumerate}
	\item \label{lskfcjhaglfiuh} free variables belonging to $\Vars^-$; they are non-impulsive;
	\item \label{gosekldufygekiuyrg} dependent state variables sitting in the tail, i.e., of the form $\pprimeppostset{m}{K_{\!f}}{x}$ for $(\edge)\in\cM$, the solution of Problem~(\ref{liftuerhpituhepu8}) and $0\leq m<c_{\!f}$; they are also non-impulsive;
	\item \label{slifvjhsglsikugh} other variables, possibly rescaled and/or auxiliary.
\end{enumerate}
Variables of categories~\ref{lskfcjhaglfiuh} and~\ref{gosekldufygekiuyrg} are physically interpretable (left- and right-limits of state variables at the mode change). In contrast, variables of category~\ref{slifvjhsglsikugh} are not: it makes sense to eliminate them. The resulting system of equations is included in ideal $\cI_0$. Such an elimination can be performed using numerical methods for linear systems, and Gr\"obner bases and elimination theory~\cite{ideals2010} for polynomial systems. For the latter class, large systems are totally beyond reach. Even for linear systems, elimination is of order $n^3$, where $n$ is the dimension of the mode change array. 

\section{Proof of correctness}
\label{skduhvgskjvghkjh}
In this section we present the proof of the main result, namely the correctness of our approach formalized by Theorem~\ref{skldjhfgslkdufghikl}. Other proofs are deferred to appendices.
Throughout this section, Conventions~\ref{skdjhasdgfckjsadhg},~\ref{slrikfgsjhfglsiugh}, and~\ref{conv:remove-facts} are in force.
We will use the following result regarding rescaling offsets:
\begin{lemma}
	\label{jdschtyasdfcjhagf}   
		Let $x,z$ be as in Euler identity $(\ref{utrwshjyhtrf})$, and let $\rescaling{}$ be a rescaling offset. Then, $\rescaling{x}=n+\max_{0\leq i\leq n} \rescaling{(\ppostset{(-i)}{z})}$.
\end{lemma}
\begin{proof}
The lemma is a direct consequence of identity (\ref{utrwshjyhtrf}) and Axiom (\ref{hgfdjhgfmjm}).
\end{proof}

\noindent We now proceed with the proof of Theorem~\ref{skldjhfgslkdufghikl}.

\paragraph*{Proof of Statement~\ref{jkhtgfvkjhgkjh} of Theorem~\ref{skldjhfgslkdufghikl}:}
\label{klsdjghbsvkjh}
Let $\chi$ be the Heaviside function and let $\chi^\eta$ be a smoothing of it, such that $0\leq\chi^\eta\leq 1$ and $\chi^\eta=\chi$ outside the interval $[0,\eta]$, with $\eta$ being a small positive parameter. Expand $\max(x,y)=x\chi(x-y)+y\chi(y-x)$ and define, inductively:
\[\bea{rclll}
\max^\eta(x,y) &\eqdef&\negesp& x\;\chi^\eta(x-y)+y\;\chi^\eta(y-x)
\\ [1mm]
\max^\eta(x,{\rm list}) &\eqdef&\negesp&y
x\;\chi^\eta(x-\max^\eta({\rm list}))
+y\;\chi^\eta(\max^\eta({\rm list})-x)
\eea\]
Then, operator $\max^\eta$ is smooth in its arguments. Next,  with reference to Axiom (\ref{hgfdjhgfmjm}), setting
\beqq
\tilde{\mu}_{\!f} \eqdef 
\max^{\hspace*{8mm}\eta}_{x\in\Lin{f}}\tilde{\mu}_{x}
-
\max^{\hspace*{8mm}\eta}_{x\in\Lin{f\downarrow\cM_K}}\tilde{\mu}_{x}\,,
&\mbox{ where }&\tilde{\mu}_{x} \eqdef n_x{+}\rescaling{x}{+}\rescaling{\!f_x}\,,
\eeqq 
the mapping 
\beqq
(\edge) &\mapsto& 
(\tilde{\mu}_{\!f},\rescaling{x})
\eeqq
maps the variable-complete matching $\cM_K$ for system $A_\vsmall{=}0$ to a variable-complete matching $\cM_{\rescaling{}}$ for the system 
\beq
\forall f\in A_\vsmall&:& \tilde{\mu}_{\!f}=0
\label{htgfdjtrgdf}
\eeq	
whose dependent variables are the $(\rescaling{x})_{x\in X}$. Hence, system (\ref{htgfdjtrgdf}) is structurally nonsingular, let $\rescaling{x}^\eta$ denote its generically unique solution.

Let $\eta_n\searrow{0}$. The sequence $n\mapsto(\rescaling{x}^{\eta_n})_{x\in{X}}$ is bounded. It thus possesses converging subsequences. By abuse of notation we rename $n\mapsto(\rescaling{x}^{\eta_n})_{x\in{X}}$ one of them, and we denote by $(\rescaling{x}^*)_{x\in{X}}$ its limit. It is a solution of system (\ref{slfighlifghpuifhiu}). This proves the existence of solutions to this system. Uniqueness, however, cannot be guaranteed at this point.

\paragraph*{Proof of Statement~\ref{drtjhsrtjhrsth} of Theorem~\ref{skldjhfgslkdufghikl}:}
We decompose it into its successive substatements:
\begin{enumerate}
	  \item[a)] \label{klsifhwegloig} Array $\rescaled{A}_\vsmall$  involves only variables belonging to $\Vars_{A_\vsmall}$;
		\item[b)] \label{kifjhwgfliwurgh} Array $\rescaled{A}_0$ is structurally nonsingular;
		\item[c)] \label{lfiweruhfghlpuihg} $\rescaled{A}_0$ is the unique limit of $\rescaled{A}_\vsmall$ when $\vsmall\searrow{0}$; and
		\item[d)] \label{lfiwughfwlifu} The restart system $A_\star$ uniquely determines the states $\pplus{(\pprime{m}{y})}$ for the hot restart of the new mode.\eproof
\end{enumerate}
\emph{Proof of Substatement~a):} This substatement deals with rescaling offsets of variables. Let $x$ be as in Euler identity (\ref{utrwshjyhtrf}), and select $z$ such that $n\eqdef m_x{-}m_z =\min(\rescaling{x},m_x)$. Using  (\ref{utrwshjyhtrf}) yields:
\beqq
\rescaled{x} = \vsmall^{\rescaling{x}}{\;}x
&=& \vsmall^{\rescaling{x}-n}{\;}
\underbrace{\sum_{i=0}^{n}
\left(\negesp\bea{c}n \\ i\eea\negesp\right)
(-1)^i\underbrace{\ppostset{(-i)}{z}}_{z_i}}_{z_*}\,.
\eeqq
By Lemma~\ref{jdschtyasdfcjhagf}, the following two cases occur:
\begin{itemize}
	\item 
Case $\rescaling{x}{\leq}{m}$: we have
\begin{itemize} 
	\item $\rescaled{x}={z_*}$ and, in the expansion of $z_*$, $\vsmall$ does not occur and no variable $z_i$ is impulsive; 
	\item by condition (\ref{klifgyegrujytfgeikuy}), all the variables $z_i$ occurring in the expansion of $z$ sit within ${\Vars_A}$.
\end{itemize}
\item Case $\rescaling{x}{>}{m}$: we have
\begin{itemize}
	\item $\rescaled{x}=\rescaled{z_*}=\vsmall^{\rescaling{x}-m}{}z_*$ is not impulsive;
	\item for $1{\leq} i{\leq} n$, $\vsmall^{\rescaling{x}-m}{}z_i{=}\vsmall^{\rescaling{i}}{}\rescaled{z_i}$ for $\rescaling{i}{=}{\rescaling{x}{-}m}{+}{\rescaling{z_i}}$;
	\item by condition (\ref{klifgyegrujytfgeikuy}), all the variables occurring in the expansion of ${z}$ sit within ${\Vars_A}$.
\end{itemize}
\end{itemize}
This analysis is summarized by the following statement:
\beq
\mbox{\begin{minipage}{12cm}
In the two cases, $\rescaled{x}$ expands as a linear combination of terms $\vsmall^{\rescaling{i}}{}\rescaled{z_i}$, where  $\rescaling{i}{\geq} 0$ and variable $\rescaled{z_i}$ sits within array $A_\vsmall$.
\end{minipage}}	
\label{lougvhwselouhgilo}
\eeq
This proves {Substatement}~a).
\smallskip

\emph{Proof of Substatement~b), the crux of the proof (guided by Rule~$\ref{ksjhcgskdjgh})$:} Let
\[\bea{rcl}
f = \sum_{x\in\Lin{f}}\;\vsmall^{-n_x}\;x \;f_x
\eea\]
be as in (\ref{jdhsgfsdkjhfg}) and such that $\rescaling{f}{<}\infty$. Using axioms (\ref{hgfxhngfdxcnjhcg}--\ref{grsjutrwskjthea}), we get $\rescaling{f_x}=0$, implying $\rescaled{f_x}=f_x$. Setting $\tilde{\mu}_x\eqdef n_x{+}\rescaling{x}{+}\rescaling{f_x}=n_x{+}\rescaling{x}$, rescaling (\ref{grshrefrsriyyk}) for $f$ yields:
\[\bea{rcl}
\rescaled{f} &=& \vsmall^{\rescaling{\!f}}{}\sum_{x\in\Lin{f}}\vsmall^{-n_x}\,x \,f_x
\\ [1mm]
&=& \vsmall^{\rescaling{\!f}}{}\sum_{x\in\Lin{f}} \bigl(\vsmall^{-n_x}{}\,(\vsmall^{-\rescaling{x}}\rescaled{x})\,f_x\bigr)
\\ [1mm]
&=& \vsmall^{\rescaling{\!f}}{}\sum_{x\in\Lin{f}} \bigl(\vsmall^{-\tilde{\mu}_x}{}\,\rescaled{x}\,f_x\bigr)
\\ [1mm]
&=& \wemph{\vsmall^{\rescaling{\!f}}{}}\sum_{x\in\Lin{f}} \bigl(\vsmall^{\rescaling{\!f}-\tilde{\mu}_x}{}\,\rescaled{x}\,f_x\bigr)\,.
\eea
\label{lriguhwseluh}
\]
Using Axiom (\ref{hgfdjhgfmjm}) and equation (\ref{slfighlifghpuifhiu}), we get $\rescaling{\!f}-\tilde{\mu}_x\geq 0$, with equality for the term in the sum involving the unique variable $x$ such that $(\edge)\in\cM_K$. This shows that 
\beq
\mbox{
\begin{minipage}{12cm}
	 when setting $\vsmall:={0}$ in the rescaled array $\rescaled{A_\vsmall}$, the presence of term $\rescaled{x}f_x$ remains guaranteed for $(\edge)\in\cM_K$.
\end{minipage}
} \label{lsifughlsiufh}
\eeq
Hence,  rescaling $A_\vsmall$ and then setting $\vsmall:={0}$ ensures that $\rescaled{\cM_0}$ is a variable-complete matching for $\rescaled{A}_0$, where $\rescaled{\cM_0}\eqdef\{(\rescaled{f},\rescaled{x}){\,\mid\,}(\edge){\,\in\,}\cM_K\}$. This proves {Substatement}~(b). 

\smallskip
\emph{Proof of Substatement~c):}
{It is a direct consequence of Substatement~b), by Proposition~\ref{jsdhgfaksjyhcg} in Section~\ref{ksdjhfgksjhg}. 

\smallskip
\emph{Proof of Substatement~d):}
Substatements~a,b,c) together imply the following, for array $\rescaled{A}_0$:
\begin{itemize}
	\item[(i)] it is structurally nonsingular;
	\item[(ii)] it involves only variables belonging to $\Vars_{A_\vsmall}$;
	\item[(iii)] when constructing it, no equation belonging to the tail of $A_\vsmall$ was disabled.
\end{itemize}
Then, $A_\star$ is obtained from $\rescaled{A}_0$ via renaming (\ref{erlfgiuehrliueh}). By (ii) and Lemma~\ref{sldkjfhsblkj}, the renaming (\ref{erlfgiuehrliueh}) is total and injective. By (iii), all the consistency constraints of the new mode are satisfied by the solution of $A_\star=0$. Since  the renaming (\ref{erlfgiuehrliueh}) is injective by (i), $A_\star$ is also structurally nonsingular. This proves {Substatement}~(d) and finishes the proof of Statement~\ref{drtjhsrtjhrsth} of Theorem~\ref{skldjhfgslkdufghikl}.

\paragraph*{Proof of Statement~\ref{jcsdghsdvckjhsdgk} of Theorem~\ref{skldjhfgslkdufghikl}:}
\label{selrijghvlsdkjh}
We refine Statement~\ref{jcsdghsdvckjhsdgk} as the following lemma:
\begin{lemma} \label{lvdkfjhbvlfdkujh}
	\	 
	 \begin{enumerate}
	\item \label{jdhtfgfedsfm}
	For every variable $x\in\dependent{\Vars_A}$, its rescaling offset $\rescaling{x}$ is independent from the particular choice for the pair $\bigl(K,\cM_K\bigr)$ in the solution of Problem~$\ref{likgusefhdrlpouhiouh}$ satisfying goodness Conditions $(\ref{jsdyhgcdfasjyht}$--$\ref{klifgyegrujytfgeikuy})$;
		 \item \label{sdlvcsowujhvsloivch}
	 For $f\in\enable{A}$, its rescaling offset $\rescaling{\!f}$ is independent from the particular choice for the pair $\bigl(K,\cM_K\bigr)$ in the solution of Problem~$\ref{likgusefhdrlpouhiouh}$ satisfying goodness Conditions $(\ref{jsdyhgcdfasjyht}$--$\ref{klifgyegrujytfgeikuy})$;
	
	\item \label{hjtgfkujyhglk}
The height function $K{:}f{\mapsto} K_{\!f}$, for $f{\in}{F}$, is unique. 
	 \end{enumerate}
\end{lemma}
We prove the three statements successively. For the first two, since the variable height  $K$ of the array $A_\vsmall$ is not fixed, we extend rescaling offsets beyond array $A_\vsmall$, by setting $\rescaling{x}=0$ for $x\not\in\dependent{\Vars_{A_\vsmall}}$ and $\rescaling{\!f}=0$ for $f\not\in{A_\vsmall}$. 

Consider Statement~\ref{jdhtfgfedsfm} of Lemma~\ref{lvdkfjhbvlfdkujh}. With the above convention, denote by $\rescaling{x}^{K,\cM}$ the rescaling offset of $x$ for a pair $(K,\cM_K)$ being part of a solution of Problem~\ref{likgusefhdrlpouhiouh} satisfying goodness Conditions (\ref{jsdyhgcdfasjyht}--\ref{klifgyegrujytfgeikuy}). Consider two such different pairs $(K_1,\cM_{1})$ and $(K_2,\cM_{2})$. 
Suppose $\rescaling{1}\,{\eqdef}\,\rescaling{x}^{K_1,\cM_1}<\rescaling{2}\,{\eqdef}\,\rescaling{x}^{K_2,\cM_2}$ for two such pairs $(K_i,\cM_i),i{=}1,2$. Let $\rescaled{x}_i=\vsmall^{\rescaling{i}}x,i{=}1,2$ be the rescaling of $x$ based on the first and second solution. Then, $\rescaled{x}_2=\vsmall^{\rescaling{2}}x=\vsmall^{\rescaling{2}-\rescaling{1}}\rescaled{x}_1$. Therefore, setting $\vsmall:={0}$ will zero $\rescaled{x}_2$ since $\rescaling{2}-\rescaling{1}>0$, thus erasing $\rescaled{x}$ from the matching $\cM_{2}$. This  contradicts Statement~\ref{drtjhsrtjhrsth} of Theorem~\ref{skldjhfgslkdufghikl}. 

Statement~\ref{sdlvcsowujhvsloivch} of Lemma~\ref{lvdkfjhbvlfdkujh} is proved by using a similar argument. However, the condition that $f{\in}\enable{A}$ is used to ensure that $f{\in}\cM$ whatever the matching selected in the considered solution of Problem~\ref{likgusefhdrlpouhiouh}. Let $\rescaling{\!f}^{K,\cM}$ be the rescaling offset of $f$ for a pair $(K,\cM_K)$ being part of a solution of Problem~\ref{likgusefhdrlpouhiouh} satisfying goodness Conditions (\ref{jsdyhgcdfasjyht}--\ref{klifgyegrujytfgeikuy}). Suppose $\rescaling{1}\eqdef\rescaling{\!f}^{K_1,\cM_1}<\rescaling{2}\eqdef\rescaling{\!f}^{K_2,\cM_2}$ for two such pairs $(K_i,\cM_{\vsmall i}),i{=}1,2$. Let $\rescaled{f}_1=\vsmall^{\rescaling{1}}f$ be the rescaling of $f$ based on the first solution. Then, using the same notation, $\rescaled{f}_2=\vsmall^{\rescaling{2}}f=\vsmall^{\rescaling{2}-\rescaling{1}}\rescaled{f}_1$. Therefore, setting $\vsmall:={0}$ will zero $\rescaled{f}_2$ since $\rescaling{2}-\rescaling{1}>0$, thus making the corresponding equation a tautology. This, again, contradicts Statement~\ref{drtjhsrtjhrsth} of Theorem~\ref{skldjhfgslkdufghikl}.

We also prove Statement~\ref{hjtgfkujyhglk} of Lemma~\ref{lvdkfjhbvlfdkujh} by contradiction. If $K_1$ and $K_2$ are two variable heights, seen as functions, then $K_1{\neq}K_2$ iff either $K_1{\neq}\max(K_1,K_2)$ holds, or the opposite holds. Thus, it is enough to investigate one of the latter cases. Let $(K_1,\cM_{1})$ and $(K_2,\cM_{2})$ be two solutions such that $K_1{\leq}K_2,K_1{\neq}K_2$, and denote by $\rescaling{x}^1$ and $\rescaling{x}^2$ the rescaling offset of variable $x$ for solutions $1$ and $2$. Extend $(K_1,\cM_{1})$ until $K_2$ by adding, to $\cM_1$, the set of pairs 
	$$
	\left\{\left.(\pprimeppostset{c_{\!f}}{k}{f},\pprimeppostset{d_x}{k}{x})
	\;\right|\, (\edge)\in\cM,\,K_{1,f}<k\leq K_{2,f}
	\right\}\,.
$$
 Denote by $(K_2,\overline{\cM}_{1})$ this extension, and, as before, we extend rescaling offsets by setting 
\beq
\mbox{$\rescaling{x}^1=0$ if $x$ sits between $K_1$ and $K_2$.} 
\label{eloriguehgkliuh}
\eeq
By statement~\ref{jdhtfgfedsfm} of Lemma~\ref{lvdkfjhbvlfdkujh}, $\rescaling{x}$ is the same for both solutions:
\beq
\rescaling{x}^1=\rescaling{x}^2 \mbox{, for every dependent variable $x$ of }A_{K_2}\,. \label{lsfvhsgdblvk}
\eeq
Consider the tail of the mode change array of height $K_2$. The equations of this tail must all be matched in ${\cM}_{2}$. Let $f$ be an equation of the new mode with equation offset $c_{\!f}{>}0$ and let $x$ be the variable  matched with $f$ in the matching $\cM$ found in Problem~(\ref{liftuerhpituhepu8}). Then, $\pprimeppostset{(c_{\!f}-1)}{K_{2,f}}{f}$ is a consistency equation sitting in the tail, hence it is enabled and matched in $\cM_{2}$ with the state variable $\pprimeppostset{(d_x-1)}{K_{2,f}}{x}$, where $d_x$ is the variable offset of $x$ in Problem~(\ref{liftuerhpituhepu8}). Consider the variable $\pprimeppostset{d_x}{(K_{2,f\,}-1)}{x}$ sitting in the previous block. We have $\pprimeppostset{d_x}{(K_{2,f\,}-1)}{x}\sim\pprimeppostset{(d_x-1)}{K_{2,f}}{x}$, implying that $\pprimeppostset{d_x}{(K_{2,f\,}-1)}{x}$ cannot be matched with an equation in $\cM_{2}$. Therefore, it must be matched in $\cM_{2}$ with the Euler identity 
\[
0=\vsmall{\times}\pmatch{\pprimeppostset{d_x}{(K_{2,f}-1)}{x}}-\bigl(\pprimeppostset{(d_x-1)\,}{K_{2,f}}{x}-\pprimeppostset{(d_x-1)\,}{(K_{2,f}-1)}{x}\bigr)\,.
\]
Setting $z\eqdef\pprimeppostset{(d_x-1)}{(K_{2,f}-1)}{x}$, this Euler identity rewrites
\[
0=\vsmall{\times}\pmatch{\dot{z}}-(\postset{z}-z)\,,
\]
which yields the rescaling equation
\beqq
\rescaling{\pmatch{\dot{z}}}^2 &=& 1+\max\left(\rescaling{\postset{z}}^2,\rescaling{z}^2\right) ~\geq~1\,.
\eeqq
This contradicts (\ref{lsfvhsgdblvk}) since, by (\ref{eloriguehgkliuh}), $\rescaling{\dot{z}}^1=0$ should hold because $K_{2,f}-1\geq K_{1,f}$. This completes the proof of Lemma~\ref{lvdkfjhbvlfdkujh}, and thus proves Statement~\ref{jcsdghsdvckjhsdgk} of Theorem~\ref{skldjhfgslkdufghikl}.

\section{Addendum: solving nondeterminism and handling cascaded mode changes}  \label{sdlkjhfcklsujdg}
So far we considered mode changes between two successive ``long modes'' (lasting for a positive duration with the same DAE dynamics). Such mode change may not be sufficiently determined, in which case either the structural analysis of the array or Corollary~\ref{hgrfdjrgfdfkm} may provide appropriate diagnosis. In this section we consider active determinization, by which constraints are added to the mode change to determinize it. By doing so, we will also encompass cascades of transient (of zero-duration) modes, with limited objectives, however. We will illustrate this with the cup-and-ball example (\ref{loeifuhpwoui}) in the elastic impact case (\ref{hythfjnhgfedshgf}). 

Most of the previous material is reused with little change, if any. Pointers to \bemph{\textbf{reused}} material will be given, and emphasis will be put in \remph{red} on what is new.

\subsection{Revisiting hot restart}
\begin{definition}[mode change revisited]
	\label{sldkjvhcsdblkgh} A \emph{mode change} is defined as a tuple $(S^-,\changetime,\remph{T},S)$, where $S^-$ is the \emph{previous long mode}, $S$ is the \emph{new long mode}, 
	\remph{$T$ is the (optional) \emph{hot restart constraint}}, 
	and $\changetime{\,\in\,}\bR_+$ is the instant when the cascaded mode change occurs. We assume that: 
	\begin{itemize}
		\item $S^-{=}(F^-,X^-)$ and $S{=}(F,X)$ are square DAE systems; $S^-$ is index-reduced, i.e., $S^-=\indexreduced{S^-}$; $S$ is completed with its latent equations, i.e., $S=\complete{S}$; and
		\item $T{=}(F_T,X^-_T{\cup}\,Z\,{\cup}X_T)$ is a static system of equations where 
		\beqq
		X_T^-\subset\bigcup_{k\geq 0}\pprime{k}{{X^-}} &\mbox{and}& X_T\subset\bigcup_{k\geq 0}\pprime{k}{X}\,,
		\eeqq
		and $Z$ collect auxiliary variables.
		Sets $X^-_T$ and $Z\,{\cup}X_T$ collect the free and dependent variables of $T$, respectively. We require $T$ not to be structurally overdetermined.
	\end{itemize}
\end{definition}
Static system $T$ relates variables of the previous long mode to variables of the new long mode, thus contributing to the determination of the hot restart.

Recalling that $\complete{X}$ denotes the set collecting the leading and state variables of the new long mode, we consider the subset $X^+\subseteq\complete{X}$ of state variables of the new mode. The hot restart consists in finding values for all the variables belonging to $X^+$. 

\begin{nochange}\rm 
	\label{skldifujhgaliug} The following material is reused with no change:
\begin{itemize}
	\item Problem~\ref{dsjkchgfckajycgk} of hot restart, with the exception that ``previous/new mode'' is replaced by ``previous/new \emph{long} mode'';
	\item Assumption~\ref{hytdrfjytrjytr} regarding how the instant $\changetime$ of mode change is defined.\eproof
\end{itemize}
\end{nochange}
\paragraph*{Syntax for hot restart:} To specify a pair $(T,S)$ collecting the new long mode and its restart transient mode, we propose the following syntax:\footnote{This syntax conforms to Modelica, in which ``$\prog{when}\;\guard$'' denotes the onset of condition $\guard$.}
\beq
\bea{rl}
\when\negesp&\guard\;\doo\;S \\ \prog{and}\;\prog{when}\negesp&\guard\;\doo\;T
\eea
\label{sldjhvcgsdlkjhg}
\eeq

\subsection{Revisiting mode change arrays}
\begin{nochange}\rm 
	\label{sldkvjhgxvlkgh} The following material is reused with no change:
\begin{itemize}
	\item The introduction of Section~\ref{hgrfdkhgfcvkjmh} introducing our model of discrete time for the mode change;
	\item Section~\ref{jghtsdhgfsdhfgsd} regarding the toolbox relating shifting and differentiating.\eproof
\end{itemize}
\end{nochange}
The construction of arrays, however, must be revisited, so we detail it carefully. 
%
We use the notations from Section~\ref{nhmgfcvkjh}: for $F$ a DAE system, $\indexreduced{F},\consistency{F}$, and $\complete{F}$ refer to the ``index-reduced'', ``consistency'', and ``completion'' of $F$ once index reduction has been performed. When considering mode changes, we use notations (\ref{sjdcytsafcjsyutf},\ref{kduhgvskdug}) for variables and functions, and Notations~\ref{sdolifkjshgikldu} for shifting transition relations. Convention~\ref{slrikfgsjhfglsiugh} regarding $\sim$\,-\,completion is again in force.
\paragraph*{Mode change array:}
To the new long mode $S{=}(F,X)$ we again associate following (\ref{fdkjhsfkjhhj}) the infinite system $A$ of equations obtained by stacking successive shifts of the completion $\complete{F}$.	
The set of \emph{past variables} of $A$ is its subset of variables whose value is determined by the previous mode:
\beq
\bea{rcl}
\Vars^- = \bigl\{
x\in\Vars_A
\left| \;
\exists x^-\in\bigcup_{k<0}\ppostset{k}{(X^-)} \mbox{~ s.t. ~} x^-{\sim\,} x
\right.
\bigr\}\,.
\eea
\label{kdsjhgsdkjuhycg}
\eeq
The set of dependent variables of $A$ is
\beq
\dependent{\Vars_{A}} &\eqdef& \Vars_{A} \setminus \Vars^-\,,    \label{ldkghdliuga}
\eeq
where $\Vars_{A}$ is the set of \emph{all} the variables involved in $A$ (both leading and state variables). 
\begin{nochange}\rm 
	\label{saldifaglkjug} The following material is reused with no change:
	\begin{itemize}
		\item Lemma~\ref{jhgffdhghkjh} and Definition~\ref{dkjhgdkhgl} regarding the definition of facts;
		\item Convention~\ref{conv:remove-facts} still applies, i.e., facts are removed from $A$.\eproof
	\end{itemize}
\end{nochange}
We address Step~\ref{jytrgdhtsgdyntgfv} of Approach~\ref{isduyfsegdiksuygu} by introducing finite arrays, based on the notion of height.
A \emph{height} for mode $S=(F,X)$ denotes a function
\beq
\mbox{$K:F\ra\bN$,~~ denoted ~~$f\mapsto K_{\!f}$}\,. \label{slvkjhslksjlvhslk}
\eeq
In the following definition, $A$ denotes the infinite array (\ref{skdujcgskdjgyvckuj}), and its set $\dependent{\Vars_A}$ of dependent variables is defined by (\ref{ldkghdliuga}). The changes with reference to Definition~\ref{jsahgdfkjagdf} are pointed in \remph{red}:
\smallskip
\begin{definition}[mode change array]
	\label{lsdiuvhsliuhi} 
		 Any height $K$ for mode $S{=}(F,X)$ defines the $K$-\emph{mode change array}
		\beq
		&&A_K~=~\left\{\left.
		\pprimeppostset{m}{k}{f}
		\,\right|\,
		f{\in}{F}\,,\,0{\leq} m{\leq} c_{\!f}\,,\,0{\leq} k{\leq} K_{\!f}
		\right\}\remph{~\cup~ T_K}\,,\emph{ where}
		\label{lssialcjskuvhlkoj}
		\\
		&&\emph{
\begin{minipage}{12cm}
 $T_K$ is equal to $T$ in which each $\pprime{d}{x}\in\mmoins{X_T}$ is replaced by $\pprimeppostset{d}{(-1)}{x}$ and each $\pprime{d}{x}\in{X_T}$ is replaced by $\pprimeppostset{d}{K_{\!f}}{x}$, for $(\edge)\in\cM$, $\cM$ being the matching found by solving the \sigmamethod\ for $S$.
\end{minipage}
} 		\nonumber
	\eeq
The set of dependent variables of $A_K$ is $\dependent{\Vars_{A_K}}=\dependent{\Vars_A}{\cap}\Vars_{A_K}$. Define
	the tail of $A_K$ and its set of dependent variables by
	\beq\bea{lcl}
	\ttail{A_K} &=& \,\left\{ \pprimeppostset{m}{K_{\!f}}{f} \left|\bea{l}f{\in}{F} \\ 0\leq m\leq c_{\!f} \eea\negesp\right.\right\}\remph{~\cup~ T_K}
	\\ [4mm]
	\ttail{\Vars_{A_K}} &=& \left\{
	\pprimeppostset{d}{K_{\!f}}{x} \;\left|\bea{l}
	f{\in}{F},(\edge)\in\cM \\
	0\leq m\leq c_{\!f} \\
	d=\sigma_{\!\edge}+m
	\eea\negesp\right.\right\}
	\\ [6mm]
	\hhead{\Vars_{A_K}}&=& ~{\Vars_{A_K}}\setminus\ttail{\Vars_{A_K}}
	\eea 
	\label{ldfijnbvolugc}
	\eeq
We define the subset of equations that must be enabled by
		\beq\bea{rcl}
\enable{A_K}&=& 
\ttail{A_K} \cup \bigl(
A_K\cap\;\bigcup_{k\geq 0} \,\ppostset{k}{(F^-)}
\bigr)
\eea
\label{lisduvhsldhvliu}
\eeq
and we set $\disable{A_K} \eqdef A_K\,\setminus\enable{A_K}$.
\end{definition}
\begin{nochange}\rm
	\label{skdjuhfgliukg} The following material is reused with no change:
	\begin{itemize}
		\item Lemma~\ref{sldkjfhsblkj} guaranteeing the soundness of the definition of $T_K$ in (\ref{lssialcjskuvhlkoj}).\eproof
	\end{itemize}
\end{nochange}

\subsection{Reusing our previous analysis}
\begin{nochange}\rm
	\label{slivkusdgvlsidug} The following material is reused with no change:
	\begin{itemize}
		\item Structural implicit function theorem (Proposition~\ref{jsdhgfaksjyhcg} and Rule~\ref{ksjhcgskdjgh});
		\item Section~\ref{dfkjshlksgfskdjfmg} on the rescaling analysis of mode change arrays;
		\item Section~\ref{ldxkjvbvskljvh} collecting the main theorems, and their associated proofs.\eproof
	\end{itemize}
\end{nochange}
Indeed, once mode change arrays are properly adapted (see the \remph{modifications} in Definition~\ref{lsdiuvhsliuhi}), the above listed material can be reused with no change (except for how a mode change is named), since it only takes the mode change array as its input.

\subsection{Cup-and-ball with elastic impact}
We discuss here the cup-and-ball example of Section~\ref{jhtgfkuytjrsdfghk} in the case of elastic impact, see (\ref{hythfjnhgfedshgf}). 

\paragraph*{A specification where transience is made explicit:}
The model (\ref{loeifuhpwoui}) is modified as follows:
\beq\left\{\bea{rlll}
& 0= \ddot{x}+{\tension}x & (\eqq_1) \\
& 0= \ddot{y}+{\tension}y+g  & (\eqq_2) \\
& {{\guard}= [s^-\leq{0}]; \guard(0)=\fff}    \\ [2mm]
\prog{when} \; \guard\; \doo&0 = {L^2}{-}{(x^2{+}y^2)} & {({\straight_0})} & \mbox{(hot restart)}\\
\prog{and}& 0={\tension}+{s} &({\straight_1})  & \mbox{(hot restart)} \\
[2mm]
\prog{else}& 0=\tension   & (\straight_3) \\
\prog{and}& 0=({L^2}{-}(x^2{+}y^2))-s   & (\straight_4)
\eea\right.
\label{ksefudygskuedyfgk}
\eeq
Model (\ref{ksefudygskuedyfgk}) is obtained from (\ref{loeifuhpwoui}) by replacing the long guard ``$\when\;\guard$'' by the transient guard ``$\prog{when}\;\guard$''. We consider the cascade $\guard:\fff\ra\ttt\ra\fff$, where restart occurs. We regard the transient mode $\guard{=}\ttt$ as contributing, by being the ``$T$'' player, to the hot restart of the new long mode $\guard{=}\fff$ following Definition~\ref{sldkjvhcsdblkgh}. Following the construction (\ref{lsdiuvhsliuhi}), the resulting mode change array $A_0$ is 
\beq
A_0=\left\{\bea{lll}
0= {\ddot{x}}+{\tension}x & (\eqq_1) & \guard{=}\fff,\changetime \\
0= {\ddot{y}}+{\tension}y+g  & (\eqq_2) \\ 
{0=\tension} &{(\straight_3)} & \\
0=({L^2}{-}(x^2{+}y^2))-{s} & (\straight_4) 
\\ [2mm]
\gremph{0 = {L^2}{-}{(x^2{+}y^2)}} & \gremph{({\straight_0})} & \mbox{(hot restart) \gremph{fact}}\\
0={\tension}+{s} &({\straight_1}) & \mbox{(hot restart)}
\eea\right.
\label{lrsiugfshkvjudsgh}
\eeq
Array $A_0$ consists of two blocks. The first block collects the dynamics of the new long mode $\guard{=}\fff$, snapshot at $\changetime$. The second block collects the transient dynamics $\guard{=}\ttt$ at the same instant $\changetime$. Note that \emph{\normalsize no latent equation was added:} it makes no sense to add them since a transient mode takes zero duration whereas differentiations tightly relate to time shifts as we have seen in Section~\ref{jghtsdhgfsdhfgsd}.

Equation $\gremph{(k_0)}$ is a fact, since it is the zero-crossing triggering the mode change. The dependent variables of this array are $\tension,s,\ddot{x},\ddot{y}$, whereas $x,y,\dot{x},\dot{y}$ are the past variables. $A_0$ is thus structurally singular with $4$ dependent variables and $5$ equations. Actually, the subsystem $(k_3,k_4,k_1)$, which involves only the two dependent variables $\tension,s$, is conflicting. Hence, one of these equations must be disabled, which contradicts the definition (\ref{lisduvhsldhvliu}) for $\enable{A_0}$. We thus need to consider $A_1$:
\beq
A_1=\left\{\bea{lll}
0= \pmatch{\ddot{x}}+{\tension}x & (\eqq_1) & \guard{=}\fff,\changetime \\
0= {\ddot{y}}+\pmatch{\tension}y+g  & (\eqq_2) \\ 
\remph{0=\tension} &\remph{(\straight_3)} & \remph{\mbox{conflict}}\\
0=({L^2}{-}(x^2{+}y^2))-\pmatch{s} & (\straight_4) 
\\ [2mm]
0= \vsmall{\times}\pmatch{\ddot{y}}-(\postset{\dot{y}}{-}\dot{y})
& {(\euler{\pmatch{\ddot{y}}}{\postset{\dot{y}}})} 
\\ 
0= \vsmall{\times}{\ddot{x}}-(\pmatch{\postset{\dot{x}}}{-}\dot{x})
& {(\euler{{\ddot{x}}}{\pmatch{\postset{\dot{x}}}})} 
\\  [2mm]
0= \pmatch{\postset{\ddot{x}}}+\postset{\tension}\postset{x} & (\postset{\eqq_1}) & \guard{=}\fff,\changetime{+}\vsmall \\
0= \pmatch{\postset{\ddot{y}}}+{\postset{\tension}}\postset{y}+g  & (\postset{\eqq_2}) \\ 
\remph{0={\postset{\tension}}} & \remph{(\postset{\straight_3})} & \remph{\mbox{conflict}} \\
0=({L^2}{-}\postset{(x^2{+}y^2)})-\pmatch{\postset{s}} & (\postset{\straight_4})  \\ [2mm]
\gremph{0 = {L^2}{-}\postset{(x^2{+}y^2)}} & \gremph{({\straight_0})} & \mbox{(hot restart) }\gremph{\mbox{fact}} \\
0=\pmatch{\postset{\tension}}+{\postset{s}} & {({\straight_1})} &\mbox{(hot restart) } \\
\eea\right.
\label{kfuyhgakuyag}
\eeq
We pinpoint in \remph{red} the conflicting equations. They belong to $\disable{A_1}$, so we disable them. Then, we show in \pmatch{blue} an equation-complete matching. Unfortunately, the dependent variable $\postset{\dot{y}}$ is not matched. Further augmenting the array height will not solve this problem. The hot restart specified by model (\ref{ksefudygskuedyfgk}) was insufficiently determined. To cope with this, we propose adding the following fully elastic Newton impact law as an additional hot restart constraint following Definition~\ref{sldkjvhcsdblkgh}:
\beq\left\{\bea{rllr}
& 0= \ddot{x}+{\tension}x & (\eqq_1) \\
& 0= \ddot{y}+{\tension}y+g  & (\eqq_2) \\
& {{\guard}= [s^-\leq{0}]; \guard(0)=\fff}    \\ [2mm]
\prog{when} \; \guard\; \doo&0 = {L^2}{-}{(x^2{+}y^2)} & {({\straight_0})} & \mbox{(hot restart)} \\
\prog{and}& 0={\tension}+{s} &({\straight_1}) & \mbox{(hot restart)} \\
\prog{and}& 0={(x\dot{x}{+}y\dot{y})} {+}\mmoins{(x\dot{x}{+}y\dot{y})} & ({\cN}) & \mbox{(additional hot restart)} \\ [2mm]
\prog{else}& 0=\tension   & (\straight_3) \\
\prog{and}& 0=({L^2}{-}(x^2{+}y^2))-s   & (\straight_4)
\eea\right.
\label{liushlfvsiufhvsk}
\eeq
With reference to Definition~\ref{sldkjvhcsdblkgh} of cascaded mode change and its pair $(T,S)$, equations $({\straight_0}),({\straight_1}),({\cN})$ constitute the $T$-part whereas the $\prog{else}$ branch augmented with the always active equations $(\eqq_1),(\eqq_2)$, constitutes the $S$-part.
We consider the cascade $\guard:\fff\ra\ttt\ra\fff$, where restart occurs. Its mode change array $A_1$ with matching $\cM$ highlighted in \bemph{blue} is
\beq
A_1=\left\{\bea{lll}
0= \pmatch{\ddot{x}}+{\tension}x & (\eqq_1) & \guard{=}\fff,\changetime \\
0= {\ddot{y}}+\pmatch{\tension}y+g  & (\eqq_2) \\ 
\remph{0=\tension} &\remph{(\straight_3)} & \remph{\mbox{conflict}}\\
0=({L^2}{-}(x^2{+}y^2))-\pmatch{s} & (\straight_4) 
\\ [2mm]
0= \vsmall{\times}\pmatch{\ddot{y}}-(\postset{\dot{y}}{-}\dot{y})
& {(\euler{\pmatch{\ddot{y}}}{\postset{\dot{y}}})} 
\\ 
0= \vsmall{\times}{\ddot{x}}-(\pmatch{\postset{\dot{x}}}{-}\dot{x})
& {(\euler{{\ddot{x}}}{\pmatch{\postset{\dot{x}}}})} 
\\  [2mm]
0= \pmatch{\postset{\ddot{x}}}+\postset{\tension}\postset{x} & (\postset{\eqq_1}) & \guard{=}\fff,\changetime{+}\vsmall \\
0= \pmatch{\postset{\ddot{y}}}+{\postset{\tension}}\postset{y}+g  & (\postset{\eqq_2}) \\ 
\remph{0={\postset{\tension}}} & \remph{(\postset{\straight_3})} & \remph{\mbox{conflict}} \\
0=({L^2}{-}\postset{(x^2{+}y^2)})-\pmatch{\postset{s}} & (\postset{\straight_4})  \\ [1mm]
\gremph{0 = {L^2}{-}\postset{(x^2{+}y^2)}} & \gremph{({\straight_0})} & \gremph{\mbox{fact}} \\
0=\pmatch{\postset{\tension}}+{\postset{s}} & {({\straight_1})}  \\
0=\postset{(x\dot{x}{+}y\pmatch{\dot{y}})} {+}\ppostset{(-1)}{(x\dot{x}{+}y{\dot{y}})} & ({\cN}) 
\eea\right.
\label{ksahcgakjc}
\eeq
The three equations constituting the ``$\prog{when}$'' branch sit at the end of the array, with the due replacements following (\ref{skudjcghskdjhcg}). Note that $\gremph{({\straight_0})}$ is a fact. The four equations constituting the ``$\prog{else}$'' branch occur at both instants $\changetime$ and $\changetime{+}\vsmall$. The needed Euler identities are included. The two equations $\remph{({\straight_3})},\remph{(\postset{\straight_3})}$ are conflicting with the rest of the array: since disabling them is allowed, they get disabled. A perfect matching is shown, showing structural nonsingularity of the black part of the array. The useful subarray (obtained by discarding leading equations of the last block) is:
\beq
A_1=\left\{\bea{lll}
0= \pmatch{\ddot{x}}+{\tension}x & (\eqq_1) & \guard{=}\fff,\changetime \\
0= {\ddot{y}}+\pmatch{\tension}y+g  & (\eqq_2) \\ 
0=({L^2}{-}(x^2{+}y^2))-\pmatch{s} & (\straight_4) 
\\ [2mm]
0= \vsmall{\times}\pmatch{\ddot{y}}-(\postset{\dot{y}}{-}\dot{y})
& {(\euler{\pmatch{\ddot{y}}}{\postset{\dot{y}}})} 
\\ 
0= \vsmall{\times}{\ddot{x}}-(\pmatch{\postset{\dot{x}}}{-}\dot{x})
& {(\euler{{\ddot{x}}}{\pmatch{\postset{\dot{x}}}})} 
\\  [2mm]
0=\postset{(x\dot{x}{+}y\pmatch{\dot{y}})} {+}\ppostset{(-1)}{(x\dot{x}{+}y{\dot{y}})} & ({\cN}) 
\eea\right.
\label{dfoivhdfloijh}
\eeq
The corresponding rescaling analysis is:
\beq\bea{rclcl}
\rescaling{\ddot{x}} &\negesp=\negesp& \rescaling{e_1} &\negesp=\negesp& \max(\rescaling{\ddot{x}},\rescaling{\tension}) \\
\rescaling{\tension} &\negesp=\negesp& \rescaling{e_2} &\negesp=\negesp& \max(\rescaling{\ddot{y}},\rescaling{\tension}) \\
\rescaling{s} &\negesp=\negesp& \rescaling{k_4}  \\
\rescaling{{\ddot{y}}} &\negesp=\negesp& \rescaling{\euler{\pmatch{\ddot{y}}}{\postset{\dot{y}}}} &\negesp=\negesp& \max(\rescaling{{\ddot{y}}},1+\rescaling{\postset{\dot{y}}}) \\
1{+}\rescaling{\postset{\dot{x}}} &\negesp=\negesp& \rescaling{\euler{{\ddot{x}}}{\pmatch{\postset{\dot{x}}}}} &\negesp=\negesp& \max(\rescaling{{\ddot{x}}},1+\rescaling{\postset{\dot{x}}}) \\
\rescaling{\postset{\dot{y}}} &\negesp=\negesp& \rescaling{{\cN}} &\negesp=\negesp&\max(\rescaling{\postset{\dot{y}}},\rescaling{\postset{\dot{x}}})
\eea
\label{sldufgvshdlfouh}
\eeq
Its good solution is:
\[
\rescaling{\ddot{x}}=\rescaling{\tension}=\rescaling{\ddot{y}}= 1 = \rescaling{\euler{\pmatch{\ddot{y}}}{\postset{\dot{y}}}}= \rescaling{\euler{{\ddot{x}}}{\pmatch{\postset{\dot{x}}}}}=\rescaling{e_2}= \rescaling{e_1}
\]
whereas other variables and functions have rescaling offset zero. Rescaling variables and equations using (\ref{gtrfdkyjrgs}) yields
\beqq\bea{rclcrcrcl}
\rescaled{\ddot{x}} &\negesp\eqdef\negesp& \vsmall{\times}\ddot{x} &\negesp=\negesp& \postset{\dot{x}}{-}\dot{x} &\negesp;\negesp&
\rescaled{e_1} &\negesp\eqdef\negesp& \vsmall{\times}{e_1} 
\\ [1mm]
\rescaled{\ddot{y}} &\negesp\eqdef\negesp& \vsmall{\times}\ddot{y} &\negesp=\negesp& \postset{\dot{y}}{-}\dot{y}  &\negesp;\negesp&
\rescaled{\euler{\pmatch{\ddot{y}}}{\postset{\dot{y}}}} &\negesp\eqdef\negesp& \vsmall{\times}{\euler{\pmatch{\ddot{y}}}{\postset{\dot{y}}}} 
\\ [1mm]
\rescaled{\tension} &\negesp\eqdef\negesp& \vsmall{\times}\tension && &\negesp;\negesp&
\rescaled{e_2} &\negesp\eqdef\negesp& \vsmall{\times}{e_2} 
\\ [1mm]
 &&  && &;& \rescaled{\euler{{\ddot{x}}}{\pmatch{\postset{\dot{x}}}}} &\negesp\eqdef\negesp& \vsmall\times{\euler{{\ddot{x}}}{\pmatch{\postset{\dot{x}}}}} 
\eea 
\eeqq
The rescaled array $\rescaled{A_1}$ is then:
\beqq
\rescaled{A_1}:\left\{\bea{cclcll}
 0&\!\!\!\!=\!\!\!\!& 
\pmatch{\rescaled{\ddot{x}}}
+{\rescaled{\tension}}x 
\\ 
 0&\!\!\!\!=\!\!\!\!& 
\rescaled{\ddot{y}}
+\pmatch{\rescaled{\tension}}y
\\
0&\!\!\!\!=\!\!\!\!&({L^2}{-}(x^2{+}y^2))-\pmatch{s}
\\
[1mm]
0&\!\!\!\!=\!\!\!\!& \rescaled{\ddot{x}}-(\pmatch{\postset{\dot{x}}}-\dot{x})
\\
0&\!\!\!\!=\!\!\!\!& \pmatch{\rescaled{\ddot{y}}}-({\postset{\dot{y}}}-\dot{y})
\\ [1mm]
{0} &\!\!\!\!{=}\!\!\!\!& (\postset{x}\postset{\dot{x}}{+}\postset{y}\pmatch{\postset{\dot{y}}}){+}{\ppostset{(-1)}{(x\dot{x}{+}y{\dot{y}})}}
\eea\right.
\eeqq
and the resulting (expected) restart system is obtained by renaming as in (\ref{erlfgiuehrliueh}):
\beq
\mbox{restart}:\left\{\bea{cclcll}
 0&\!\!\!\!=\!\!\!\!& 
\pmatch{\pplus{{\rescaled{\ddot{x}}}}}
+\pplus{{\rescaled{\tension}}}\mmoins{x} 
\\ 
 0&\!\!\!\!=\!\!\!\!& 
\pplus{{\rescaled{\ddot{y}}}}
+\pmatch{\pplus{{\rescaled{\tension}}}}\mmoins{y}
\\
0&\!\!\!\!=\!\!\!\!&({L^2}{-}\mmoins{(x^2{+}y^2)})-\pmatch{\pplus{s}}
\\
[1mm]
0&\!\!\!\!=\!\!\!\!& \pplus{{\rescaled{\ddot{x}}}}-(\pmatch{\pplus{\dot{x}}}-\mmoins{\dot{x}})
\\
0&\!\!\!\!=\!\!\!\!& \pmatch{\pplus{{\rescaled{\ddot{y}}}}}-(\pplus{{{\dot{y}}}}-\mmoins{\dot{y}})
\\ [1mm]
{0} &\!\!\!\!{=}\!\!\!\!& (\mmoins{x}\pplus{\dot{x}}{+}\mmoins{y}\pmatch{\pplus{\dot{y}}}){+}{\mmoins{(x\dot{x}{+}y{\dot{y}})}}
\eea\right.
\label{hytfglioihjolk}
\eeq

\paragraph*{An inconsistent specification:} By guarding the mode $\guard{=}\ttt$ with $\prog{when}$ only, model (\ref{liushlfvsiufhvsk}) directly specifies that the mode $\guard{=}\ttt$ is transient. 
One could, alternatively, consider the following model:
\beq\left\{\bea{rll}
& 0= \ddot{x}+{\tension}x & (\eqq_1) \\
& 0= \ddot{y}+{\tension}y+g  & (\eqq_2) \\
& {{\guard}= [s^-\leq{0}]; \guard(0)=\fff}    \\ [2mm]
\when \; \guard\; \doo&0 = {L^2}{-}{(x^2{+}y^2)} & {({\straight_0})} \\
\prog{and}& 0={\tension}+{s} &({\straight_1}) \\
\prog{and}\;\prog{when} \; \guard\; \doo&0 = {(x\dot{x}{+}y\dot{y})} {+}\mmoins{(x\dot{x}{+}y\dot{y})} & ({\cN}) \\ [2mm]
\prog{else}& 0=\tension   & (\straight_3) \\
\prog{and}& 0=({L^2}{-}(x^2{+}y^2))-s   & (\straight_4)
\eea\right.
\label{sldifvuhsdlikfvuh}
\eeq
With reference to model (\ref{liushlfvsiufhvsk}), in (\ref{sldifvuhsdlikfvuh}) equations $(k_0),(k_1)$ were shifted, from the $T$ part of the new mode to its $S$ part, see Definition~\ref{sldkjvhcsdblkgh}. Since mode $\guard{=}\ttt$ is possibly long, we must consider the mode change $\guard:\fff\ra\ttt$. Also, following (\ref{fdkjhsfkjhhj}), latent equations must be added to (\ref{ksahcgakjc}) when forming the array:
\beq
A_1=\left\{\bea{lll}
0= \pmatch{\ddot{x}}+{\tension}x & (\eqq_1) & \guard{=}\ttt,\changetime \\
0= {\ddot{y}}+\pmatch{\tension}y+g  & (\eqq_2) \\ 
\gremph{0={L^2}{-}(x^2{+}y^2)} & \gremph{(\straight_4)} & \gremph{\mbox{fact}} \\
\remph{0=x\dot{x}+y\dot{y}} & \remph{(\dot{\straight_4})} & \remph{\mbox{conflict}} \\
\remph{0=x\ddot{x}+\dot{x}^2+\dot{y}^2+y\ddot{y}} & \remph{(\ddot{\straight_4})} & \remph{\mbox{conflict}}
\\ [2mm]
0= \vsmall{\times}\pmatch{\ddot{y}}-(\postset{\dot{y}}{-}\dot{y})
& {(\euler{\pmatch{\ddot{y}}}{\postset{\dot{y}}})} 
\\ 
0= \vsmall{\times}{\ddot{x}}-(\pmatch{\postset{\dot{x}}}{-}\dot{x})
& {(\euler{{\ddot{x}}}{\pmatch{\postset{\dot{x}}}})} 
\\ [2mm]
\gremph{0={L^2}{-}\postset{(x^2{+}y^2)}} & \gremph{(\postset{\straight_4})}  & \gremph{\mbox{fact}}
\\ 
\remph{0=\postset{(x\dot{x}+y\dot{y})}} & \remph{(\postset{\dot{\straight_4}})}  & \remph{\mbox{conflict}} \\
0=\postset{(x\dot{x}{+}y\pmatch{\dot{y}})} {+}\ppostset{(-1)}{(x\dot{x}{+}y{\dot{y}})} & ({\cN}) & \changetime{+}\vsmall
\eea\right.
\label{ldfvuhsdliduvh}
\eeq
Note that equations $\remph{(\postset{\dot{\straight_4}})}$ and $({\cN})$ are in conflict. Solving this conflict requires disabling one of them, which is forbidden since both belong to the set of must-be-enabled equations by (\ref{lisduvhsldhvliu}). Augmenting the array will not change this. The conclusion is that model (\ref{sldifvuhsdlikfvuh}) is inconsistent due to the conjunction of ``$\when\;\guard$'' and ``$\prog{when}\;\guard$''. This is indeed expected from physical standpoint: one cannot have both elastic impact $({\cN})$ and mode $\guard{=}\ttt$ being long, meaning inelastic impact.

\paragraph*{A redundant specification of the impact law:}
What happens, however, if we explicitly add as $\bemph{({\cN})}$ a law modeling inelastic impact?
\beq\left\{\bea{rll}
& 0= \ddot{x}+{\tension}x & (\eqq_1) \\
& 0= \ddot{y}+{\tension}y+g  & (\eqq_2) \\
& {{\guard}= [s^-\leq{0}]; \guard(0)=\fff}    \\ [2mm]
\when \; \guard\; \doo&0 = {L^2}{-}{(x^2{+}y^2)} & {({\straight_0})} \\
\prog{and}& 0={\tension}+{s} &({\straight_1}) \\
\prog{and}\;\prog{when} \; \guard\; \doo& \bemph{0 = {(x\dot{x}{+}y\dot{y})}}  & \bemph{({\cN})} \\ [2mm]
\prog{else}& 0=\tension   & (\straight_3) \\
\prog{and}& 0=({L^2}{-}(x^2{+}y^2))-s   & (\straight_4)
\eea\right.
\label{vdlvojhsdvouh}
\eeq
Then array $A_1$ of (\ref{ldfvuhsdliduvh}) is modified as follows in its last two equations:
\beq
A_1=\left\{\bea{lll}
0= \pmatch{\ddot{x}}+{\tension}x & (\eqq_1) & \guard{=}\ttt,\changetime \\
0= {\ddot{y}}+\pmatch{\tension}y+g  & (\eqq_2) \\ 
\gremph{0={L^2}{-}(x^2{+}y^2)} & \gremph{(\straight_4)} & \gremph{\mbox{fact}} \\
\remph{0=x\dot{x}+y\dot{y}} & \remph{(\dot{\straight_4})} & \remph{\mbox{conflict}} \\
\remph{0=x\ddot{x}+\dot{x}^2+\dot{y}^2+y\ddot{y}} & \remph{(\ddot{\straight_4})} & \remph{\mbox{conflict}}
\\ [2mm]
0= \vsmall{\times}\pmatch{\ddot{y}}-(\postset{\dot{y}}{-}\dot{y})
& {(\euler{\pmatch{\ddot{y}}}{\postset{\dot{y}}})} 
\\ 
0= \vsmall{\times}{\ddot{x}}-(\pmatch{\postset{\dot{x}}}{-}\dot{x})
& {(\euler{{\ddot{x}}}{\pmatch{\postset{\dot{x}}}})} 
\\ [2mm]
\gremph{0={L^2}{-}\postset{(x^2{+}y^2)}} & \gremph{(\postset{\straight_4})}  & \gremph{\mbox{fact}}
\\ 
{0=\postset{(x\dot{x}+y\dot{y})}} & {(\postset{\dot{\straight_4}})}  &  \\
0=\postset{(x\dot{x}+y\pmatch{\dot{y}})}  & ({\cN}) & \changetime{+}\vsmall
\eea\right.
\label{sldivujshdlfvuih}
\eeq
The situation is now different: equations $(\postset{\dot{\straight_4}})$ and $({\cN})$ are syntactically identical, showing redundancy. So we could fuse them by discarding $({\cN})$, which brings us back to the inelastic impact case handled in Section~\ref{jhtgfkuytjrsdfghk}. This is a clear and easy case.

However, depending on how the impact law is formulated, it could be redundant with the first latent equation, but still different from it in its syntax. Checking redundancy could be performed at run time: $({\cN})$ is discarded at compile time and evaluated at run time to check that $(\postset{\dot{\straight_4}})$ and $({\cN})$ take equal values. Performing the same checking on array (\ref{ldfvuhsdliduvh}) would violate redundancy checking.

\subsection{On the modeling power of our approach}
To understand the power and limits of our approach, we consider the Newton cradle displayed in Fig.\,\ref{slvdikjfvhbcslvkujh}.
\begin{figure}[!h]
\centerline{\includegraphics[width=0.4\textwidth]{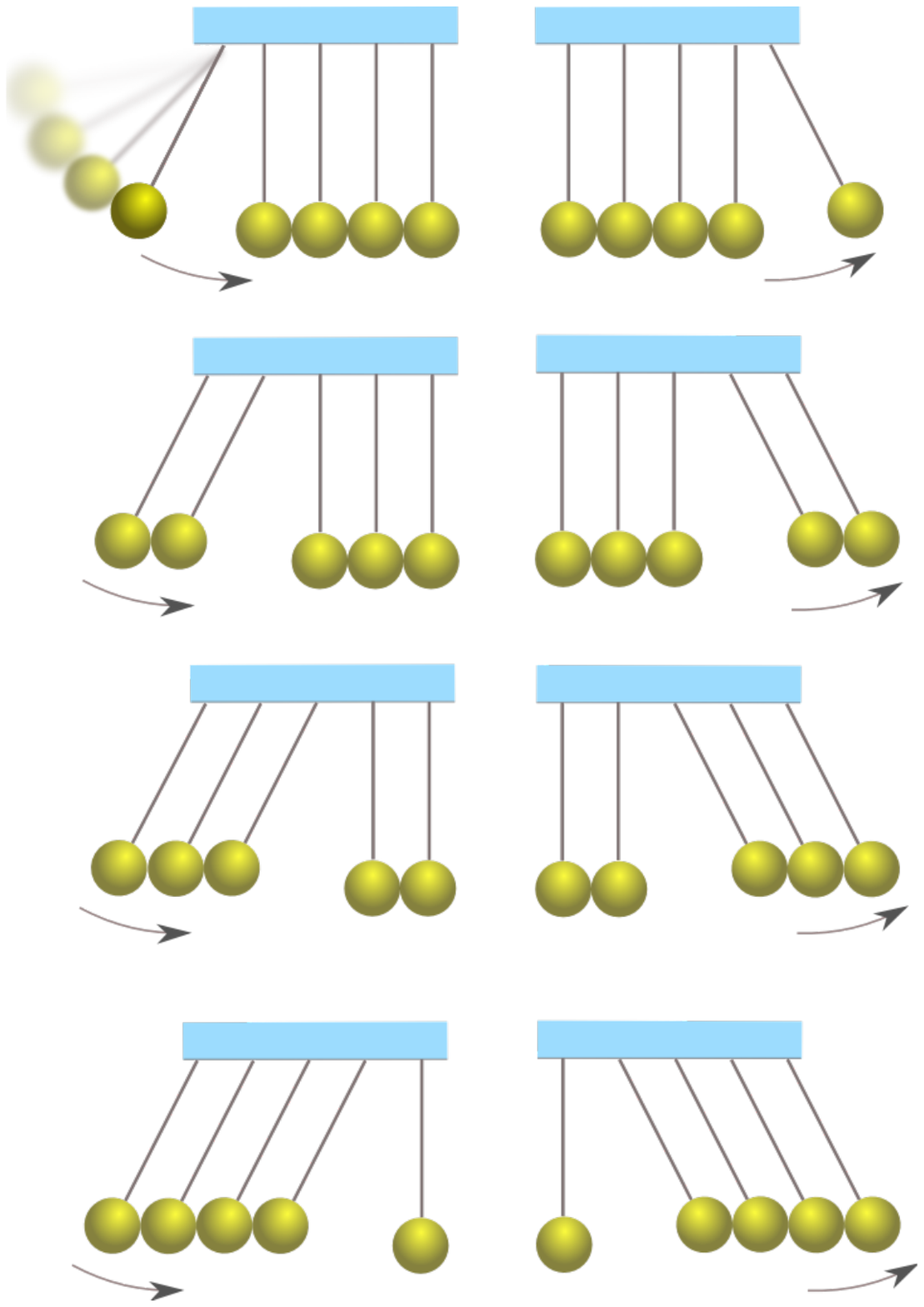}}
\caption{\sf Newton cradle. Source: \url{https://fr.wikipedia.org/wiki/Pendule_de_Newton}.}
	\label{slvdikjfvhbcslvkujh} 
\end{figure}

We assume elastic impact. The transient modes are the impacts. Long modes inherit an index two from the pendulum in Cartesian coordinates. In performing the modeling of this device, the difficult part is the model of the impacts, see, e.g., the interesting discussion in Section 4 of~\cite{DBLP:journals/tcps/Lee16}. The case of two balls is easy since the two tangent velocities resulting from the impact are determined by using two laws, namely the conservation of momentum and the conservation of kinetic energy. The resulting model can then be handled as the cup-and-ball in the elastic impact case. The modeling is more subtle if a larger number of balls is considered. Using only the above two laws results in nondeterminism. The analysis must be complemented by considering the potential energy associated to deformations and its propagation through waves. To circumvent such difficulties, several authors determinize their model by using smooth but stiff versions of the contact---the resulting simulations highly vary with how this smoothing is performed~\cite{DBLP:journals/tcps/Lee16}.

The impact modeling problem is indeed by itself a subject matter~\cite{GILARDI20021213}. Clearly, our approach provides no help for discovering the impact model. We only offer means to specify the impact model once it is known. Admittedly, this impairs modularity of modeling, since one can hardly expect being able to specify a single pendulum with its general context, and combine such models to get a model for the Newton cradle.

\section{Conclusion and perspectives}
We contributed to the mathematical soundness of DAE-based modeling languages such as Modelica, by formally solving the hot restart problem for a multimode DAE. We targeted modeling languages that are physics agnostic. Accordingly, our approach is physics agnostic. Despite being physics agnostic, our solution for the hot restart preserves (possibly latent) invariants.

We handle impulsive behaviors with no restriction. Our approach generates hot restart for general nonlinear DAE provided that the non-polynomial part does not interfere with the impulsive part---this condition is checked by our algorithm.\footnote{The class of multimode DAE for which hot restart is generated can be widened. By allowing for rational offsets, our rescaling analysis can be extended to handle the partition polynomial/nonpolynomial (instead of linear/nonlinear). Variables inheriting a non-integer offset can still be handled, provided that they do not need to be integrated or differentiated using Euler identities. It is not clear how useful this widening would be in practice.}

Our approach was first developed for mode changes involving two successive long modes. We extended it to handle the hot restart for \emph{finite cascades of transient mode changes}, such as occurring in the cup-and-ball with elastic impact.\footnote{This mode change takes the form free-motion$\,\ra\,$straight-rope$\,\ra\,$free-motion, with zero time spent in the ``straight-rope'' transient mode.} This was achieved by introducing additional hot restart constraints, and then reusing with little changes the analysis developed in Sections~\ref{jashcgfackjhamgfjh}--\ref{skduhvgskjvghkjh} for mode changes separating long modes.

Our approach builds on the principles of \emph{structural analysis} abstracting numerical equations as their incidence graphs. Our algorithm for synthesizing the hot restart bears similarities with the \sigmamethod\ for the index reduction of (single mode) DAE systems. 

So far we solved the hot restart problem for a mode change considered in isolation. However, a multimode DAE system has a large number of modes, its number of mode changes is even larger. Enumerating mode changes is the beyond reach. One way of overcoming this consists in reducing the set of mode changes by specifying assertions in the form of invariants. A complementary way is to reuse the techniques proposed in~\cite{electronics11172755} for the all-modes-at-once structural analysis of multimode DAE systems.

We proved in~\cite{Caillaud2020a,benveniste:hal-04295096,caillaud:hal-05257001} that a notion of structural interface can be associated to any DAE system having more dependent variables than equations\,---\,the mathematical model for a model class in Modelica. With this notion of structural interface, index reduction can be performed in a modular way (i.e., at the class level) instead of globally\,---\,as performed in DAE based tools today. This allows an impressive scaling-up for the compilation of Modelica models~\cite{caillaud:hal-05257001}, by handling models with millions of equations. From the above mentioned similarity, we expect that the structural interface can be extended to include mode change information. This is in contrast with the computational cost of methods using the Quasi-Weierstrass decomposition, not to mention the use of elimination which becomes totally prohibitive when polynomial systems are considered. The way forward is clear: modular methods must be developed for the hot restart of mode changes. 

\clearpage
\bibliographystyle{plain}
\bibliography{modelica.bib}
\clearpage
\appendix
\section*{Appendices}
In these appendices, missing proofs are collected and additional illustrative examples are developed.

\section{Proof of Corollary~\ref{hgrfdjrgfdfkm}}
\label{hjgfgkkjmhnm}
In the proof of Statement~\ref{drtjhsrtjhrsth} of Theorem~\ref{skldjhfgslkdufghikl}, goodness Conditions (\ref{jsdyhgcdfasjyht}--\ref{klifgyegrujytfgeikuy}) play a role only in the proof of Substatement~d). Substitutions (\ref{kjhyfgliukjk}), however, restore the validity of Conditions (\ref{jsdyhgcdfasjyht}--\ref{klifgyegrujytfgeikuy}) for this substatement. 

\section{Proof of Lemma~\ref{jfgklujlkughlkj}} \label{kcsjgvcsldjghkj}
 Writing $\rescaling{i}$ instead of $\rescaling{f_i}$ for $i{=}1,2$, and using Lemma~\ref{htdfkfkjhgssssghch}, we have, by setting $\widetilde{\mu}_i\eqdef\max(\rescaling{1},\rescaling{2}){-}\rescaling{i}$:
\beq\bea{rcr}
\Restart{f_{1}{+}f_{2}}&=&\rescaled{(f_{1}{+}f_{2})}[\vsmall:= 0]
\\ [1mm]
&=& \bigl(\vsmall^{\rescaling{(f_{1}{+}f_{2})}}(f_{1}{+}f_{2})\bigr)[\vsmall:= 0]
\\ [1mm]
&=& \bigl(\vsmall^{\max(\rescaling{1},\rescaling{2})}(f_{1}{+}f_{2})\bigr)[\vsmall:= 0]
\\ [1mm]
&=& \bigl(\vsmall^{\widetilde{\mu}_1}\vsmall^{\rescaling{1}}f_{1}+\vsmall^{\widetilde{\mu}_2}\vsmall^{\rescaling{2}}f_{2}\bigr)[\vsmall:= 0]
\\ [1mm]
&=& \sum_{i:\rescaling{i}=\max(\rescaling{1},\rescaling{2})}\Restart{f_{i}}\,,\hspace*{6mm}
\eea \label{gfesuhtrjysdm}
\eeq
and
\beq\bea{rcr}
\Restart{f_{1}{\times} f_{2}}&=&\rescaled{(f_{1}{\times} f_{2})}[\vsmall:= 0]
\\ [1mm]
&=& \bigl(\vsmall^{\rescaling{(f_{1}{\times} f_{2})}}(f_{1}{\times} f_{2})\bigr)[\vsmall:= 0]
\\ [1mm]
&=& \bigl(\vsmall^{(\rescaling{1}{+}\rescaling{2})}(f_{1}{\times} f_{2})\bigr)[\vsmall:= 0]
\\ [1mm]
&=& \vsmall^{\rescaling{1}}f_{1}[\vsmall:= 0]\times \vsmall^{\rescaling{2}}f_{2}[\vsmall:= 0]
\\ [1mm]
&=&\Restart{f_{1}}\times \Restart{f_{2}}\,.\hspace*{20.5mm}
\eea  \label{hjytgfkutjfghkjyhm} 
\eeq
The lemma is proved.

\section{Proof of Theorem~\ref{sekdfuysageloiug}}
\label{slviwsjrhfvluikvhk}
This proof will require a more in-depth analysis of the structure of matchings $\cM_\vsmall$, which in turn requires an extended analysis of Euler identities and their role in matchings. As a consequence, we will need to handle $\sim$\,--\,closures with care, thus we will indicate them explicitly (alternative case of Convention~\ref{skdjhasdgfckjsadhg}).

The mode change arrays $A_\vsmall$ considered in Problem~\ref{likgusefhdrlpouhiouh} are $\sim$\,--\,closed. To further study their matchings, it will be convenient to handle equivalence classes of variables modulo $\sim$. Recalling Definition~\ref{jhtgfjhgrsgdfk}, $\simquotient{x}$ denotes the equivalence class of $x$ modulo $\sim$. Relation $\sim$ is stable under both differentiation and shifting. Hence we can define
\beqq
\dot{{\simquotient{x}}}\eqdef\simquotient{{\,\dot{x}\,}} &\mbox{and}&
\postset{{\simquotient{x}}}\eqdef\simquotient{{\postset{x}}}
\eeqq
Let $S=(F,X),X\subseteq\allVars$ be a system, and set $\simquotient{X}\eqdef\{\simquotient{x}\mid x{\in}{X}\}$. The incidence graph $\cG_S=(F{\cup}X,\Edges)$ of $S$ was introduced in Section~\ref{jhtgfkjhseghfaerv}. We will also consider the \emph{$\sim$\,--\,quotient of $\cG_S$}, which is a bipartite graph $\simquotient{\cG}_S=(F{\cup}\simquotient{X},\simquotient{\Edges})$ over the pair $(F,\simquotient{X})$, defined as follows:
\beq
(f,\simquotient{x})\in\simquotient{\Edges} &\mbox{iff}& \exists x\in\simquotient{x}:(\edge)\in\Edges\,.
\label{htrdehtrgsdf}
\eeq
The whole apparatus of bipartite graphs applies, which includes in particular Section~\ref{jhtgfkjhseghfaerv} regarding matchings. Consider both
\begin{itemize}
	\item the $\sim$\,--\,closure $\simclosure{S}$ of $S$ and its incidence graph $\cG_{\simclosure{S}}$, and
	\item the quotient graph $\simquotient{\cG}_S$ defined in (\ref{htrdehtrgsdf}).
\end{itemize}
The following result holds:
\smallskip
\begin{lemma}
	\label{jdshcfgvskdjmhgn} For $S,\simclosure{S}$, and $\simquotient{\cG}_S$ given as above, there exists a total injective function $\Gamma$, mapping any variable-complete (respectively equation-complete) matching $\simquotient{\cM}$ for $\simquotient{\cG}_S$, to a variable-complete (respectively equation-complete) matching $\simclosure{\cM}=\Gamma(\simquotient{\cM})$ for $\simclosure{S}$.
\end{lemma}
\smallskip
\begin{proof}	
By definition of $\simquotient{\cG}_S$, for every $(f,\simquotient{x})\in\simquotient{\cM}$, we can select a variable $x\in\simquotient{x}$ such that $(f,x)\in\cG_S$, we denote by $\chi(\simquotient{x})$ this selection. Perform the following finite recursion:
\begin{enumerate}
	\item Initialization: $\cM_0=\{(f,x)\mid x=\chi(\simquotient{x})\}$, $X_0=X$;
	\item While $X_{n}{\supset} X_{n-1}$, update $(\cM_{n+1},X_{n+1}){\gets}(\cM_n,X_n)$ by
	\begin{itemize}
		\item adding the pair	$(\euler{x}{z},z)$ to $\cM_n$, for every $x,z{\in}X_n$ such that $x{\in}\cM_n,z{\sim}x,z{\not\in}\cM_n$ (we use Notations~\ref{htgfkjthndfj});
		\item adding to $X_n$ every variable $u$ occurring in some above added $\euler{x}{z}$ and such that $u{\not\in}\cM_n,u{\neq}z$.
	\end{itemize}
\end{enumerate}
Let $(\cM_N,X_N)$ be the fixpoint. We set $\simclosure{\cM}=\Gamma(\simquotient{\cM})\eqdef\cM_N$ and we consider the map $\Gamma:\simquotient{\cM}\mapsto\simclosure{\cM}$. If $\simquotient{\cM}$ is variable-complete, then so is $\simclosure{\cM}$ by construction. Assume next that $\simquotient{\cM}$ is equation-complete. Then, the above recursion only adds equations that occur in the increasing sequence of matchings. Thus, $\simclosure{\cM}$ is equation-complete as well.
\end{proof}

\noindent Lemma~\ref{jdshcfgvskdjmhgn} allows us to reduce the search for a matching over the $\sim$\,--\,closure of a system, to the (simpler) search for a matching over its $\sim$\,--\,quotient.
Let us apply this to the array $(A_K,\Vars_{A_K})$, where $A_K$ was defined in (\ref{skudjcghskdjhcg}). In particular,
\beq\bea{rl}
\simquotient{\Psi}(K,\cM) \eqdef &
\!
\left\{
\bigl(\pprime{m}{f},\pprime{d}{{\simquotient{x}}}\bigr)
\left|\bea{l}
(\edge)\in{\cM} \\
0\leq m < c_{\!f}  \\
d=d_x{-}c_{\!f}{+}m
\eea\right.\negesp\right\}
\bigcup 
\left\{
\bigl(\pprimeppostset{c_{\!f}}{k}{f},\pprimeppostset{d_x}{k}{{\simquotient{x}}}\bigr)
\left|\bea{l}
(\edge){\in}{\cM} \\
0{\leq} k {\leq} K_{\!f}
\eea\right.\negesp\right\}
\\ [7mm]
\eqdef& \,\consistency{\simquotient{\cM}}~\bigcup~\indexreduced{(\simquotient{\cM}_K)}
\eea 
\label{ukjhygklughjmmk}
\eeq
defines a matching $\simquotient{\Psi}(K,\cM)\eqdef\consistency{\simquotient{\cM}}\cup\indexreduced{(\simquotient{\cM}_K)}$, where 
\begin{itemize}
	\item $\consistency{\simquotient{\cM}}$ identifies with $\consistency{{\cM}}$, the equation-complete matching (\ref{kujyfglukhoih}) associated to consistency conditions in the \sigmamethod; 
	\item $\indexreduced{(\simquotient{\cM}_K)}$ is a variable-complete matching for the quotient incidence graph of the pair $\bigl(\indexreduced{(A_K)},\indexreduced{(\simquotient{\Vars}_{A_K})}\bigr)$, where
	\beqq
	\indexreduced{(A_K)}&\negesp\eqdef\negesp&\left\{\pprimeppostset{c_{\!f}}{k}{f}
	\mid f{\in}{F}\,,\,0{\leq} k{\leq} K_{\!f}
	\right\}, \mbox{ and}
	\\
	\indexreduced{(\simquotient{\Vars}_{A_K})}&\negesp\eqdef\negesp&\left\{\pprimeppostset{d_x}{k}{{\simquotient{x}}}
	\mid x{\in}{X}\,,\,(\edge)\in\cM\,,\,0{\leq} k{\leq} K_{\!f}
	\right\}
		\eeqq
collect all the leading equations and the equivalence classes of all the leading variables of the array.
\end{itemize}
By Lemma~\ref{jdshcfgvskdjmhgn} 
\beq
\cM_\vsmall&\eqdef&
\Gamma\bigl(\simquotient{\Psi}(K,\cM)\bigr)
\label{htrgfsdjkuyfgnm}
\eeq
yields a matching  for the $\sim$\,--\,closure $A_\vsmall$ of array $A_K$, and the images by $\Gamma$ of $\consistency{\simquotient{\cM}}$ and $\indexreduced{(\simquotient{\cM}_K)}$ are equation-complete and variable-complete, respectively.

We now consider requirement (\ref{srrftjhtr}) stating, in particular, that all the equations of the tail block should be enabled by $\cM_\vsmall$.
\smallskip
\begin{lemma}
	\label{slfjhsglsujfv} To ensure that the above matching $\cM_\vsmall$ for $(A_\vsmall,\Vars_{A_\vsmall})$ satisfies requirement (\ref{srrftjhtr}), it is enough to select a variable height $K$ such that
	\beq
	\forall f\in{F}&\Ra& K_{\!f}\geq c_{\!f}\,. \label{sldfvikjshbflkjh}
	\eeq
\end{lemma}
\smallskip
\begin{proof}
The reader is referred to Problem~(\ref{liftuerhpituhepu8}) (index reduction) applied to the new mode $(F,X)$. The lemma is vacuously true if $c_{\!f}=0$ for every $f$, since no consistency equation exists in this case and $K=0$ is fine. So, we assume that $c_{\!f}>0$ holds for some $f$. Let $f\in{F}$ be such that $c_{\!f}>0$, and consider 
\beq
g{=}\pprimeppostset{m}{K_{\!f}}{f}{\in}\ttail{A_\vsmall},
&\mbox{with} & 0{\leq} m{<}c_{\!f}\,.
\label{ksfdjhsdgvkfjsdhkgjh}
\eeq
Since  consistency equations $\pprimeppostset{m}{k}{f},0{\leq}k{\leq}K_{\!f}$ are all unmatched in $\indexreduced{(\simquotient{\cM}_K)}$, so is equation $g$. By (\ref{sldfvikjshbflkjh}), $m{<}c_{\!f}{\leq} K_{\!f}$ holds. Consider
\beq
h\eqdef\pprimeppostset{c_{\!f}}{k}{f}, &\mbox{where}&
k\eqdef K_{\!f}{-}c_{\!f}{+}m \,.
\label{kdfjvhblskfjvbhl}
\eeq
Since ${c_{\!f}}{+}{k}={c_{\!f}}{+}K_{\!f}{-}c_{\!f}{+}m=K_{\!f}{+}m$, we have $h{\sim} g$. Equation $h$ is matched in $\indexreduced{(\simquotient{\cM}_K)}$ since it is a shifted leading equation; it is actually matched with $\pprimeppostset{d_x}{k}{x}$. Setting $d{=} d_x{-}c_{\!f}{+}m$, we have $d{+}K_{\!f}=d_x{-}c_{\!f}{+}m{+}K_{\!f}=d_x{+}k\,,$ hence $\pprimeppostset{d_x}{k}{x}\sim\pprimeppostset{d}{K_{\!f}}{x}$. By Lemma~\ref{ujytsagrehd}, variables $\pprimeppostset{d_x}{k}{x}$ and $\pprimeppostset{d}{K_{\!f}}{x}$ are related by Euler identity $0{=}\euler{u}{v}$, where $u{=}\pprimeppostset{d}{K_{\!f}}{x}$ and $v{=}\pprimeppostset{d_x}{k}{x}$, and this Euler identity belongs to $A_\vsmall$ by Lemma~\ref{hngtedgrefsdg} since $A_\vsmall$ is the $\sim$\,--\,closure of $A_K$.

\emph{The introduction of map $\chi$ is the key step of the proof:} Consider the subarray $A_g$ of $A_\vsmall$:
\beq
A_g\eqdef\left\{\negesp\bea{l}
0=h(\pmatch{v},V) \\ 0=\euler{u}{v}(\pmatch{u},v,U)
\eea\right. 
\label{hfrshtrgfsdjm}
\eeq
where $u=\pprimeppostset{d}{K_{\!f}}{x},v=\pprimeppostset{d_x}{k}{x}$, $U$ collects variables $\prec{u}$, and the matchings are indicated in \pmatch{blue}. We associate to $A_g$ the array $\chi(A_g)$ defined by
\beq
\chi(A_g)\eqdef\left\{\negesp\bea{l}
0=g(\pmatch{u},W) \\ 0=\euler{u}{v}({u},\pmatch{v},U)
\eea\right. 
\label{jyhtdhtrgtjmgfkjhn}
\eeq
Since this change is only local, the remaining part of the matching is not modified. Performing the mapping $A_g{\mapsto}\chi(A_g)$ for every consistency equation $g$ belonging to the tail of $A_K$, yields a matching satisfying (\ref{srrftjhtr}). The lemma is proved.
\end{proof}

\paragraph*{Taking past variables into account:}
Lemma~\ref{slfjhsglsujfv} does not take past variables $\Vars^-$ into account. The effect of $\Vars^-$ is twofold. First, some equations become facts, and thus get removed from the array. Second, the past variables are no longer dependent variables, and thus cannot participate in matchings. We investigate these two effects in detail. 

\paragraph*{\normalsize  Taking facts into account:} If $K$ is the variable height of the considered mode change array, consistency equations $\pprimeppostset{m}{K_{\!f}}{f}$ that are facts are removed from the array. To account for this, we modify Lemma~\ref{slfjhsglsujfv} by 
replacing Condition (\ref{sldfvikjshbflkjh}) by
\beq 
\forall f{\in}{F}  \Ra & K_{\!f}\geq \max\left\{c_{\!f}{-}m \,\left|\bea{c} 
0\leq m\leq c_{\!f}
\\ [1mm]
\pprimeppostset{m}{K_{\!f}}{f}\not\in\facts{A_\vsmall}
\eea\negesp\right.\right\}  \label{leejhghleigheiu}
\eeq
No other change is needed in the proof. 

Lemma~\ref{slfjhsglsujfv} with its extension (\ref{leejhghleigheiu}) is illustrated in Fig.\,\ref{kefygfkjfgjhgf}, displaying array $A_1$ for the cup-and-ball example. The exchange map $\chi$ is exemplified by the substitution of $(\ddot{k_1})$ for $(\postset{\dot{k_1}})$ when disabling equations.

\paragraph*{\normalsize Removing past variables from dependent variables:}
Consider $g$ and $h$ defined in (\ref{ksfdjhsdgvkfjsdhkgjh}) and (\ref{kdfjvhblskfjvbhl}). We know that $m{+}K_{\!f}=c_{\!f}{+}k$, and $k=m{+}K_{\!f}{-}c_{\!f}\geq{m}$ holds by (\ref{leejhghleigheiu}). 
Equation $h$ belongs to the $k$-th instant $(\ppostset{k}{\indexreduced{F}},\ppostset{k}{\indexreduced{X}})$ in array $A_K$. Therefore, in any possible choice for $\cM_\vsmall$, equation $g$ is paired with some variable belonging to $\ppostset{k}{\indexreduced{X}}$. This leads to considering the set $\indexreduced{X}^g$ collecting all the variables occurring in $g$ that could possibly be paired with $g$ in some matching $\cM_\vsmall$:
\beq
\indexreduced{X}^g \eqdef \left\{
\pprimeppostset{d}{K_{\!f}}{x}\mbox{ occurs in }g \,\left|\;\bea{c}
\pprimeppostset{d_x}{k}{x}  \in \ppostset{k}{\indexreduced{X}} 
\\
d_x{+}k = d{+}K_{\!f}
\eea\right.\negesp\right\}
\eeq
Let $K$ be a variable height for the array satisfying condition (\ref{leejhghleigheiu}). For $f{\in}{F}$ such that $c_{\!f}{>}0$, and $0{\leq} m{<}c_{\!f}$, we consider the following properties, where $g{=}\pprimeppostset{m}{K_{\!f}}{f}$:
	\beqq
\cP_*(K,f,m) &\negesp\eqdef\negesp&
\mbox{either $\pprimeppostset{m}{K_{\!f}}{f}$ is a fact, or $\indexreduced{X}^g\,{\setminus}\,\Vars^-{\neq}\emptyset$}
\\ 
\cP^*(K,f,m) &\negesp\eqdef\negesp&
\mbox{either $\pprimeppostset{m}{K_{\!f}}{f}$ is a fact, or $\indexreduced{X}^g{\cap}\Vars^-{=}\emptyset$}
\eeqq
so that $\cP^*(K,f,m)\Ra\cP_*(K,f,m)$, and we set
\beqq
\cP_*(K)
&\eqdef&
\bigwedge\left\{\cP_*(K,f,m) \mid {f{\in}{F},0{\leq} m{<}c_{\!f}}\right\}
\\ 
\cP^*(K)
&\eqdef&
\bigwedge\left\{\cP^*(K,f,m) \mid {f{\in}{F},0{\leq}m{<}c_{\!f}}\right\} 
\eeqq
and, finally:
\beqq
K_*&\eqdef&\min\bigl\{K \mid K\models(\ref{leejhghleigheiu}) \mbox{ and } K\models\cP_*(K) \bigr\}
\\
K^*&\eqdef&\min\bigl\{K \mid K\models(\ref{leejhghleigheiu}) \mbox{ and } K\models\cP^*(K) \bigr\}
\eeqq
so that $K_*\leq K^*$.
\smallskip
\begin{lemma}
	\label{lergurlouguih} 
With these notations, to find a variable-complete matching $\cM_\vsmall$ satisfying requirement $(\ref{srrftjhtr})$: 
\begin{enumerate}
	\item \label{hgrfdkjyhfgkljh} it is enough to select $K\geq K^*$, whereas
	\item \label{jhgrdjhgfjhgv} no such matching can be expected unless $K\geq K_*$.
\end{enumerate}
\end{lemma}
\begin{proof} Let $K$ satisfy (\ref{leejhghleigheiu}), and assume the existence of a variable-complete matching $\cM_\vsmall$ satisfying requirement (\ref{srrftjhtr}).
Assume that $\pprimeppostset{m}{K_{\!f}}{f}$ is not a fact, and there exists a variable $\pprimeppostset{d}{K_{\!f}}{x}$ that is paired with $\pprimeppostset{m}{K_{\!f}}{f}$ in the matching  $\simquotient{\Psi}(K,\cM)$ defined in (\ref{ukjhygklughjmmk}). The set $\Vars^-$ of past variables forbids such a pairing if and only if it contains $\pprimeppostset{d}{K_{\!f}}{x}$. We deduce that
\beq
\mbox{
\begin{minipage}{10cm}
	 there exists a matching $\indexreduced{\cM}$ following (\ref{dskfjhsfbhjn}), such that $(\pprimeppostset{m}{K_{\!f}}{f},\pprimeppostset{d}{K_{\!f}}{x})\in\simquotient{\Psi}(K,\cM)$ and $\pprimeppostset{d}{K_{\!f}}{x}\not\in\Vars^-$.
\end{minipage}
}
\label{kserghvelrghvelqrh}
\eeq
This leads to considering the property
\beqq
\cP(K,f,m) &\eqdef& \mbox{either $\pprimeppostset{m}{K_{\!f}}{f}$ is a fact, or (\ref{kserghvelrghvelqrh}) holds.}
\eeqq
Property $\cP(f,m)$ explores all the alternatives in chosing a candidate variable $\pprimeppostset{d}{K_{\!f}}{x}$, for pairing with $\pprimeppostset{m}{K_{\!f}}{f}$ in some matching $\simquotient{\Psi}(K,\cM)$ considered in (\ref{ukjhygklughjmmk}). 

Condition (\ref{kserghvelrghvelqrh}) is too complicated, as it involves exploring the set of all the matchings $\cM$ solution of Problem~(\ref{liftuerhpituhepu8}). We thus consider the following properties:
\beqq
\cP_*(K,f,m) &\negesp\eqdef\negesp&
\mbox{either $\pprimeppostset{m}{K_{\!f}}{f}$ is a fact, or $\indexreduced{X}^g\setminus\Vars^-\neq\emptyset$.}
\\
\cP^*(K,f,m) &\negesp\eqdef\negesp&
\mbox{either $\pprimeppostset{m}{K_{\!f}}{f}$ is a fact, or $\indexreduced{X}^g{\cap}\Vars^-{=}\emptyset$.}
\\
\mbox{where}&\negesp\negesp&g=\pprimeppostset{m}{K_{\!f}}{f}\,.
\eeqq
Then, the following implications hold:
\[
\cP^*(K,f,m) \Ra \cP(K,f,m) \Ra \cP_*(K,f,m)\,,
\]
which proves statements~\ref{hgrfdkjyhfgkljh} and~\ref{jhgrdjhgfjhgv} of  the lemma.
\end{proof}

\noindent The following result states that there is no point going beyond $K^*$ in searching for a good solution of Problem~\ref{likgusefhdrlpouhiouh}:
\begin{lemma}
	\label{sekrghelrghlkjh} Let $K$ be the variable height of a good solution of Problem~$\ref{likgusefhdrlpouhiouh}$. Then, $K\leq K^*$ holds.
\end{lemma}
\begin{proof}
	Let $(f,m)$ be a pair such that: 1) $m=c_{\!f}{-}1$, and 2) $\pprimeppostset{m}{K_{\!f}}{f}$ is not a fact and is paired with $\pprimeppostset{(d_x-1)}{K_{\!f}}{x}$ in the matching $\cM_\vsmall$. Then, by (\ref{jyhtdhtrgtjmgfkjhn}), the Euler identity
	\beqq
	\euler{y}{z}
	&:&
	\vsmall{\times}\underbrace{\pmatch{\pprimeppostset{d_x}{(K_{\!f}-1)}{x}}}_{y}=\underbrace{\pprimeppostset{(d_x-1)}{K_{\!f}}{x}}_{z}-\underbrace{\pprimeppostset{(d_x-1)}{(K_{\!f}-1)}{x}}_{u}
	\eeqq
	needs to be added to the array. Since $\rescaling{z}=0$ by goodness condition (\ref{jsdyhgcdfasjyht}), this Euler identity induces the following rescaling equation:
	\beq
	\rescaling{y}=\rescaling{\euler{y}{z}}\geq 1+\rescaling{u}\,. \label{klgjhrlkghlkujh}
	\eeq
	\emph{The sequel of the proof is by contradiction:} 
	We will prove:
	\beq
	\mbox{
	\begin{minipage}{14cm}
		 if $K{>} K^*$,  we can find a pair $(f,m)$ as above, such that $\rescaling{u}{>}0$\,.
	\end{minipage}
	} \label{lkgjetrhgjhkjh}
	\eeq
	Property (\ref{lkgjetrhgjhkjh}) would imply $\rescaling{y}>1$ by (\ref{klgjhrlkghlkujh}). In this case, goodness condition (\ref{klifgyegrujytfgeikuy}) would get violated for $y=\pprimeppostset{d_x}{(K_{\!f}-1)}{x}$, since $K_{\!f}{-}1+\min(d_x,\rescaling{y})>K_{\!f}$. 
	
	\emph{It thus remains to prove \emph{(\ref{lkgjetrhgjhkjh})}:}
	Since $K>K^*$, we can find a pair $(f,m)$ as above, such that $\indexreduced{X}^h\cap\Vars^-=\emptyset$, where 
	$h=\pprimeppostset{m}{(K_{\!f}-1)}{f}$, and $u\in\indexreduced{X}^h$. Since $u$ is not a past variable, it must be paired in the matching $\cM_\vsmall$, namely with the Euler identity
	\beqq
	\euler{u}{v}
	&:& \vsmall{\times}\underbrace{\pmatch{\pprimeppostset{(d_x-1)}{(K_{\!f}-1)}{x}}}_{\pmatch{u}}=\underbrace{\pprimeppostset{(d_x-2)}{K_{\!f}}{x}}_{v}-\underbrace{\pprimeppostset{(d_x-2)}{(K_{\!f}-1)}{x}}_{w}
	\eeqq
	which induces the following rescaling equation:
	\beqq
	\rescaling{u}=\rescaling{\euler{u}{v}}\geq 1+\rescaling{w}&>&0\,, 
	\eeqq
	proving (\ref{lkgjetrhgjhkjh}). The proof is now complete.
\end{proof}

\noindent Lemmas~\ref{lergurlouguih} and~\ref{sekrghelrghlkjh} together prove Theorem~\ref{sekdfuysageloiug} and make the heights $K_*$ and $K^*$ precise.

Theorem~\ref{sekdfuysageloiug} was illustrated for the cup-and-ball example by the results $K_*=K^*=1$. We illustrate in Appendix~\ref{hgnhtdfiytujhgn} what happens if we still insist taking $K>1$ for this example.

\section{Examples having no good solution}
\label{ksdfiygskdyjgf}
\subsection{The cup-and-ball with $K{>}1$}
\label{hgnhtdfiytujhgn}
In Section~\ref{jhtgfkuytjrsdfghk}, we successfully studied the cup-and-ball example with a minimal mode change array $A$ of height $K{=}1$. In this section we study what happens if we select a non-minimal array. Note that we know by Statement~\ref{hjtgfkujyhglk} of Theorem~\ref{skldjhfgslkdufghikl} that no solution should be found. Nevertheless, we like to investigate how the failure to find a solution occurs. Hence, we consider the choice $K{=}2$. We will show that Problem~\ref{likgusefhdrlpouhiouh} possesses no solution with these choices. The same holds for any $K{>}1$. Thus, $K1$ is the only choice for the height of $A$.

Here follow the details. The mode change array with $K{=}2$ is the following:
 \beqq\small
A_2=\left\{\bea{rclcll}
 0&\negesp=\negesp& \pmatch{\ddot{x}}+{\tension}x 
& (f_1) 
& \nstime
\\
 0&\negesp=\negesp& \ddot{y}+\pmatch{\tension}y+g 
& (f_2)  & \cdots
\\
\gremph{0} &\!\!\!\!\gremph{=}\!\!\!\!& \gremph{{L^2}{-}(x^2{+}y^2)} 
& \gremph{(\straight_1)}  & \cdots
\\
\remph{0} &\!\!\!\!\remph{=}\!\!\!\!& \remph{x\dot{x}{+}y\dot{y}} 
& \remph{(\dot{\straight_1})}  &  \cdots
\\
\remph{0} &\!\!\!\!\remph{=}\!\!\!\!&  \remph{x\ddot{x}{+}\dot{x}^2{+}\dot{y}^2{+}y\ddot{y}}
& \remph{(\ddot{\straight_1})}  &  \cdots
\\ [2mm]
0&\negesp=\negesp& \postset{(\pmatch{\ddot{x}}+{\tension}x)}
& (f_3) 
& \nstime{+}\vsmall
\\
 0&\negesp=\negesp& \postset{(\ddot{y}+\pmatch{\tension}y+g)}
& (f_4)  & \cdots
\\
\gremph{0} &\!\!\!\!\gremph{=}\!\!\!\!& \gremph{\postset{({L^2}{-}(x^2{+}y^2))} }
& \gremph{(\postset{\straight_1})}    & \cdots
\\
\remph{0} &\!\!\!\!\remph{=}\!\!\!\!& \remph{\postset{{(x\dot{x}{+}y{\dot{y}})}}}
& \remph{(\postset{\dot{\straight_1}})}  & \cdots
\\
\remph{0} &\!\!\!\!\remph{=}\!\!\!\!&  \remph{\postset{({x\ddot{x}{+}\dot{x}^2{+}\dot{y}^2{+}y\ddot{y}})}}
& \remph{(\postset{\ddot{\straight_1}})} & \cdots
\\ [2mm]
0&\negesp=\negesp& \ppostset{2}{(\pmatch{\ddot{x}}+{\tension}x)}
& (f_5) 
& \nstime{+}2\vsmall
\\
 0&\negesp=\negesp& \ppostset{2}{(\ddot{y}+\pmatch{\tension}y+g)}
& (f_6)  & \cdots
\\
{0} &\!\!\!\!{=}\!\!\!\!& {\ppostset{2}{({L^2}{-}(x^2{+}\pmatch{y^2}))} }
& {(f_7)}    & \cdots
\\
{0} &\!\!\!\!{=}\!\!\!\!& {\ppostset{2}{{(x\dot{x}{+}y\pmatch{\dot{y}})}}}
& {(f_8)}  & \cdots
\\
{0} &\!\!\!\!{=}\!\!\!\!&  \ppostset{2}{({x\ddot{x}{+}\dot{x}^2{+}\dot{y}^2{+}y\pmatch{\ddot{y}}})}
& {(f_9)}   & \cdots
\\ [2mm]
{0} &\!\!\!\!{=}\!\!\!\!& \ddot{x}-\vsmall^{-2}(\pmatch{\ppostset{2}{{x}}}-2\postset{x}+x) &(f_{10})& \mbox{Euler ids.}
\\
{0} &\!\!\!\!{=}\!\!\!\!& \pmatch{\ddot{y}}-\vsmall^{-2}({\ppostset{2}{{y}}}-2\postset{y}+y) &(f_{11})
\\
{0} &\!\!\!\!{=}\!\!\!\!& \postset{\ddot{x}}-\vsmall^{-1}(\pmatch{\ppostset{2}{\dot{x}}}-\postset{\dot{x}}) &(f_{12})
\\
{0} &\!\!\!\!{=}\!\!\!\!& \pmatch{\postset{\ddot{y}}}-\vsmall^{-1}(\ppostset{2}{\dot{y}}-\postset{\dot{y}}) &(f_{13})
\\
{0} &\!\!\!\!{=}\!\!\!\!& \pmatch{\postset{\dot{x}}}-\vsmall^{-1}(\ppostset{2}{x}-\postset{x}) &(f_{14})
\\
{0} &\!\!\!\!{=}\!\!\!\!& \pmatch{\postset{\dot{y}}}-\vsmall^{-1}(\ppostset{2}{y}-\postset{y}) &(f_{15})
\eea\right.
\eeqq
$\Vars^-$ is the same as for $K{=}1$. \gremph{Facts} and \remph{conflicts} are pointed. The subsystem in black is structurally nonsingular, with a perfect matching  $\pmatch{\cM}$ highlighted in {\color{blue}blue}. In the right most column we indicate the origin of each equation: for example, $\nstime{+}\vsmall$ indicates that the corresponding equation originates from the $1$-shifted discretized dynamics.

The corresponding rescaling calculus is:
\[\small\bea{rclcl}
\rescaling{\ddot{x}} &\negesp=\negesp& \rescaling{\!{f_1}} &\negesp\geq\negesp& \rescaling{\tension}
\\
\rescaling{\tension} &\negesp=\negesp& \rescaling{\!{f_2}} &\negesp\geq\negesp& \rescaling{\ddot{y}}
\\
\rescaling{\postset{\ddot{x}}} &\negesp=\negesp& \rescaling{\!{f_3}} &\negesp\geq\negesp& \rescaling{\postset{\tension}}
\\
\rescaling{\postset{\tension}} &\negesp=\negesp& \rescaling{\!{f_4}} &\negesp\geq\negesp& \rescaling{\postset{\ddot{y}}}
\\
\rescaling{\ppostset{2}{\ddot{x}}} &\negesp=\negesp& \rescaling{\!{f_5}} &\negesp\geq\negesp& \rescaling{\ppostset{2}{\tension}}
\\
\rescaling{\ppostset{2}{\tension}} &\negesp=\negesp& \rescaling{\!{f_6}} &\negesp\geq\negesp& \rescaling{\ppostset{2}{\ddot{y}}}
\\
 && \rescaling{\!{f_7}} &\negesp=\negesp& \rescaling{\ppostset{2}{x}}=\rescaling{\ppostset{2}{y}}=0
\\
\rescaling{\ppostset{2}{y}}+\rescaling{\ppostset{2}{\dot{y}}} &\negesp=\negesp& \rescaling{\!{f_8}} &\negesp\geq\negesp& \rescaling{\ppostset{2}{x}}+\rescaling{\ppostset{2}{\dot{x}}}
\\
\rescaling{\ppostset{2}{y}}{+}\rescaling{\ppostset{2}{\ddot{y}}} &\negesp=\negesp& \rescaling{\!{f_9}} &\negesp\geq\negesp& \max(\rescaling{\ppostset{2}{x}}{+}\rescaling{\ppostset{2}{\ddot{x}}},2\rescaling{\ppostset{2}{\dot{x}}},2\rescaling{\ppostset{2}{\dot{y}}})
\\
2+\rescaling{\ppostset{2}{x}} &\negesp=\negesp& \rescaling{\!{f_{10}}} &\negesp\geq\negesp& \rescaling{\ddot{x}}
\\
\rescaling{\ddot{y}} &\negesp=\negesp& \rescaling{\!{f_{11}}} &\negesp\geq\negesp& 2+\rescaling{\ppostset{2}{y}}
\\
1+\rescaling{\postset{\dot{x}}} &\negesp=\negesp& \rescaling{\!{f_{12}}} &\negesp\geq\negesp& \rescaling{\postset{\ddot{x}}}
\\
\rescaling{\postset{\ddot{y}}} &\negesp=\negesp& \rescaling{\!{f_{13}}} &\negesp\geq\negesp& 1+\rescaling{\postset{\dot{y}}}
\\
\rescaling{\postset{\dot{x}}} &\negesp=\negesp& \rescaling{\!{f_{14}}} &\negesp\geq\negesp& 1+\rescaling{\ppostset{2}{{x}}}
\\
\rescaling{\postset{\dot{y}}} &\negesp=\negesp& \rescaling{\!{f_{15}}} &\negesp\geq\negesp& 1+\rescaling{\ppostset{2}{{y}}}
\eea
\]
Note that the rescaling equation for $f_7$ corresponds to rescaling equation (\ref{ksdjfhsgdkfjgh}) for nonlinear functions. By goodness condition (\ref{jkdsfhgfjsmdhfgjk}), the following equations and variables possess a zero rescaling offset: $f_7,x,y,{\ppostset{2}{x}},{\ppostset{2}{y}},{\ppostset{2}{\dot{x}}},{\ppostset{2}{\dot{y}}},{\ppostset{2}{\ddot{x}}},{\ppostset{2}{\ddot{y}}},{\ppostset{2}{\tension}}$. 
	The solution of the rescaling calculus is the following, where we list the variables and equations with rescaling offsets $2,1,0$:
	\[\small\bea{lll}
	2&\negesp:\negesp&
	f_1,f_2,f_3,f_4,f_{10},f_{11},f_{12},f_{13},{\ddot{x}},{\ddot{y}},
	{\tension},\remph{{\postset{\ddot{x}}},{\postset{\ddot{y}}}},{\postset{\tension}} 
	\\
1&\negesp:\negesp&f_{14},f_{15},{\postset{\dot{x}}},{\postset{\dot{y}}}
	\\
0&\negesp:\negesp&\mbox{other equations and variables}
	\eea\]
	The rescaling offsets violating goodness condition (\ref{klifgyegrujytfgeikuy}) were highlighted in {\color{red}red}. Hence, Problem~\ref{likgusefhdrlpouhiouh} possesses no solution for $K{\,=\,}2$.
	
	Still, suppose that we insist producing the restart system following Procedure~\ref{skedfjuyghwsaliukgh}. The useful part of the rescaled system $\rescaled{A_2}$ is the following (rescaled Euler identities have been directly used to expand impulsive derivatives, and are no longer shown):
 \beqq
\rescaled{A_2}=\left\{\bea{rclcll}
 0&\negesp=\negesp& \pmatch{\ppostset{2}{x}}-2\postset{x}+x+\rescaled{\tension}x 
\\
 0&\negesp=\negesp& {\ppostset{2}{y}}-2\postset{y}+y+\pmatch{\rescaled{\tension}}y 
\\ [2mm]
0&\negesp=\negesp& \pmatch{\ppostset{3}{x}}-2\ppostset{2}{x}+\postset{x}+\rescaled{{\postset{\tension}}}\postset{x}
\\
 0&\negesp=\negesp& {\ppostset{3}{y}}-2\ppostset{2}{y}+\postset{y}+\pmatch{\rescaled{{\postset{\tension}}}}\postset{y}
\\ [2mm]
{0} &\!\!\!\!{=}\!\!\!\!& {\ppostset{2}{({L^2}{-}(x^2{+}\pmatch{y^2}))} }
\\
{0} &\!\!\!\!{=}\!\!\!\!& {\ppostset{2}{{(x\dot{x}{+}y\pmatch{\dot{y}})}}}
\eea\right. 
\eeqq
The restart system follows by performing renaming (\ref{erlfgiuehrliueh}). The array involves variables beyond the tail instant $\nstime+2\vsmall$, namely $\ppostset{3}{x},\ppostset{3}{y}$. Hence, the renaming (\ref{erlfgiuehrliueh}) is no longer bijective, which kills the structural nonsingularity of the restart system:
\beqq
\mbox{restart system}=\left\{\bea{rclcll}
 0&\negesp=\negesp& \pmatch{\pplus{x}}-\mmoins{x}+\rescaled{\tension}\mmoins{x} 
\\
 0&\negesp=\negesp& {\pplus{y}}-\mmoins{y}+\pmatch{\rescaled{\tension}}\mmoins{y} 
\\ [2mm]
0&\negesp=\negesp& {{\pplus{{x}}}+\mmoins{{x}}+\rescaled{{\postset{\tension}}}\mmoins{x}}
\\
 0&\negesp=\negesp& {\pplus{{y}}}+\mmoins{{y}}+\pmatch{\rescaled{{\postset{\tension}}}}\mmoins{y}
\\ [2mm]
{0} &\!\!\!\!{=}\!\!\!\!&{L^2}{-}(({\pplus{x}})^2{+}\pmatch{(\pplus{y})^2})
\\
{0} &\!\!\!\!{=}\!\!\!\!& \pplus{x}\pplus{{\dot{x}}}{+}\pplus{y}\pmatch{\pplus{{\dot{y}}}}
\eea\right. 
\eeqq
The resulting restart system possesses $6$ equations and $6$ dependent variables: $\pplus{x},\pplus{{\dot{x}}},\pplus{y},\pplus{{\dot{y}}},\rescaled{\tension},\rescaled{{\postset{\tension}}}$. Nevertheless it is structurally singular: $\pplus{{\dot{x}}}$ and the third equation are both unmatched. However, adding the equation $\rescaled{\tension}=\rescaled{{\postset{\tension}}}$ has the following effects:
\begin{itemize}
	\item The first and third equations become identical; so we can discard the third equation: the resulting system is no longer structurally conflicting;
	\item Still, variable $\pplus{{\dot{x}}}$ remains unmatched.
\end{itemize}
The conclusion is that the additional non-structural post-processing we applied did not provide any further progress toward finding a solution to the hot restart.

\smallskip
\subsection{The cup-and-ball with exogenous mode change}
\label{klsdifgujhsjdhfh}
In this section we modify the cup-and-ball example, by making the control of mode changes external (with reference to the original model (\ref{loeifuhpwoui}), we removed equation $(\straight_0)$):
\beqq\left\{\bea{rll}
& 0= \ddot{x}+{\tension}x & (\eqq_1) \\
& 0= \ddot{y}+{\tension}y+g  & (\eqq_2) \\
\when \; \guard\; \doo& 0={L^2}{-}(x^2{+}y^2)   & (\straight_1) \\
\prog{and}& 0=\tension+s   & (\straight_2) \\
\when \;\prog{not}\; \guard\; \doo& 0=\tension   & (\straight_3) \\
\prog{and}& 0=({L^2}{-}(x^2{+}y^2))-s   & (\straight_4) \\
\eea\right.
\eeqq
%
We focus again on the mode change $\guard:\fff\ra\ttt$ and we show the mode change array with $K=2$:
 \beq\small
\left\{\!\!\!\bea{rclcll}
 0&\negesp=\negesp& \pmatch{\ddot{x}}+{\tension}x 
& (\eqq_1) & \guard{=}\ttt,\nstime
\\
 0&\negesp=\negesp& \ddot{y}+\pmatch{\tension}y+g 
& (\eqq_2)  & \dots\dots
\\
\remph{0} &\!\!\!\!\remph{=}\!\!\!\!& \remph{{L^2}{-}(x^2{+}y^2)} 
& \remph{(\straight_1)}  & \dots\dots 
\\
\remph{0} &\!\!\!\!\remph{=}\!\!\!\!& \remph{x\dot{x}{+}y\dot{y}} 
& \remph{(\dot{\straight_1})}  & \dots\dots 
\\
\remph{0} &\!\!\!\!\remph{=}\!\!\!\!&  \remph{x\ddot{x}{+}{\dot{x}}^2{+}{\dot{y}}^2{+}y\ddot{y}}
& \remph{(\ddot{\straight_1})}  & \dots\dots 
\\ [2mm]
 0&\negesp=\negesp& \postset{(\pmatch{\ddot{x}}+{\tension}x)}
& (\postset{\eqq_1}) & \guard{=}\ttt,\postset{\nstime}
\\
 0&\negesp=\negesp& \postset{(\ddot{y}+\pmatch{\tension}y)}+g 
& (\postset{\eqq_2})  & \dots\dots
\\
\remph{0} &\!\!\!\!\remph{=}\!\!\!\!& \remph{\postset{({L^2}{-}(x^2{+}y^2))} }
& \remph{(\postset{\straight_1})}  & \dots\dots 
\\
\remph{0} &\!\!\!\!\remph{=}\!\!\!\!& \remph{\postset{{(x\dot{x}{+}y\dot{y})}}}
& \remph{(\postset{\dot{\straight_1}})}  & \dots\dots 
\\
\remph{0} &\!\!\!\!\remph{=}\!\!\!\!& \remph{\postset{{(x\ddot{x}{+}{\dot{x}}^2{+}{\dot{y}}^2{+}y\ddot{y})}}}
& \remph{(\postset{{\ddot{\straight_1}}})}  & \dots\dots 
\\ [2mm]
{0} &\!\!\!\!{=}\!\!\!\!& {\ppostset{2}{({L^2}{-}(x^2{+}\pmatch{y^2}))} }
& {(\ppostset{2}{\straight_1})}  & \guard{=}\ttt,\ppostset{2}{\nstime} & 
\\
{0} &\!\!\!\!{=}\!\!\!\!& {\ppostset{2}{{(x\dot{x}{+}y\pmatch{\dot{y}})}}}
& {(\ppostset{2}{\dot{\straight_1}})}  & \dots\dots & 
\eea\right.
\label{ksdjhvgdfkisdyujhg}
\eeq
Observe the effect of the change in the model. No fact occurs. Instead, we have two more conflicting equations, namely $(k_1)$ and $(\postset{k_1})$. This causes the need for a longer array compared to the one in Fig.\,\ref{kefygfkjfgjhgf} for the original cup-and-ball. 
In (\ref{ksdjhvgdfkisdyujhg}), the \remph{conflicts} were pointed. Corresponding equations are disabled. A perfect matching modulo $\sim$ is highlighted in {\color{blue}blue}. We make it explicit by adding the needed Euler identities:
%
 \beq\small
\left\{\!\!\!\bea{rclr}
 0&\negesp=\negesp& \pmatch{\ddot{x}}+{\tension}x 
& (\eqq_1:1) 
\\ 
 0&\negesp=\negesp& \ddot{y}+\pmatch{\tension}y+g 
& (\eqq_2:2)  
\\ 
 0&\negesp=\negesp& \postset{(\pmatch{\ddot{x}}+{\tension}x)}
& (\postset{\eqq_1}:3) 
\\ 
 0&\negesp=\negesp& \postset{(\ddot{y}+\pmatch{\tension}y)}+g 
& (\postset{\eqq_2}:4) 
\\ 
0&\!\!\!\!{=}\!\!\!\!& {\ppostset{2}{({L^2}{-}(x^2{+}\pmatch{y^2}))} }
& {(\ppostset{2}{\straight_1}:5)}  
\\ 
0 &\!\!\!\!{=}\!\!\!\!& {\ppostset{2}{{(x\dot{x}{+}y\pmatch{\dot{y}})}}}
& {(\ppostset{2}{\dot{\straight_1}}:6)}  
\\ 
0 &\!\!\!\!{=}\!\!\!\!& \vsmall^2{\times}\ddot{x}-(\pmatch{\ppostset{2}{x}}-2\postset{x}+x) & (\mathcal{E}_{\ddot{x},\ppostset{2}{x}}:7)
\\ 
0 &\!\!\!\!{=}\!\!\!\!& \vsmall^2{\times}\pmatch{\ddot{y}}-(\ppostset{2}{y}-2\postset{y}+y) & (\mathcal{E}_{\ddot{y},\ppostset{2}{y}}:8)
\\ 
0 &\!\!\!\!{=}\!\!\!\!& \vsmall{\times}\pmatch{\postset{\dot{x}}}-({\ppostset{2}{{x}}}-\postset{{x}}) & (\mathcal{E}_{\postset{\dot{x}},\ppostset{2}{x}}:9)
\\ 
0 &\!\!\!\!{=}\!\!\!\!& \vsmall{\times}\pmatch{\postset{\dot{y}}}-({\ppostset{2}{{y}}}-\postset{{y}}) & (\mathcal{E}_{\postset{\dot{y}},\ppostset{2}{y}}:10)
\\ 
0 &\!\!\!\!{=}\!\!\!\!& \vsmall{\times}\postset{\ddot{x}}-(\pmatch{\ppostset{2}{\dot{x}}}-\postset{\dot{x}}) & (\mathcal{E}_{\postset{\ddot{x}},\ppostset{2}{\dot{x}}}:11)
\\ 
0 &\!\!\!\!{=}\!\!\!\!& \vsmall{\times}\pmatch{\postset{\ddot{y}}}-(\ppostset{2}{\dot{y}}-\postset{\dot{y}}) & (\mathcal{E}_{\postset{\ddot{y}},\ppostset{2}{\dot{y}}}:12)
\eea\right.
\label{yhgtrdhtgfjhfg}
\eeq
The resulting rescaling calculus is:
 \beqq\small
\bea{rclcl}
\rescaling{\ddot{x}} &\negesp=\negesp& \rescaling{{\!f_1}} &\negesp\geq\negesp& \rescaling{\tension}
\\ 
\rescaling{\tension} &\negesp=\negesp& \rescaling{{\!f_2}} &\negesp\geq\negesp& \rescaling{\ddot{y}}
\\ 
\rescaling{\postset{\ddot{x}}} &\negesp=\negesp& \rescaling{{\!f_3}} &\negesp\geq\negesp& \rescaling{\postset{\tension}}
\\ 
\rescaling{\postset{\tension}} &\negesp=\negesp& \rescaling{{\!f_4}} &\negesp\geq\negesp& \rescaling{\postset{\ddot{y}}}
\\ 
&& \rescaling{{\!f_5}} &\negesp=\negesp& \rescaling{\ppostset{2}{y}}=\rescaling{\ppostset{2}{x}}=0
\\ 
\rescaling{\ppostset{2}{\dot{y}}}+\rescaling{\ppostset{2}{y}} &\negesp=\negesp& \rescaling{{\!f_6}} &\negesp\geq\negesp& \rescaling{\ppostset{2}{\dot{x}}}+\rescaling{\ppostset{2}{x}}
\\ 
2+\rescaling{\ppostset{2}{x}} &\negesp=\negesp& \rescaling{{\!f_7}} &\negesp\geq\negesp& \rescaling{\ddot{x}}
\\ 
2+\rescaling{\ppostset{2}{y}} &\negesp=\negesp& \rescaling{{\!f_8}} &\negesp\geq\negesp& \rescaling{\ddot{y}}
\\ 
1+\rescaling{\postset{\dot{x}}} &\negesp=\negesp& \rescaling{\!{f_{9}}} &\negesp\geq\negesp& \rescaling{\postset{\ddot{x}}}
\\
\rescaling{\postset{\ddot{y}}} &\negesp=\negesp& \rescaling{\!{f_{10}}} &\negesp\geq\negesp& 1+\rescaling{\postset{\dot{y}}}
\\
\rescaling{\postset{\dot{x}}} &\negesp=\negesp& \rescaling{\!{f_{11}}} &\negesp\geq\negesp& 1+\rescaling{\ppostset{2}{{x}}}
\\
\rescaling{\postset{\dot{y}}} &\negesp=\negesp& \rescaling{\!{f_{12}}} &\negesp\geq\negesp& 1+\rescaling{\ppostset{2}{{y}}}
\eea
\eeqq
The rescaling equation for $f_5$ corresponds to rescaling equation (\ref{ksdjfhsgdkfjgh}) for nonlinear functions.
The solution of the rescaling calculus is the following, where we list the variables and equations with rescaling offsets $2,1,0$:
	\[\small\bea{rclcl}
2&\negesp:\negesp&	f_1,f_2,f_3,f_4,f_7,f_8,{f_9},{f_{10}},{\ddot{x}},{\ddot{y}},{\tension}, \remph{{\postset{\ddot{x}}},{\postset{\ddot{y}}}},{\postset{\tension}} 
	\\
	1&\negesp:\negesp&{f_{11}},f_{12},{\postset{\dot{x}}},{\postset{\dot{y}}}
	\\
	0&\negesp:\negesp&\mbox{other equations and variables}
	\eea\]
The rescaling offsets violating goodness condition (\ref{klifgyegrujytfgeikuy}) are highlighted in {\color{red}red}.
This solution violates goodness condition (\ref{klifgyegrujytfgeikuy}), hence Problem~\ref{likgusefhdrlpouhiouh} possesses no good solution. The considered mode change is insufficiently determined. 

Since goodness conditions (\ref{jsdyhgcdfasjyht},\ref{jkdsfhgfjsmdhfgjk}) are satisfied, we can be more precise, however, by applying Corollary~\ref{hgrfdjrgfdfkm}. The rescaling offsets violating goodness condition (\ref{klifgyegrujytfgeikuy}) are highlighted in {\color{red}red}. The involved variables are ${\postset{\ddot{x}}}$ and ${\postset{\ddot{y}}}$. Remove the equations that are matched with these variables, namely $\postset{e_1}$ and $\mathcal{E}_{\postset{\ddot{y}},\ppostset{2}{\dot{y}}}$. Then, a Dulmage-Mendelsohn decomposition of the system of (\ref{yhgtrdhtgfjhfg}), we denote it by ${A_\vsmall}$ yields the following result: 
\[\bea{l}
A_\vsmall^o = \emptyset
\\ [1mm]
A_\vsmall^r = \bigl((e_1,e_2,\ppostset{2}{k_1},\mathcal{E}_{\ddot{x},\ppostset{2}{x}},\mathcal{E}_{\ddot{y},\ppostset{2}{y}})
,(\ddot{x},\ppostset{2}{x},\ddot{y},\ppostset{2}{y},\tension)\bigr)
\\ [1mm]
A_\vsmall^u = \mbox{other equations and variables}
\eea\]
Rescaling of the regular subsystem $A_\vsmall^r$ yields 
\[\left\{\bea{l}
0={\ppostset{2}{x}}-2\postset{x}+(1+\rescaled{\tension})x \\
0={\ppostset{2}{y}}-2\postset{y}+(1+\rescaled{\tension})y \\
0=(\ppostset{2}{y})^2+(\ppostset{2}{y})^2
\eea\right.\]
which maps to the restart system
\beq\mbox{restart system}=\left\{\bea{l}
0={\pplus{x}}-\mmoins{x}+\rescaled{\tension}\mmoins{x} \\
0={\pplus{y}}-\mmoins{y}+\rescaled{\tension}\mmoins{y} \\
0=(\pplus{y})^2+(\pplus{y})^2
\eea\right.
\label{htrsejytdik}
\eeq
Conclusion: \emph{the positions are determined by $(\ref{htrsejytdik})$, whereas the velocities remain undetermined.}

\smallskip
\subsection{A ``strange'' cup-and-ball}
\label{sjdufghvcsedmgh}
With this example we aim at analyzing difficult cases related to facts. The following ``strange'' cup-and-ball is a different modification of the cup-and-ball example, by which the zero-crossing function defining the mode change was  modified, which modifies the facts. Its model is the following, where the modification is highlighted in {\color{red}red}:
\beqq\left\{\bea{rll}
& 0= \ddot{x}+{\tension}x & (\eqq_1) \\
& 0= \ddot{y}+{\tension}y+g  & (\eqq_2) \\
& \remph{\guard= [s'^-\leq{0}]}; \guard(0)=\fff  & (\straight_0) \\
\when \; \guard\; \doo& 0={L^2}{-}(x^2{+}y^2)   & (\straight_1) \\
\prog{and}& 0=\tension+s   & (\straight_2) \\
\when \;\prog{not}\; \guard\; \doo& 0=\tension   & (\straight_3) \\
\prog{and}& 0=({L^2}{-}(x^2{+}y^2))-s   & (\straight_4) \\
\eea\right.
\eeqq
The dynamics in modes $\guard=\fff$ and $\guard=\ttt$ are not modified. However, in mode $\guard=\fff$, we have 
\[
s' = ({L^2}{-}(x^2{+}y^2))' = 2(xx'+yy')\,.
\] 
Consequently, the candidate fact is now $xx'{+}yy'=0$ (compared to ${L^2}{-}(x^2{+}y^2)=0$ for the cup-and-ball). 
We show next the mode change array $A_2$:
 \beq
\small\left\{\negesp\bea{cclcll}
 0&\negesp=\negesp& \pmatch{\ddot{x}}+{\tension}x 
& (\eqq_1) & \guard{=}\ttt,\nstime
\\
 0&\negesp=\negesp& \ddot{y}+\pmatch{\tension}y+g 
& (\eqq_2)  & \dots\dots
\\
\remph{0} &\!\!\!\!\remph{=}\!\!\!\!& \remph{{L^2}{-}(x^2{+}y^2)} 
& \remph{(\straight_1)}  & \dots\dots 
\\
\gremph{0} &\!\!\!\!\gremph{=}\!\!\!\!& \gremph{x\dot{x}{+}y\dot{y}} 
& \gremph{(\dot{\straight_1})}  & \dots\dots 
\\
\remph{0} &\!\!\!\!\remph{=}\!\!\!\!&  \remph{x\ddot{x}{+}{\dot{x}}^2{+}{\dot{y}}^2{+}y\ddot{y}}
& \remph{(\ddot{\straight_1})}  & \dots\dots 
\\ [2mm]
 0&\negesp=\negesp& \postset{(\pmatch{\ddot{x}}+{\tension}x)}
& (\postset{\eqq_1}) & \guard{=}\ttt,\postset{\nstime}
\\
 0&\negesp=\negesp& \postset{(\ddot{y}+\pmatch{\tension}y)}+g 
& (\postset{\eqq_2})  & \dots\dots
\\
\remph{0} &\!\!\!\!\remph{=}\!\!\!\!& \remph{\postset{({L^2}{-}(x^2{+}y^2))} }
& \remph{(\postset{\straight_1})}  & \dots\dots 
\\
\remph{0} &\!\!\!\!\remph{=}\!\!\!\!& \remph{\postset{{(x\dot{x}{+}y\dot{y})}}}
& \remph{(\postset{\dot{\straight_1}})}  & \dots\dots 
\\
\remph{0} &\!\!\!\!\remph{=}\!\!\!\!& \remph{\postset{(x\ddot{x}{+}{\dot{x}}^2{+}{\dot{y}}^2{+}y\ddot{y})}}
& \remph{(\postset{{\ddot{\straight_1}}})}  & \dots\dots 
\\ [2mm]
{0} &\!\!\!\!{=}\!\!\!\!& {\ppostset{2}{({L^2}{-}(x^2{+}\pmatch{y^2}))} }
& {(\ppostset{2}{\straight_1})}  & \guard{=}\ttt,\ppostset{2}{\nstime} 
\\
{0} &\!\!\!\!{=}\!\!\!\!& {\ppostset{2}{(x\dot{x}{+}y\pmatch{\dot{y}})}}
& {(\ppostset{2}{\dot{\straight_1}})}  & \dots\dots 
\eea\right.
\label{hjyfgdsgrfbdmnh}
\eeq
The fact is shown in {\color{green}green} and the conflicts are shown in {\color{red}red}. The dependent variables are 
$$\tension,\pmatch{\ppostset{2}{x}}{\sim}\pprimeppostset{}{}{x}{\sim}\ddot{x},
\pmatch{\ppostset{2}{y}}{\sim}\pprimeppostset{}{}{y}{\sim}\ddot{y}\,;\,\postset{\tension},
\pmatch{\pprimeppostset{}{2}{x}}{\sim}\postset{\ddot{x}},\pmatch{\pprimeppostset{}{2}{y}}{\sim}\postset{\ddot{y}}
$$
The black subsystems of (\ref{hjyfgdsgrfbdmnh}) and (\ref{ksdjhvgdfkisdyujhg}) (cup-and-ball with exogenous mode change) coincide. Consequently, the analyses of the two examples coincide as well.


\section{Nonlinear systems: a clutch}
\label{hgrdfhreagfdsuj}
\begin{figure}[ht]
    \centerline{\includegraphics[width=0.75\textwidth]{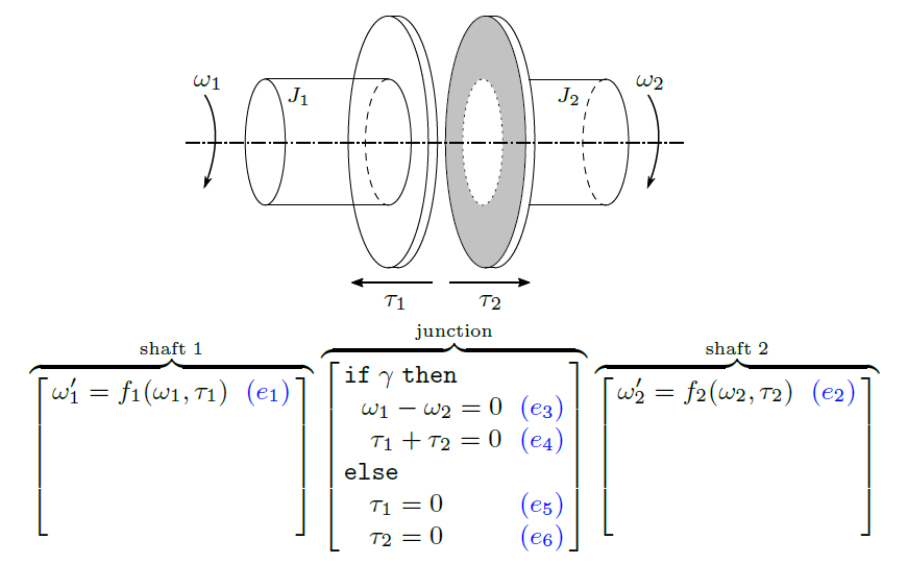}}
  \caption{A simple clutch with two shafts and its model from first principles. }
 \label{figclutch}
\end{figure} 
The clutch and its model from first principles are shown on Fig.\,\ref{figclutch}.
We separately model each schaft: rotation velocities are $\omega_i,i{=}1,2$ and torques are $\tau_i,i{=}1,2$. In the model of their junction, the Boolean $\guard$ is an  input. The junction exhibits two modes: \emph{engaged} $\guard{\,=\,}\ttt$, and \emph{released} $\guard{\,=\,}\fff$. We  focus on the mode change from released to engaged $\guard{\,:\,}\fff{\,\ra\,}\ttt$. The physics suggests the following: 1) The mode change $\guard{\,:\,}\fff{\,\ra\,}\ttt$ should be determined; hot restart differs form cold start, which possesses one degree of freedom (the common rotation velocity $\omega_1{\,=\,}\omega_2$ is not determined);  2) The torques will experience an impulsion. Physically meaningful modeling leads to  take $f_1,f_2$ linear with respect to torques. We will explore the consequences of chosing $f_1,f_2$ nonlinear.
Here follows the mode change array $A_1$ for this example:
\beq\small
A_1:\left\{\negesp\bea{lll}
\pmatch{\dot{\omega_1}}=f_1(\omega_1,\tau_1) & (e_1)
&\guard{=}\ttt, t
\\ [0.5mm]
\dot{\omega_2}=f_2(\omega_2,\pmatch{\tau_2}) & (e_2)
\\ [0.5mm]
\remph{\omega_1=\omega_2} & \remph{(e_3)}
&\remph{\mbox{disabled}}
\\ [0.5mm]
\remph{\dot{\omega_1}=\dot{\omega_2}} & \remph{(\dot{e_3})}
&\remph{\mbox{disabled}}
\\ [0.5mm]
\pmatch{\tau_1}+\tau_2=0 & (e_4)
\\ [2mm]
\dot{\omega_1}=\vsmall^{-1}(\pmatch{\postset{\omega_1}}-{\omega_1}) \hspace*{1mm} & (\euler{1}{})
\\ [0.5mm]
\pmatch{\dot{\omega_2}}=\vsmall^{-1}(\postset{\omega_2}-{\omega_2}) & (\euler{2}{})
\\ [2mm]
\pmatch{\postset{\dot{\omega_1}}}=f_1(\postset{\omega_1},\postset{\tau_1}) & (\postset{e_1})
&\guard{=}\ttt, t{+}\vsmall
\\ [0.5mm]
\postset{\dot{\omega_2}}=f_2(\postset{\omega_2},\pmatch{\postset{\tau_2}}) & (\postset{e_2})
\\ [0.5mm]
{\postset{\omega_1}=\pmatch{\postset{\omega_2}}}  & (\postset{e_3})
\\ [0.5mm]
{\postset{\dot{\omega_1}}=\pmatch{\postset{\dot{\omega_2}}}}  & (\postset{\dot{e_3}})
\\ [0.5mm]
\pmatch{\postset{\tau_1}}+\postset{\tau_2}=0  & (\postset{e_4})
\eea\right.
\label{jkcdysgcfxjatshdtfg}
\eeq
Mode change array $A_1$ showing the dynamics around the mode change $\guard{\,:\,}\fff{\,\ra\,}\ttt$ occurring at instant $t$, by stacking the discrete time dynamics at instants $t$, and $t{+}\vsmall$. Disabled equations are highlighted in {\color{red}red}. The enabled subsystem (in black) is structurally nonsingular, as evidenced by the one-to-one matching between variables and equations highlighted in {\color{blue}blue}.
%
We have $\Vars^-{=}\{\omega_1,\omega_2\}$. 
Since we focus on state variables for restart, we are not interested in leading equations of instant $\nstime{+}1=\postset{\nstime}$. By omitting equations of this instant, the enabled part of the mode change array $A_1$ in (\ref{jkcdysgcfxjatshdtfg}) is:
	 \[\left\{\negesp\bea{ll}
J_1\pmatch{\dot{\omega_1}}=f_1(\tau_1,\omega_1) & (f_1)
\\
J_2\dot{\omega_2}=f_2(\pmatch{\tau_2},\omega_2) & (f_2)
\\
\pmatch{\tau_1}+\tau_2=0 & (f_3)
\\ [2mm]
\postset{\omega_1}=\pmatch{\postset{\omega_2}} & (f_4)
\\ [2mm]
\dot{\omega_1}=\vsmall^{-1}{\times}(\pmatch{\postset{\omega_1}}-{\omega_1}) & (f_5)
\\
\pmatch{\dot{\omega_2}}=\vsmall^{-1}{\times}(\postset{\omega_2}-{\omega_2}) & (f_6)
\eea\right.\]
So far, no particular form was assumed for $f_i,i{=}1,2$. 

\smallskip
\subsubsection*{Rescaling calculus with $f_i$ linear in $\tau_i$}
Assume $f_i(\tau_i,\omega_i)=\tau_i-g_i(\omega_i)$ is linear with respect to the torque. The resulting rescaling calculus is:
\beq\left\{\negesp\bea{ll}
\rescaling{{\pmatch{\dot{\omega_1}}}}=\max(\rescaling{{\pmatch{\dot{\omega_1}}}},\rescaling{{\tau_1}})
\\
\rescaling{{\pmatch{\tau_2}}}=\max(\rescaling{{\dot{\omega_2}}},\rescaling{{\pmatch{\tau_2}}})
\\
\rescaling{{\pmatch{\tau_1}}}=\max(\rescaling{{\pmatch{\tau_1}}},\rescaling{{\tau_2}})
\\ [2mm]
\rescaling{{\pmatch{\postset{\omega_2}}}}=\max(\rescaling{{\pmatch{\postset{\omega_2}}}},\rescaling{{\postset{\omega_1}}})
\\ [2mm]
\rescaling{{\pmatch{\postset{\omega_1}}}}+1=\max(\rescaling{{\pmatch{\postset{\omega_1}}}}+1,\rescaling{{\dot{\omega_1}}})
\\
\rescaling{{\pmatch{\dot{\omega_2}}}}=\max(\rescaling{{\pmatch{\dot{\omega_2}}}},\rescaling{{\postset{\omega_2}}}+1)
\eea\right.
\label{hgfdjyrgsgfd}
\eeq
The following solution of this rescaling analysis yields a good solution for Problem~\ref{likgusefhdrlpouhiouh}:
\beq\bea{rcl}
1&\negesp:\negesp&\dot{\omega_1},\dot{\omega_2},\tau_1,\tau_2\,;\,f_1,f_2,f_3,f_5,f_6 \\
0&\negesp:\negesp&\postset{\omega_1},\postset{\omega_2}\,;\,f_4 \\
\eea
\label{hgrsdhgsrdhgfdh}
\eeq
Expanding $\rescaled{\dot{\omega_i}}={\postset{\omega_i}}-{\omega_i},i=1,2$
The rescaled system is
\[\left\{\negesp\bea{ll}
J_1(\pmatch{\postset{\omega_1}}-\omega_1)=\rescaled{\tau_1}
\\
J_2(\postset{\omega_2}-\omega_2)=\pmatch{\rescaled{\tau_2}}
\\
\pmatch{\rescaled{\tau_1}}+\rescaled{\tau_2}=0 
\\ 
\postset{\omega_1}=\pmatch{\postset{\omega_2}} 
\eea\right.\]
which yields, by setting $\pplus{\omega}\eqdef\postset{\omega_1}={\postset{\omega_2}}$, the restart system:
\beqq
(J_1+J_2)\pplus{\omega}&=&J_1\mmoins{\omega_1}+J_2\mmoins{\omega_2}\,,
\eeqq
reflecting the preservation of angular momentum. This again illustrates Theorem~\ref{sjdufghafejahyfg}. Note that having $g_i$ nonlinear was not an obstacle, since the rotation velocities are not impulsive at the mode change (they are past variables).

\smallskip
\subsubsection*{Rescaling calculus with $f_i$ nonlinear in $\tau_i$}
Assume now that $f_i$ is nonlinear in its arguments. Rescaling calculus (\ref{hgfdjyrgsgfd}) modifies as follows:
\beq\left\{\negesp\bea{ll}
\rescaling{{\pmatch{\dot{\omega_1}}}}=\rescaling{f_1}=\mathbf{if}\;\max(\rescaling{\tau_1},\rescaling{\omega_1})=0\;\mathbf{then}\;0\;\mathbf{else}\;\infty
\\
\rescaling{\pmatch{\tau_2}}=\rescaling{f_2}=\mathbf{if}\;\max(\rescaling{\tau_2},\rescaling{\omega_2})=0\;\mathbf{then}\;0\;\mathbf{else}\;\infty
\\
\rescaling{{\pmatch{\tau_1}}}=\max(\rescaling{{\pmatch{\tau_1}}},\rescaling{{\tau_2}})
\\ [2mm]
\rescaling{{\pmatch{\postset{\omega_2}}}}=\max(\rescaling{{\pmatch{\postset{\omega_2}}}},\rescaling{{\postset{\omega_1}}})
\\ [2mm]
\rescaling{{\pmatch{\postset{\omega_1}}}}+1=\max(\rescaling{{\pmatch{\postset{\omega_1}}}}+1,\rescaling{{\dot{\omega_1}}})
\\
\rescaling{{\pmatch{\dot{\omega_2}}}}=\max(\rescaling{{\pmatch{\dot{\omega_2}}}},\rescaling{{\postset{\omega_2}}}+1)
\eea\right.
\label{jsdhycgdfscjsdhgf}
\eeq
The last equation of this rescaling calculus prevents the existence of a good solution to Problem~\ref{likgusefhdrlpouhiouh} since it prohibits $\rescaling{{\pmatch{\dot{\omega_2}}}}=0$. Our approach does not solve the restart problem in this case.

\smallskip
\subsubsection*{What if the junction is nonlinear in velocities} We keep $f_i=\tau_i-g(\omega_i),i{=}1,2$ but we change the junction model by changing $(e_3)$ to $\omega_2=h(\omega_1)$, where $h$ is nonlinear. The rescaling calculus is now
\beq\left\{\negesp\bea{ll}
\rescaling{{\pmatch{\dot{\omega_1}}}}=\max(\rescaling{{\pmatch{\dot{\omega_1}}}},\rescaling{{\tau_1}})
\\
\rescaling{{\pmatch{\tau_2}}}=\max(\rescaling{{\dot{\omega_2}}},\rescaling{{\pmatch{\tau_2}}})
\\
\rescaling{{\pmatch{\tau_1}}}=\max(\rescaling{{\pmatch{\tau_1}}},\rescaling{{\tau_2}})
\\ [2mm]
\rescaling{{\pmatch{\postset{\omega_2}}}}=\max\left(\rescaling{{\pmatch{\postset{\omega_2}}}},\mathbf{if}\;\rescaling{{\postset{\omega_1}}}=0\;\mathbf{then}\;0\;\mathbf{else}\;\infty\right)
\\ [2mm]
\rescaling{{\pmatch{\postset{\omega_1}}}}+1=\max(\rescaling{{\pmatch{\postset{\omega_1}}}}+1,\rescaling{{\dot{\omega_1}}})
\\
\rescaling{{\pmatch{\dot{\omega_2}}}}=\max(\rescaling{{\pmatch{\dot{\omega_2}}}},\rescaling{{\postset{\omega_2}}}+1)
\eea\right.
\label{hjtfhgfkjmhg}
\eeq
for which (\ref{hgrsdhgsrdhgfdh}) is still a good solution. Making the junction model non-polynomial added no difficulty with reference to the original clutch.

\section{{Comparing with Trenn et al.}}
\label{sldkvjshdblvckujsh}
\subsection{A linear DAE system}
Consider the following linear model
\[\left\{\bea{ll}
0=\dot{v_1}-\pmatch{i} & (f_1)\\
0=\dot{\pmatch{v_2}}-i+v_R  & (f_2)\\
0=\dot{i}-\pmatch{v_R}  & (f_3)\\
0=\pmatch{v_1}+u & (f_4)
\eea\right.\]
The dependent variables of this DAE are $i,v_1,v_2,v_R$ and $u$ is an input with $u,\dot{u},\ddot{u}$ known at $t=0^-$. There is a unique perfect matching (shown in \pmatch{blue}) and its total weight is $1$.  The equations defining the offsets are (with $c_i\eqdef c_{f_i}$)
\[\bea{rcl}
d_i &=& c_1 \\
d_{v_2} &=& 1+c_2 \\
d_{v_R} &=& c_3 \\
d_{v_1}&=& c_4 \\
\eea \hspace*{1cm} \bea{rcl}
c_4-c_1 &\geq& 1 \\
c_1-c_2 &\geq& 0\\
c_3-c_2 &\geq& 0 \\
c_1-c_3 &\geq& 1 
\eea\]
The minimal solution is 
\[\bea{c}
c_1=1,c_2=c_3=0, c_4=2 \\
d_i=1,d_{v_1}=2,d_{v_2}=1,d_{v_R}=0
\eea\]
The index reduced system is
\[\left\{\bea{ll}
0=\dot{v_1}-{i} & (f_1)\\
0=\ddot{v_1}-\pmatch{\dot{i}} & (\dot{f_1})\\
0={\pmatch{\dot{v_2}}}-i+v_R  & (f_2)\\
0=\dot{i}-\pmatch{v_R}  & (f_3)\\
0={v_1}+u & (f_4)\\
0={\dot{v_1}}+\dot{u} & (\dot{f_4})\\
0=\pmatch{\ddot{v_1}}+\ddot{u} & (\ddot{f_4})
\eea\right.\]
and we show in \pmatch{blue} a perfect matching for the leading variables $\dot{i},\dot{v_2},v_R,\ddot{v_1}$. The state variables are $i,v_2,v_1,\dot{v_1},u,\dot{u}$, subject to the consistency constraints $(f_1,f_4)$, and $\ddot{u}$ is a free (input) variable.

The mode change array $A_1$ completed with Euler identities, in which we omit the leading equations of the last block, together with the resulting rescaling calculus, is  the following:
\[\small\left\{\negesp\bea{lll}
0=\dot{v_1}-\pmatch{i} & \rescaling{i}\geq\rescaling{\dot{v_1}} & (f_1)\\
\remph{0=\ddot{v_1}-{\dot{i}}} &  \remph{\mbox{conflict}}\\
0=\dot{i}-\pmatch{v_R} & \rescaling{v_R}\geq\rescaling{\dot{i}}  & (f_3)\\
0=\pmatch{v_1}+\bemph{u} & \rescaling{v_1}\geq \bemph{0}  & (f_4)\\
\remph{0={\dot{v_1}}+{\dot{u}}} &  \remph{\mbox{conflict}}
\\
\remph{0={\ddot{v_1}}+\ddot{u}} &  \remph{\mbox{conflict}}
 \\ [2mm]
0=\postset{\dot{v_1}}-\pmatch{\postset{i}} & \rescaling{\postset{i}}\geq\rescaling{\postset{\dot{v_1}}}  & (f_5)\\
0=\pmatch{\postset{v_1}}+\bemph{\postset{u}}  & \rescaling{\postset{v_1}}\geq \bemph{0}  & (f_6)\\
0=\pmatch{\postset{\dot{v_1}}}+{\postset{\dot{u}}} & \rescaling{\postset{\dot{v_1}}}\geq \rescaling{\postset{\dot{u}}}  & (f_7)
\\ [2mm]
0=\vsmall{\times}\pmatch{\dot{i}}-(\postset{i}-i) & \rescaling{\dot{i}} \geq 1+\max(\rescaling{\postset{i}},\rescaling{{i}})  & (f_8)\\
0=\vsmall{\times}\pmatch{\dot{v_1}}-(\postset{v_1}-v_1) & \rescaling{\dot{v_1}} \geq 1+\max(\rescaling{\postset{v_1}},\rescaling{{v_1}})  & (f_9)\\
0=\vsmall{\times}\pmatch{\ddot{v_1}}-(\postset{\dot{v_1}}-\dot{v_1}) & \rescaling{\ddot{v_1}} \geq 1+\max(\rescaling{\postset{\dot{v_1}}},\rescaling{\dot{v_1}})  & (f_{10}) \\
0=\vsmall{\times}\pmatch{\ddot{u}}-(\postset{\dot{u}}-\dot{u}) & \rescaling{\ddot{u}} \geq 1+\max(\rescaling{\postset{\dot{u}}},\rescaling{\dot{u}})  & (f_{12})
\eea\right.\]
The values of $\preset{u},\preset{\dot{u}},\preset{\ddot{u}}$ are given. Hence, the set of past variables is
\[
\Vars^- = \{
\bemph{u}{=}\preset{u}{+}\vsmall\,{\preset{\dot{u}}},\bemph{\dot{u}}{=}\preset{\dot{u}}{+}\vsmall\,{\preset{\ddot{u}}},\bemph{\postset{u}}{=}
\preset{u}{+}2\vsmall\,{\preset{\dot{u}}}{+}\vsmall^2\,{\preset{\ddot{u}}})
\}
\]
The solution of the rescaling calculus is the following, where we list the variables and equations with rescaling offsets $2,1,0$:
\[\bea{rcl}
2&\negesp:\negesp&v_R,\dot{i},\remph{{{\ddot{v_1}}}}\,;\,{f_{10}},{f_{8}},{f_{3}} \\
1&\negesp:\negesp& {\ddot{u}},{{i}},{{\dot{v_1}}}\,;\,{f_{9}},{f_{1}},{f_{12}}\\
0&\negesp:\negesp& \mbox{other variables and equations}
\eea\]
Goodness condition (\ref{klifgyegrujytfgeikuy}) is violated by variable $\ddot{v_1}$, highlighted in \remph{red}. So this mode change is nondetermined according to our approach. Equation $f_{10}$ is matched with $\ddot{v_1}$. Hence, we remove it, and perform the Dulmage-Mendelsohn decomposition of the so reduced mode change array $A_1$. The following variables are reachable from $\ddot{v_1}$ by following an alternating path: $v_1,\dot{v_1},i,\postset{v_1},\dot{i},\postset{i},\postset{\dot{v_1}},\postset{\dot{u}}$; no variable is determined.

Our method fails finding a solution to hot restart, unlike Trenn's method. The reason is that we do not allow for changes in the state basis, which is precisely performed by the Weierstrass decomposition~\cite{TrennPhD2009,Tren09b,LibeTren12}.
\smallskip

\subsection{Further investigating the clutch} 
Being structural, our approach does not recombine states. We illustrate some consequences of this on a linear version of the clutch, namely: $f_i(\omega_i,\tau_i)=\tau_i-a_i\omega_i$, where $a_i$ are constants, already studied in Appendix~\ref{hgrdfhreagfdsuj}. 

Now, reconsider the same example, but with the change of variable $(\omega_i,\tau_i)\mapsto(w_i,\sigma_i)$, where:
\[
w_1=\omega_1{-}\tau_1, w_2=\omega_2{-}\tau_2, \sigma_1=\omega_1{+}\tau_1, \sigma_2=\omega_2{+}\tau_2\,.
\]
or, equivalently, 
\[
2\omega_1=w_1{+}\sigma_1, 2\omega_2=w_2{+}\sigma_2, 2\tau_1=\sigma_1{-}w_1, 2\tau_2=\sigma_2{-}w_2\,.
\]
The dynamics in the $\fff$ mode rewrites
\[
\left\{\bea{l}
0=\dot{w_1}{+}\dot{\sigma_1}-(\sigma_1{-}w_1{-}a(\sigma_1{+}w_1)) \\
0=\dot{w_2}{+}\dot{\sigma_2}-(\sigma_2{-}w_2{-}a(\sigma_2{+}w_2)) \\
0={w_1}{-}{\sigma_1} \\
0={w_2}{-}{\sigma_2} 
\eea\right.
\]
It is now a DAE, structurally! The index reduction yields
\[
\left\{\bea{l}
0=\dot{w_1}{+}\dot{\sigma_1}-(\sigma_1{-}w_1{-}a(\sigma_1{+}w_1)) \\
0=\dot{w_2}{+}\dot{\sigma_2}-(\sigma_2{-}w_2{-}a(\sigma_2{+}w_2)) \\
0={w_1}{-}{\sigma_1} \\
0={w_2}{-}{\sigma_2} \\
0=\dot{w_1}{-}\dot{\sigma_1} \\
0=\dot{w_2}{-}\dot{\sigma_2} 
\eea\right.
\]
Hence, for the mode change, $\Vars^-=\{w_1,w_2.\sigma_1,\sigma_2\}$.
The mode change array for $\guard:\fff\ra\ttt$ rewrites
\[\left\{\bea{l}
0=\dot{w_1}{+}\dot{\sigma_1}-(\sigma_1{-}w_1{-}a(\sigma_1{+}w_1)) \\
0=\dot{w_2}{+}\dot{\sigma_2}-(\sigma_2{-}w_2{-}a(\sigma_2{+}w_2)) \\
\remph{0=({w_1}{+}{\sigma_1})-({w_2}{+}{\sigma_2})} \\
\remph{0= (\sigma_1{+}\sigma_2)-(w_1{+}w_2)}\\
\remph{0=(\dot{w_1}{+}\dot{\sigma_1})-(\dot{w_2}{+}\dot{\sigma_2})} \\
\remph{0= (\dot{\sigma_1}{+}\dot{\sigma_2})-(\dot{w_1}{+}\dot{w_2})}
\\ [2mm]
0=\postset{\dot{w_1}}{+}\postset{\dot{\sigma_1}}-(\postset{\sigma_1}{-}\postset{w_1}{-}a(\postset{\sigma_1}{+}\postset{w_1})) \\
0=\postset{\dot{w_2}}{+}\postset{\dot{\sigma_2}}-(\postset{\sigma_2}{-}\postset{w_2}{-}a(\postset{\sigma_2}{+}\postset{w_2})) \\
0=(\postset{w_1}{+}\postset{\sigma_1})-(\postset{w_2}{+}\postset{\sigma_2}) \\
0= (\postset{\sigma_1}{+}\postset{\sigma_2})-(\postset{w_1}{+}\postset{w_2})\\
0=(\postset{\dot{w_1}}{+}\postset{\dot{\sigma_1}})-(\postset{\dot{w_2}}{+}\postset{\dot{\sigma_2}}) \\
0= (\postset{\dot{\sigma_1}}{+}\postset{\dot{\sigma_2}})-(\postset{\dot{w_1}}{+}\postset{\dot{w_2}})
\eea\right.\]
Keeping only the active equations, removing the leading equations of the last block, and adding the Euler identities yields
\[\left\{\bea{l}
0=\pmatch{\dot{w_1}}{+}\dot{\sigma_1}-(\sigma_1{-}w_1{-}a(\sigma_1{+}w_1)) \\
0=\pmatch{\dot{w_2}}{+}\dot{\sigma_2}-(\sigma_2{-}w_2{-}a(\sigma_2{+}w_2)) 
\\ [1mm]
0=({\postset{w_1}}{+}\pmatch{\postset{\sigma_1}})-({\postset{w_2}}{+}{\postset{\sigma_2}}) \\
0= ({\postset{\sigma_1}}{+}\pmatch{\postset{\sigma_2}})-(\postset{w_1}{+}\postset{w_2})
\\ [1mm]
0=\vsmall{\times}\dot{w_1}-(\pmatch{\postset{w_1}}{-}w_1) \\
0=\vsmall{\times}\dot{w_2}-(\pmatch{\postset{w_2}}{-}w_2) \\
0=\vsmall{\times}\pmatch{\dot{\sigma_1}}-(\postset{\sigma_1}{-}\sigma_1) \\
0=\vsmall{\times}\pmatch{\dot{\sigma_2}}-(\postset{\sigma_2}{-}\sigma_2) 
\eea\right.\]
The rescaling calculus is
\[\left\{\bea{rcl}
\rescaling{\dot{w_1}} = \rescaling{f_1} &\geq& \rescaling{\dot{\sigma_1}} \\
\rescaling{\dot{w_2}} = \rescaling{f_2} &\geq& \rescaling{\dot{\sigma_2}} \\
\rescaling{\postset{\sigma_1}} = \rescaling{f_3} &\geq& \max(\rescaling{\postset{w_1}},\rescaling{\postset{w_2}},\rescaling{\postset{\sigma_2}}) \\
\rescaling{\postset{\sigma_2}} = \rescaling{f_4} &\geq& \max(\rescaling{\postset{w_1}},\rescaling{\postset{w_2}},\rescaling{\postset{\sigma_1}}) \\
1{+}\rescaling{\postset{w_1}}=\rescaling{f_5} &\geq& \rescaling{\dot{w_1}}\\
1{+}\rescaling{\postset{w_2}}=\rescaling{f_6} &\geq& \rescaling{\dot{w_2}} \\
\rescaling{\dot{\sigma_1}}=\rescaling{f_7} &\geq& 1{+}\rescaling{\postset{\sigma_1}}\\
\rescaling{\dot{\sigma_2}}=\rescaling{f_8} &\geq& 1{+}\rescaling{\postset{\sigma_2}}
\eea\right.\]
Its solution is the following, where we list the variables and equations with rescaling offsets $2,1,0$:
\[\bea{rcl}
1&\negesp:\negesp& {\dot{\sigma_1}},{\dot{\sigma_2}},{\dot{w_1}},{\dot{w_2}}\,;\,{f_1},{f_2},{f_7},{f_8}\\
0&\negesp:\negesp& \mbox{other variables and equations}
\eea\]
This solution satisfies the Conditions (\ref{jsdyhgcdfasjyht}--\ref{klifgyegrujytfgeikuy}). Hence, we can rescale the array, thus obtaining:
\[\left\{\bea{l}
0=(\pmatch{\postset{w_1}}{-}w_1){+}(\postset{\sigma_1}{-}\sigma_1) \\
0=(\pmatch{\postset{w_2}}{-}w_2){+}(\postset{\sigma_2}{-}\sigma_2) 
\\ [1mm]
0=({\postset{w_1}}{+}\pmatch{\postset{\sigma_1}})-({\postset{w_2}}{+}{\postset{\sigma_2}}) \\
0= ({\postset{\sigma_1}}{+}\pmatch{\postset{\sigma_2}})-(\postset{w_1}{+}\postset{w_2})
\eea\right.\]
At this point we have a problem: the above system is \emph{structurally} nonsingular; nevertheless, it is \emph{numerically} singular. The first equation yields $\postset{w_1}=w_1+\sigma_1-\postset{\sigma_1}$ and replacing $\postset{w_1}$ by $w_1+\sigma_1-\postset{\sigma_1}$ in the third equation cancels $\postset{\sigma_1}$, which cannot be a pivot.

Maybe the change of variables was responsible for this; we modify its coefficients:
\[
\omega_1=w_1{+}2\sigma_1, \omega_2=w_2{+}\sigma_2, \tau_1=\sigma_1{-}3w_1, \tau_2=2\sigma_2{-}w_2\,.
\]
The structural analysis is not modified. So we end up with the rescaled array
\[\left\{\bea{l}
0=(\pmatch{\postset{w_1}}{-}w_1){+}2(\postset{\sigma_1}{-}\sigma_1) \\
0=(\pmatch{\postset{w_2}}{-}w_2){+}(\postset{\sigma_2}{-}\sigma_2) 
\\ [1mm]
0=({\postset{w_1}}{+}2\pmatch{\postset{\sigma_1}})-({\postset{w_2}}{+}{\postset{\sigma_2}}) \\
0= ({\postset{\sigma_1}}{+}2\pmatch{\postset{\sigma_2}})-(3\postset{w_1}{+}\postset{w_2})
\eea\right.\]
which suffers from the same problem.

The reason for this pitfall is that our change of state basis resulted in combining non-impulsive variables with impulsive ones. This would not be a trouble to the method developed by Trenn et al.~\cite{TrennPhD2009,Tren09b,LibeTren12} for linear systems, since the Weierstrass decomposition that is applied first, recovers the basis in which clean separation is made between non-impulsive and impulsive variables. We do not do this.

\end{document}